\documentclass[letterpaper,11pt]{article}
\usepackage{fullpage}
\usepackage[ansinew]{inputenc}
\usepackage{verbatim}
\usepackage{graphicx}
\usepackage{tabularx}
\usepackage{array}
\usepackage{delarray}
\usepackage{dcolumn}
\usepackage{float}
\usepackage{amsmath}
\usepackage{psfrag}
\usepackage{booktabs}
\usepackage{multirow,hhline}
\usepackage{amstext}
\usepackage{pstricks}

\usepackage{pst-plot}
\usepackage{amssymb}
\usepackage{fancyhdr}
\usepackage{hyperref}
\usepackage{setspace}
\usepackage{amsthm}
\usepackage{dsfont}
\usepackage{adjustbox}
\RequirePackage[round,authoryear,longnamesfirst]{natbib}
\usepackage{makeidx}
\usepackage{tikz}
\usetikzlibrary{arrows,calc}
\usepackage{mdframed}
\usepackage{lipsum} 
\usepackage{chapterbib}
\usepackage{url}
\usepackage{lscape}
\hypersetup{citecolor=true,colorlinks=true,allcolors=blue}
\makeindex
\numberwithin{equation}{section}
\newtheorem{lem}{Lemma}[section]
\newtheorem{thm}{Theorem}[section]
\newtheorem{cor}{Corollary}[section]

\newtheorem{ass}{Assumption}
\newtheorem{prop}{Proposition}[section]

\renewcommand{\citep}[1]{\citeauthor{#1}, \citeyear{#1}}

\newcommand{\indep}{\perp\!\!\!\perp}

\newcommand{\convP}{\stackrel{p}{\longrightarrow}}
\newcommand{\convD}{\rightsquigarrow}

\newcommand{\eps}{\varepsilon}
\renewcommand{\epsilon}{\varepsilon}
\DeclareMathOperator*{\argmax}{arg\,max}
\DeclareMathOperator*{\argmin}{arg\,min}
\makeatletter
\newcommand*{\rom}[1]{\expandafter\@slowromancap\romannumeral #1@}
\makeatother

\title{Quasi-Bayesian Inference for Production Frontiers\thanks{This paper was previously circulated under the title ``Simulation-based Estimation and Inference of Production Frontiers". First draft: September, 2017. We thank the editor, Jianqing Fan, an associate editor and a referee for their helpful and insightful comments. We are grateful to Dennis Kristensen, Ulrich M\"{u}ller, Andriy Norets, Peter Phillips, and Valentin Zelenyuk as well as participants of 2017 HU-HUE-SMU Tripartite Conference on Econometrics and 2018 China Meeting of the Econometrics Society for valuable suggestions.  Liu acknowledges the financial support from the National Natural Science Fundation of China (No.72003171). Zhang acknowledges
		the financial support from Singapore Ministry of Education Tier 2 grant under grant MOE2018-T2-2-169 and the Lee Kong Chian fellowship.}\\ \vspace{2mm}
}
\author{ Xiaobin Liu\thanks{School of Economics, Academy of Financial Research, and Institute for Fiscal Big-Data \& Policy of Zhejiang University.\ E-mail~address: liuxiaobin@zju.edu.cn.}  \and Thomas Tao Yang\thanks{%
		Australian National University.\ E-mail~address: tao.yang@anu.edu.au.} \and Yichong Zhang\thanks{%
		The corresponding author. Singapore Management University.\ E-mail~address: yczhang@smu.edu.sg.} }

\begin{document}
	\maketitle
	\begin{abstract}
		This paper proposes to estimate and infer the production frontier by combining multiple first-stage extreme quantile estimates via the quasi-Bayesian method. We show the asymptotic properties of the proposed estimator and the validity of the inference procedure. The finite sample performance of our method is illustrated through simulations and an empirical application.
		

	\end{abstract}
	\newpage
	\section{Introduction}
	\label{sec:intro}
	The concept of production frontier (or data envelope) arises naturally in and applies to many fields such as manufacturing, health care, transportation, education, banking, public services, and portfolio management. \cite{G04} provide a comprehensive survey on the topic. However, the estimation and inference of the production frontier are complicated by the fact that the parameter of interest is on the boundary.

	In this article, we combine multiple extreme quantile estimates and construct a point estimate and confidence interval for the production frontier via the quasi-Bayesian method. We treat the first-stage extreme quantile estimates and their joint asymptotic distribution as observations and the corresponding likelihood, respectively. Then, we put a prior on the production frontier, draw from the posterior distribution by Markov Chain Monte Carlo (MCMC) method, and construct the point estimator and confidence interval.

	The quasi-Bayesian inference is first considered by \cite{BY69} and \cite{IH13}.  Recently, \cite{CH03}, \cite{M13}, \cite{JPW15}, \cite{Y15}, \cite{FN15}, and \cite{CCOT16} apply the method in the context of M-estimations, misspecified MLE, Maximum-score type estimations, threshold regressions, GMM, and partially identified models, respectively. \cite{CGHK15} justify the use of kernel regression instead of the MCMC method to make inference in the GMM framework. We differ from the previous literature by applying the method to the \textit{first-stage estimates} rather than the original observations. We treat a finite number of (properly scaled) first-stage estimates as new observations and conduct quasi-Bayesian estimation and inference. We mainly treat our method as an estimator-combination device. First, it is robust to certain amount of outliers as it combines extreme quantiles, rather than using the sample maximum of feasible outputs. Second, it can simultaneously produce point estimates and confidence intervals. Since extreme quantile estimators are not asymptotically normal, the standard bootstrap inference does not control size. The quasi-Bayesian approach provides an asymptotically valid alternative. Third, our method can automatically correct the downward bias between the extreme quantiles and the production frontier. It has good finite-sample performance even in samples with small and moderate sizes, as illustrated in our simulation study.

	There is a vast literature on the estimation and inference of production frontiers. \cite{Deprins84} first introduce the free-disposal hull (FDH) estimator. Its asymptotic properties have been studied by \cite{PSW00}, \cite{DFS10}, and \cite{DSW17}. Assuming convexity of the production frontier, \cite{KPS98} consider the data envelopment analysis (DEA) estimator. The asymptotic properties of DEA estimator have been investigated by \cite{KPS98}, \cite{G99},  \cite{J04}, \cite{JP06}, \cite{KSW08}, \cite{PJS10}, and \cite{KSW15}. However, neither the FDH nor DEA estimator is robust to any outliers. In addition, the inference of the FDH estimator requires estimating the normalizing rate, while a valid inference for the DEA estimator is still lacking, to the best of our knowledge. Recognizing those drawbacks, \cite{CFS02} and \cite{ADT05} suggest estimating an expected frontier, which does not envelope the data. \cite{DFS10}, \cite{DFS12}, and \cite{DGG14} propose to first estimate intermediate quantiles, and then extrapolate them to the boundary. Recently, \cite{JMR14} consider nonparametric estimations of data boundary by adaptive kernel smoothing and obtain the optimal rate of convergence. \cite{DNP16} study the global fit of boundary by constrained polynomial splines and obtain the asymptotic rate of global convergence. Although we only consider the point-wise estimation as in \cite{JMR14}, we complement both \cite{JMR14} and \cite{DNP16} by establishing the distributional theory and valid inference procedure for our frontier estimator. Overall, \cite{DLN17} provide an excellent and up-to-date literature review on the estimation and inference of the production frontier. 
	
	\cite{BHPW04}, \cite{CF11}, and \cite{Z16} study the inference of extreme quantiles in the contexts of percentiles, linear quantile regressions, and quantile treatment effects, respectively. Recently, \cite{MW16} study the inference of extreme quantiles by what they refer to as fixed-$k$ asymptotics. Our approach takes inspiration from their idea of treating the first-stage estimates as new observations. \cite{W19} further study the estimation of tail properties for censored or truncated data. We differ from the above papers by estimating the data boundary and adopting the quasi-Bayesian inference. \cite{WW16} study the optimal way to combine intermediate quantile estimates in the linear tail quantile regression. Since intermediate quantile estimates are asymptotically normal, the linear combination is optimal. Then, \cite{WW16} derive the optimal weights. On the contrary, we aim to combine extreme quantile estimates, which are not asymptotically normal. The optimal combination may be nonlinear. We propose to use the quasi-Bayesian method to combine these estimates.

	The rest of the article is organized as follows. Section \ref{sec:setup} sets up the model. Section \ref{sec:qr} establishes the asymptotic properties of extreme quantile estimators. Section \ref{sec:infer} investigates the asymptotic properties of our quasi-Bayesian method. Section \ref{sec:sim} examines the inference procedure on the simulated data. Section \ref{sec:app} applies the approach to an empirical application. We conclude with Section \ref{sec:concl}. All proofs are collected in the Appendix.
	
	Throughout this article, capital letters,
	such as $A$, $X$, and $Y$, denote random elements while their corresponding
	lower cases denote realizations. $C$ denotes an arbitrary positive constant that may not be the same in different contexts. For a sequence of random variables $%
	\{U_{n}\}_{n=1}^{\infty }$ and a random variable $U$, $U_{n}\rightsquigarrow
	U$ indicates weak convergence in the sense of \cite{VW96}. Convergence in probability is denoted as $U_n \convP U$.
	

	\section{Setup}
	\label{sec:setup}
	Following the definition in \cite{DNP16}, we suppose that the $n$ pairs of observations $(X_i,Y_i)$ are independently drawn from a joint density function $f(x,y)$. We can interpret $X_i \in \Re_+^p$ and $Y_i \in \Re_+$ as vectors of production factors (inputs) and a scalar output, respectively. The support $\mathbb{T}$ of the joint density $g(\cdot,\cdot)$ is assumed to be of the form 
	$$\mathbb{T} = \{(x,y)|y \leq \psi(x)\} \supset \{(x,y)|g(x,y)>0\} \quad \text{and} \quad \{(x,y)|y >\psi(x)\} \subset \{(x,y)|g(x,y)=0\},$$
	where $\psi(\cdot)$ corresponds to the locus of the curve above which the density $f$ is zero. Intuitively, we can view $\mathbb{T}$ as technology that
	$$\mathbb{T} = \{(x,y) \in \Re_+^{d_x} \times \Re_+ | x \text{ can produce }y  \}.$$
	
	Researchers observe a random sample of $\{X_i,Y_i\}_{i=1}^n$ such that for each $i = 1,\cdots,n$, $(X_i,Y_i)\in \mathbb{T}$. The parameter of interest is $\psi(x)$, the maximal achievable output for a given level of inputs, i.e.,
	$$\psi(x) = \sup\{y | (x,y) \in \mathbb{T}  \}.$$

	\begin{ass}
		$\{Y_i,X_i\}_{i=1}^n$ is i.i.d. $p_0 = \mathbb{P}(X_i\leq x) >0$, where $\leq$ inside the probability operator is pointwise. 
		\label{ass:iid}
	\end{ass}
	
	
	In addition, we follow the literature and assume the free disposability.

	\begin{ass}
		If $(x,y) \in \mathbb{T}$, then $(x',y') \in \mathbb{T}$ for any $(x',y')$ such that $x'\geq x$ (component-wise) and $y' \leq y$.
		\label{ass:FD}
	\end{ass}

	Let $F(y/x) = \mathbb{P}(Y \leq y|X \leq x)$ be the ``non-standard conditional distribution" in the production frontiers literature. Then under Assumption \ref{ass:FD}, \cite{CFS02} propose that 
	\begin{equation}
	\psi(x) = \sup\{ y \geq 0 |F(y/x)<1\}.
	\label{eq:psi1}
	\end{equation}
	
	Following \cite{ADT05} and \cite{DFS10}, we estimate the production frontier at $x$ by $\hat{q}_n(\hat{\tau}_n)$, where
	\begin{equation}
	\begin{aligned}
	\hat{q}_n(\tau) = \argmin_q \sum_{i=1}^n\rho_{\tau}(Y_i - q)\mathds{1}\{X_i \leq x \},
	\label{eq:qhat}
	\end{aligned}
	\end{equation}
	$\rho_\tau(u) = (\tau-\mathds{1}\{u \leq 0\})u$ is \citeauthor{KB78}'s (\citeyear{KB78}) check function, and $\hat{\tau}_n$ is some random sequence that is smaller than but converges to 1. Later, following \cite{DFS10}, we define $\hat{\tau}_n$ that depends on $\hat{p} := \frac{1}{n}\sum_{i=1}^n 1\{X_i \leq x\}$, which is a consistent estimator of $p_0= \mathbb{P}(X_i \leq x)$. The deterministic counterpart of $\hat{\tau}_n$ is denoted as $\tau_n$. We further denote $q(\tau_n) = F^{-1}(\tau_n/x)$ where $F^{-1}(\tau_n/x) = \inf\{y: F(y/x)\geq \tau_n\}$. We omit the dependence of $q(\tau_n)$ and $\hat{q}_n(\hat{\tau}_n)$ on $x$ for brevity as we focus on the point estimation throughout the paper. Based on this notation, $\psi(x)$, the production frontier at $x$, is just $q(1)$. Let $\hat{\tau}_n = 1 - \frac{k}{n \hat{p}}$ for some $k \in (0,\infty)$.
	


	\begin{ass}
		Suppose $k$ is not an integer.
		\label{ass:unique}
	\end{ass}
	The population counterpart of $\hat{\tau}_n$ is $\tau_n=1-\frac{k}{np_0}$. In the literature, $\tau_n$ is referred to as the extreme quantile index by \cite{CH05} and \cite{DFS10}, and as fixed-k asymptotics by \cite{MW16}. For comparison, the quantile index $\tau_n'$ is intermediate if
	\begin{equation}
	n(1-\tau_n') = k_n \rightarrow \infty \quad \text{and} \quad k_n/n \rightarrow 0.
	\label{eq:inter}
	\end{equation}
	Compared with \eqref{eq:inter}, $np_0(1-\tau_n) = k$, which does not diverge to infinite as the sample size increases. However, since $k$ can be greater than 1, we still use interior data points, rather than the maximum of the feasible outputs, for estimation and inference. Therefore, our method is robust to $\lceil k-1 \rceil$\footnote{$\lceil u \rceil$ denotes the smallest integer that is greater than or equal to $u$.} largest outliers, although it is indeed less robust than the existing inference based on the intermediate quantile estimations. The second part of Assumption \ref{ass:unique} is to guarantee that the limiting objective function of our minimization problem in \eqref{eq:qhat} has a unique minimizer. This assumption is mild because we have the freedom to choose $k$ and the integers are sparse on the real line.

	\section{Asymptotic Properties}
	\label{sec:qr}
	Before stating the regularity condition for our asymptotic results, we first introduce some definitions. We say the cumulative distribution function (CDF) $F$ \textit{belongs to the domain of attraction of type \rom{3} generalized extreme value (EV) distributions} if as $z \rightarrow  0$ and any $v>0$,
	\begin{align*}
	\frac{1-F(z_1-vz)}{1-F(z_1-z)} \rightarrow v^{-1/\xi},  	
	\end{align*}
	where $z_1=\sup\{z|F(z)<1\}$ and $\xi<0$ is the EV index.
	
	\begin{ass}
		The conditional CDF of $Y_i$ given $X_i \leq x$ belongs to the domain of attraction of type \rom{3} generalized EV distributions with the EV index $\xi_0 < 0$.
		\label{ass:ev}
	\end{ass}
	Assumption \ref{ass:ev} states that $1-F(y/x)$ decays polynomially (up to some slowing varying function, e.g., $\log(\cdot)$) as $y$ approaching $q(1)$ or equivalently, $F(y/x)$ has a Pareto-type upper tail. This assumption is common in the literature on the inference of extreme quantiles and production frontiers, e.g., \cite{CF11}, \cite{DFS10}, \cite{DFS12}, \cite{DGG14}, \cite{PSW00}, \cite{Z16}. 
	
	%
	\begin{ass}
		Let $k_0>0$ and $m>1$ be two constants. Then, $\lceil mk_0 \rceil > \lceil k_0 \rceil$, where for a non-integer $k$, $\lceil k \rceil$  is the unique integer that satisfies $k \leq \lceil k \rceil \leq k +1$.
		\label{ass:m}
	\end{ass}
	Later, we will propose a random normalization factor ($\hat{\alpha}_n$) for our first-stage extreme quantile estimates. Assumption \ref{ass:m} guarantees that the normalizing factor is well-defined. This condition is innocuous as researchers have the freedom to choose $k_0$ and $m$. We discuss the choice of $k_0$, $m$, and other tuning parameters in practice in Section \ref{sec:tuning}. 
	
	Now we are ready to describe the limiting distribution of our extreme quantile estimators. For a generic $k$ that satisfies Assumption \ref{ass:unique}, let 
	$$Z_\infty(k) = -(\sum_{i=1}^{\lceil k \rceil} \mathcal{E}_i)^{-\xi_0}, \quad Z_\infty^c(k) = Z_\infty(k) + \eta(k), \quad \text{and} \quad \tilde{Z}_\infty(k) = Z_\infty(k)/(Z_\infty(k_0) - Z_\infty(mk_0)),$$ 
	where $\{\mathcal{E}_i\}_{i \geq 1}$ is a sequence of i.i.d. standard exponential random variables and $\eta(\cdot) = (\cdot)^{-\xi_0}$.
	\begin{thm}
		Let $\hat{\alpha}_n = (\hat{q}_n(1-\frac{k_0}{n\hat{p}}) - \hat{q}_n(1-\frac{mk_0}{n\hat{p}}))^{-1}$, $\hat{\tau}_{nl} = 1-\frac{k_l}{n\hat{p}}$ for $l = 1,\cdots,L$. If Assumptions \ref{ass:iid}, \ref{ass:FD}, and \ref{ass:ev} hold, and Assumption \ref{ass:unique} holds for $k=k_0, mk_0, k_1, \cdots, k_L$, then
		\begin{equation}
		\label{eq:joint}
		\begin{aligned}
		\begin{pmatrix}
		\hat{\alpha}_n(\hat{q}_n(\hat{\tau}_{n1}) - q(1)) \\
		\vdots  \\
		\hat{\alpha}_n(\hat{q}_n(\hat{\tau}_{nL}) - q(1))\\
		\end{pmatrix} \convD \begin{pmatrix}
		\tilde{Z}_\infty(k_1) \\
		\vdots \\
		\tilde{Z}_\infty(k_L)\\
		\end{pmatrix}.
		\end{aligned}
		\end{equation}
		\label{thm:feasible}
	\end{thm}
	
	Several comments are in order. First, Theorem \ref{thm:feasible} establishes the joint asymptotic distribution of $(\hat{q}_n(\hat{\tau}_{n1}), \cdots,\hat{q}_n(\hat{\tau}_{nL}))$ which extends the univariate result established in \citet[Theorem 2.2]{DFS10}. Second, we follow \cite{BHPW04} and \cite{CF11} and use a feasible convergence rate $\hat{\alpha}_n$ that is valid without any additional assumption on the tail distribution of the feasible output. Third, \eqref{eq:joint} and the fact that $\hat{\alpha}_n \rightarrow \infty$ imply that $(\hat{q}_n(\hat{\tau}_{n1}), \cdots,\hat{q}_n(\hat{\tau}_{nL}))$ are all consistent estimates for the production frontier $q(1)$. The remaining question is how to combine these $L$ estimates to construct a valid point estimate and confidence interval for $q(1)$. In the next section, we achieve this goal by the quasi-Bayesian method.

	\section{Inference}
	\label{sec:infer}
	As pointed out by \cite{BF81} and \cite{ZK99}, the standard bootstrap inference for the extreme quantile estimators is inconsistent. Instead, we combine $L$ extreme quantile estimators via a second-stage quasi-Bayesian method to infer the production frontier. Such method is the optimal way to combine these $L$ estimates, as shown in Theorem \ref{thm:optimal} below. It is also possible to just use one extreme quantile estimator and its asymptotic distribution to make inference, which is a special case of our proposed method when $L=1$.  However, this may lose information. 
	
	Denote $\tilde{Z}_n(k_l) = \hat{\alpha}_n (\hat{q}_n(\hat{\tau}_{nl}) - q(1))$ for $\hat{\tau}_{nl} = 1-k_l/(n\hat{p})$, $l=1,\cdots,L$. Then, Theorem \ref{thm:feasible} shows
	\begin{equation*}
	\begin{aligned}
	\begin{pmatrix}
	\tilde{Z}_n(k_1) \\
	\vdots  \\
	\tilde{Z}_n(k_L)\\
	\end{pmatrix} \convD \begin{pmatrix}
	\tilde{Z}_\infty(k_1) \\
	\vdots \\
	\tilde{Z}_\infty(k_L)\\
	\end{pmatrix}.
	\end{aligned}
	\end{equation*}
	We view $(\tilde{Z}_n(k_{1}),\cdots,\tilde{Z}_n(k_{L}))$ as new observations, whose joint density is parameterized by $q(1)$ and converges to the joint PDF of $(\tilde{Z}_\infty(k_1),\cdots,\tilde{Z}_\infty(k_L))$, which is denoted as $f(\cdot;\xi_0)$. Although we cannot calculate the exact finite sample likelihood of $(\tilde{Z}_n(k_{1}),\cdots,\tilde{Z}_n(k_{L}))$, we can approximate it by $f(\cdot;\xi_0)$. Then, by putting a prior on $q(1)$, we can write down the posterior distribution and conduct quasi-Bayesian inference.\footnote{We call the method ``quasi-Bayesian" because we do not use the true finite-sample likelihood.}

	
	The quasi-Bayesian estimator $\hat{q}^{BE}$ of $q(1)$ minimizes the average risk, i.e.,
	\begin{align}
	& \hat{q}^{BE} \notag \\
	= & \argmin_q \int_{\Omega} \ell_n(q-\overline{q})f(\hat{\alpha}_n(\hat{q}_n(\hat{\tau}_{n1}) -\overline{q}),\cdots,\hat{\alpha}_n(\hat{q}_n(\hat{\tau}_{nL}) -\overline{q});\xi )\pi(\overline{q}) \phi(\frac{\xi - \hat{\xi}}{\hat{\sigma}})1\{\xi \in \Gamma\}d\overline{q}d\xi,
	\label{eq:q1}
	\end{align}
	where $\ell_n(u) = \ell(\hat{\alpha}_n u)$ is a loss function, $\pi(\cdot)$ is the prior of $q(1)$, $\Omega$ is the support of $\pi(\cdot)$ that has $q(1)$ as its interior point, $\phi(\cdot)$ is the standard normal PDF, $\hat{\sigma}$ is some (potentially random) bandwidth, and $\Gamma$ is an interval that contains $\xi_0$ as an interior point. 
	
	In \eqref{eq:q1}, we set the prior for $\xi$ as a normal distribution that has mean $\hat{\xi}$ and standard error $\hat{\sigma}$, and is truncated by support $\Gamma$. As the standard error decreases to zero, the effect of this prior will vanish asymptotically. We use this prior to capture the finite sample uncertainty (randomness) of $\hat{\xi}$. In practice, we compute $\hat{\xi}$ by the default Pickands-type method, using function \textbf{dfs\_pick} in the R package \textbf{npbr}. We refer readers to \cite{DLN17} for more details about \textbf{npbr}. The asymptotic normality of Pickands-type estimator has already been established in the literature (e.g., \cite{DD89}) under some extra conditions. This motivates us the use the Gaussian kernel. In addition, \cite{CH00} and \cite{DMZ13} have already established the validity of bootstrap inference under intermediate quantile index asymptotics, which is the same asymptotic scheme that the Pickands estimator is based on. Therefore, it is natural to construct $\hat{\sigma}$ based on the bootstrap standard error of $\hat{\xi}$. The support restriction $\Gamma$ is imposed to further regularize the finite sample behaviour of $\xi$. However, its effect is asymptotically negligible. Although we motive the prior from the asymptotic normality of $\hat{\xi}$, we require only that $\hat{\xi}$ is consistent and $\hat{\sigma} = o_p(1)$ when deriving all the results in this section. Theoretically speaking, it is also valid to just plug in the consistent estimate $\hat{\xi}$ without using the Gaussian prior. We recommend using the prior as it can improve the finite-sample coverage. We provide more details about the estimation of $\hat{\xi}$ and $\hat{\sigma}$ in Section \ref{sec:sim}. 
	
	It is also possible to consider the finite-sample maximum likelihood estimator, i.e., 
	\begin{align*}
	\hat{q}^{MLE} = \argmax_q f(\hat{\alpha}_n(\hat{q}_n(\hat{\tau}_{n1}) -q),\cdots,\hat{\alpha}_n(\hat{q}_n(\hat{\tau}_{nL}) -q);\hat{\xi} ),
	\end{align*}
	which corresponds to the mode of the posterior distribution with uninformative priors, i.e., $\hat{\sigma} = 0$, $\Gamma = \Re$, and $\pi(\cdot) = 1$. We prefer the Bayesian estimator to the MLE for three reasons: (1) the Bayesian estimation does not require optimization, (2) it is natural to use prior of $\xi$ to capture the randomness of the estimator $\hat{\xi}$, and (3) the Bayesian estimator can produce point estimate and confidence intervals simultaneously.

	Let $v = \hat{\alpha}_n(\overline{q} - q(1))$, $z = \hat{\alpha}_n(q - q(1))$, and $\hat{Z}_n^{BE} = \hat{\alpha}_n(\hat{q}^{BE} - q(1))$. Then
	$$\hat{Z}_n^{BE} = \theta_n^{BE}(\tilde{Z}_n(k_1),\cdots,\tilde{Z}_n(k_L);\hat{\xi}),$$
	where
	\begin{equation}
	\theta_n^{BE}(z_{1},\cdots,z_{L};\bar{\xi}) = \argmin_z Q_n(z,z_{1},\cdots,z_{L};\bar{\xi}),
	\label{eq:gamman}
	\end{equation}	
	\begin{align}
	\label{eq:Qn}
	& Q_n(z,z_{1},\cdots,z_{L};\bar{\xi}) \notag \\
	= & \hat{\sigma}^{-1}\int \int_{\Omega_n} \ell(z-v)f(z_1-v,\cdots,z_L-v;\xi)\pi(q(1) + v/\hat{\alpha}_n)\phi(\frac{\xi - \bar{\xi}}{\hat{\sigma}})1\{ \xi  \in \Gamma\}dvd\xi,
	\end{align}
	and $\Omega_n = \hat{\alpha}_n(\Omega-q(1))$. As $n \rightarrow\infty$, it is expected that the RHS of \eqref{eq:Qn} converges (up to some constant) to
	\begin{equation}
	\begin{aligned}
	Q_\infty(z,z_{1},\cdots,z_{L})=\int_\Re \ell(z-v)f(z_1 - v,\cdots,z_L-v;\xi_0)dv.
	\label{eq:Z1}
	\end{aligned}
	\end{equation}
	Further denote $Z_\infty^{BE} = \theta_\infty^{BE}(\tilde{Z}_\infty(k_1),\cdots,\tilde{Z}_\infty(k_L))$,
	\begin{equation}
	\theta_\infty^{BE}(z_{1},\cdots,z_{L}) = \argmin_z Q_\infty(z,z_{1},\cdots,z_{L}),
	\label{eq:gamma}
	\end{equation}	
	and
	\begin{equation}
	\begin{aligned}
	\tilde{\theta}^{BE}_t(z_1,\cdots,z_L;\bar{\xi}) = \argmin_\gamma \int_{K_t} \ell(\gamma - v)f(z_1-v,\cdots,z_L-v;\bar{\xi}) dv,
	\label{eq:thetatilde}
	\end{aligned}
	\end{equation}
	where $K_t = [-t,t]$  for $t \geq 1$.

	\begin{ass}
		\begin{enumerate}
			\item $\ell(u)$ is convex and $\ell(u) \leq C |u|^{d_1}$ for some constants $C$ and $d_1>0$.
			\item $(\lceil k_0 \rceil,\lceil mk_0 \rceil,\lceil k_1 \rceil,\cdots,\lceil k_L \rceil)$ are distinct from each other. 
			\item $\hat{\xi} \convP \xi_0$ and $\hat{\sigma} = o_p(1)$. 
			\item Let $\Gamma$ be some compact subset of $(-\infty, 0)$ such that $\xi_0$ is in the interior of $\Gamma$, and $\mathcal{N}_0$ be some open neighborhood of $\xi_0$. Then $f(z_1,\cdots,z_L;\xi)$ is continuous in $\xi$ at $\xi_0$, for any fixed $M>0$, 
			\begin{align*}
			\sup_{(z_1,\cdots,z_L) \in [-M,M]^L,\xi \in \Gamma}f(z_1-v,\cdots,z_L-v;\xi) \leq H_{1M}(v)
			\end{align*}
			and 
			\begin{align*}
			\sup_{(z_1,\cdots,z_L) \in [-M,M]^L,(\xi_1,\xi_2) \in \mathcal
				{N}_0}|f(z_1-v,\cdots,z_L-v;\xi_1) - f(z_1-v,\cdots,z_L-v;\xi_2)| \leq H_{2M}(v)|\xi_1-\xi_2|,
			\end{align*}
			where for any fixed $z$, 
			\begin{align*}
			\int |\ell(z-v)|(H_{1M}(v) + H_{2M}(v))dv < \infty. 
			\end{align*}
			\item For some constant $d_2>0$, $\sup_{\bar{\xi} \in \mathcal{N}_0}|\theta_n^{BE}(z_1,\cdots,z_L;\bar{\xi})| \leq C\sum_{l=1}^L|z_l^{d_2}|$, $\sup_{\bar{\xi} \in \mathcal{N}_0}|\tilde{\theta}_t^{BE}(z_1,\cdots,z_L;\bar{\xi})| \leq C\sum_{l=1}^L|z_l^{d_2}|$, and 
			\begin{align*}
			\sup_{v \in [-t,t]}f(z_1-v,\cdots,z_L - v,\xi_0) \leq H_{3t}(z_1,\cdots,z_L)
			\end{align*}
			such that, for any $t \geq 0$ 
			\begin{align*}
			\int_{\Re^L} (\sum_{l=1}^L|z_l^{d_2}|+t)^{d_1} H_{3t}(z_1,\cdots,z_L) dz_1\cdots dz_L < \infty. 
			\end{align*}
			\item $\pi(\cdot)$ is bounded and continuous at $q(1)$.
			\item $Q_\infty(z,\tilde{Z}_\infty(k_1),\cdots,\tilde{Z}_\infty(k_L))$ is finite over a non-empty open set $\mathcal{Z}_0$ and uniquely minimized at some random variable $Z_\infty^{BE}$ w.p.1.
		\end{enumerate}
		\label{ass:rho}
	\end{ass}

	Several comments are in order. First, Assumption \ref{ass:rho}.1 is common in quasi-Bayesian estimations, e.g., \cite{CH03} and \cite{CH04}. Both $l_1$ and $l_2$ loss functions satisfy this assumption. Second, Assumption \ref{ass:rho}.2 ensures the limiting likelihood is well-defined. Third, the consistency requirement for $\hat{\xi}$ is mild.  The bandwidth $\hat{\sigma}$ will converge to zero, which is the standard requirement for the kernel type estimation. Fourth, Assumptions \ref{ass:rho}.4 and \ref{ass:rho}.5 can be verified directly because it is possible to write down $f(z_1,\cdots,z_L;\xi)$ analytically. We provide one example in Proposition \ref{prop:PDF}. In that example, $f(\cdot;\xi)$ depends on the gamma density function, which only takes values on the positive half of the real line and has an exponential tail at $+\infty$. Fifth, unlike the standard quasi-Bayesian estimation, here we only deal with a finite sample with $L$ observations. Following the example after Theorem \ref{thm:simmain}, if $L=1$ and $\pi(\cdot)=1$, then  $$\theta_n^{BE}(z;\xi) = z - c(\xi)$$
	in which
	the $c(\xi) $'s under $l_1$ and $l_2$ loss functions are just the median and mean of the random variable with density 
	$$\frac{\int f(w;\xi+u \hat{\sigma})\phi(u)1\{u \in \Gamma_n \}du}{\int \phi(u)1\{u \in \Gamma_n \}du},$$ 
	respectively. The same comment applies to $\tilde{\theta}_t^{BE}(z;\xi)$ with density 
	\begin{align*}
	\frac{f(u;\xi)1\{u \in z-t,z+t\}}{\int_{z-t}^{z+t}f(u;\xi)du}. 
	\end{align*}
	In these cases, Assumption \ref{ass:rho}.5 holds. Sixth, Assumptions \ref{ass:rho}.1, \ref{ass:rho}.4, and \ref{ass:rho}.5 induce various integrability conditions which are necessary for applying the dominated convergence theorem. Last, Assumption \ref{ass:rho}.7 implies the limiting objective function has a unique minimizer, which is necessary for applying the argmin theorem in \cite{VW96}. This type of assumption is common in the literature of quasi-Bayesian estimations, e.g., \cite{CH03} and \cite{CH04}.
	\begin{thm}
		If Assumptions \ref{ass:iid}, \ref{ass:FD}, \ref{ass:ev}, \ref{ass:m}, and \ref{ass:rho} hold, and Assumption \ref{ass:unique} holds for $k=k_0, mk_0, k_1, \cdots, k_L$, then $\hat{Z}_n^{BE} \convD Z_\infty^{BE}.$
		\label{thm:simmain}
	\end{thm}
	
	We take the special case of $L=1$ to illustrate the distribution of $Z_\infty^{BE}$. When the loss function is quadratic, i.e., $\ell(u) = u^2$, $Z_\infty^{BE}$ minimizes
	$$\int (z-v)^2f(\tilde{Z}_\infty(k)-v;\xi_0)dv.$$
	By the first-order condition and simple calculations, we obtain
	$$Z_\infty^{BE} = \tilde{Z}_\infty(k) - \mathbb{E}\tilde{Z}_\infty(k).$$
	The new limit $Z_\infty^{BE}$ is the demeaned version of the limit (i.e., $\tilde{Z}_\infty(k)$) of the original estimator, exactly because it is designed to minimize the MSE. This illustrates that our approach can automatically correct for the bias of the original estimator. Similarly, when $\ell(u) = |u|$, the quasi-Bayesian estimator is asymptotically median-unbiased, i.e., it minimizes the mean absolute deviation (MAD).
	
	Next, we confirm this property of our estimator for the general case with $L>1$. Let $\theta_n(\cdot)$ be a generic estimator, i.e., a function of data $(z_1,\cdots,z_L)$ and $\hat{\xi}$, and $K$ be a compact subset of $\Re$. Following \cite{CH03}, we denote the finite average risk of the estimator $\theta_n$ in $K$ as
	\begin{equation}
	\begin{aligned}
	AR_{\ell,K}(\theta_n) = \int_K \int_{\Re^L} \ell(\theta_n(z_1,\cdots,z_L;\hat{\xi}) -v)f(z_1-v,\cdots,z_L-v;\hat{\xi})dz_1\cdots d_{z_L}dv/\Lambda(K),
	\label{eq:risk}
	\end{aligned}
	\end{equation}
	where $\ell(\cdot)$ and $\Lambda(\cdot)$ are the loss function and the Lebesgue measure, respectively. For a generic sequence of estimators $\{\theta_n(\cdot)\}_{n \geq 1}$, the asymptotic average risk is defined as
	$$AAR_{\ell}(\{\theta_n\}) = \limsup_{t \rightarrow \infty}\limsup_{n\rightarrow \infty}AR_{\ell,K_t}(\theta_n),$$
	in which $K_t$ is defined after \eqref{eq:thetatilde}, i.e., $K_t = [-t,t]$. Recall the Bayesian estimator $\hat{Z}_n^{BE}$ is a function of first-stage estimates $(\tilde{Z}_n(k_1),\cdots,\tilde{Z}_n(k_L))$ and $\hat{\xi}$, i.e., $\hat{Z}_n^{BE} =  \hat{Z}_n^{BE} = \theta_n^{BE}(\tilde{Z}_n(k_1),\cdots,\tilde{Z}_n(k_L);\hat{\xi})$. The next theorem establishes some optimality property regarding such function (or equivalently, such way of combining first-stage estimates).  
	\begin{thm}
		If the assumptions in Theorem \ref{thm:simmain} hold, then
		$$AAR_{\ell}(\{\theta_n^{BE}\}) = \mathbb{E}\ell(Z_\infty^{BE})~a.s.$$
		In addition, let $\Theta_n$ be the collection of all estimators based on
		$(\tilde{Z}_n(k_1),\cdots,\tilde{Z}_n(k_L))$ and $\hat{\xi}$. Then
		$$\inf_{\theta_n \in \Theta_n}AAR_{\ell}(\{\theta_n\}) = AAR_{\ell}(\{\theta^{BE}_n\})~a.s.$$
		\label{thm:optimal}
	\end{thm}
	
	Theorem \ref{thm:optimal} shows that the quasi-Bayesian estimator achieves the infimum of the asymptotic average risk over the class of estimators constructed based on $(\tilde{Z}_n(k_1),\cdots,\tilde{Z}_n(k_L))$ and $\hat{\xi}$. It is possible to find other better estimators outside this class. Searching for the best estimator for the production frontier is out of the scope of this paper. The main purpose of establishing Theorem \ref{thm:optimal} is the following corollary: the confidence interval constructed using the posterior quantiles controls size asymptotically.
	
	\begin{cor}
		Let $\hat{q}^{BE}(0.5)$, $\hat{q}^{BE}(\tau')$, and $\hat{q}^{BE}(\tau'')$ be the quasi-Bayesian estimators that solve \eqref{eq:q1} with the loss function $\tilde{\ell}_\tau(u) = (\mathds{1}\{u >0 \} - \tau)u = u- \rho_{\tau}(u)$ and $\tau = 0.5$, $\tau'$ and $\tau''$, respectively. Let $Z^{BE}_\infty(0.5)$, $Z^{BE}_\infty(\tau')$ and $Z^{BE}_\infty(\tau'')$ be the limits of $\hat{\alpha}_n(\hat{q}^{BE}(0.5) - q(1))$, $\hat{\alpha}_n(\hat{q}^{BE}(\tau') - q(1))$ and $ \hat{\alpha}_n(\hat{q}^{BE}(\tau'') - q(1))$, respectively. If $0 < \tau' < \tau'' < 1$ and $Z^{BE}_\infty(0.5)$, $Z^{BE}_\infty(\tau')$ and $Z^{BE}_\infty(\tau'')$ are continuously distributed at zero, then
		$$\mathbb{P}(q(1) \leq \hat{q}^{BE}(0.5) ) \rightarrow 0.5 \quad \text{and} \quad \mathbb{P}(q(1) \in \text{CI}^{BE}(\tau''-\tau') ) \rightarrow \tau''-\tau',$$
		where $\text{CI}^{BE}(\tau''-\tau') = (\hat{q}^{BE}(\tau'), \hat{q}^{BE}(\tau''))$.
		\label{cor:posterior_quantile}
	\end{cor}
	The quasi-Bayesian estimator $\hat{q}^{BE}(\tau)$ is just the $\tau$-th posterior quantile. Corollary \ref{cor:posterior_quantile} shows we can construct a median-unbiased estimator and a valid confidence interval based on posterior quantiles. To implement the MCMC method (such as the Metropolis-Hastings algorithm) and obtain the posterior distribution, we have to evaluate $f(\cdot;\xi)$, the joint PDF of
	$(\tilde{Z}_\infty(k_1),\cdots,\tilde{Z}_\infty(k_L))$ at
	$$(\hat{\alpha}_n(\hat{q}_n(\hat{\tau}_{n1}) - \overline{q}),\cdots,\hat{\alpha}_n(\hat{q}_n(\hat{\tau}_{nL}) - \overline{q})).$$
	
	Next, we derive an analytical form for $f(u_1,\cdots,u_L;\xi)$.  
	
	\begin{ass}
		$\lceil k_0 \rceil < \lceil mk_0 \rceil < \lceil k_1 \rceil < \cdots < \lceil k_L \rceil.$
		\label{ass:h}
	\end{ass}
	The order of $k$'s is needed to derive a simple formula for the joint PDF but is not required for Theorem \ref{thm:simmain}. Essentially, Assumption \ref{ass:h} requires that $\lceil k_0 \rceil$ and $\lceil mk_0 \rceil$ are smaller than all the other $k$'s, which makes it much easier to handle the common denominator $Z_\infty(mk_0) - Z_\infty(k_0)$ in $\tilde{Z}_\infty(k_l)$ for $l=1,\cdots,L$.
	
	\begin{prop}
		Let $f_h$ be the PDF of a gamma random variable with shape and scale parameters being equal to $h$ and 1, respectively. If Assumption \ref{ass:h} holds, then
		\begin{equation*}
		\begin{aligned}
		& f(u_1,\cdots,u_L;\xi) \\
		= & \int (-1/\xi)^L \tilde{u}(t,s)^{-L/\xi} \biggl[\prod_{l=1}^Lu_l^{-1/\xi-1}f_{h_l - h_{l-1}}(v_l - v_{l-1})\biggr]  f_{\lceil k_0 \rceil}(s)f_{\lceil mk_0 \rceil -\lceil k_0 \rceil}(t) dsdt,
		\end{aligned}
		\end{equation*}
		where $h_l = \lceil k_l \rceil$ for $1 \leq l \leq L$, $h_0 = \lceil mk_0 \rceil$, $v_l = (u_l\tilde{u}(t,s))^{-1/\xi}$ for $1 \leq l \leq L$, $\tilde{u}(t,s) = (t+s)^{-\xi} - s^{-\xi}$, and $v_0 = t+s$.
		\label{prop:PDF}
	\end{prop}

	Given the analytical form of $f(u_1,\cdots,u_L;\xi)$, the estimates $(\hat{q}_n(\hat{\tau}_{n1}),\cdots,\hat{q}_n(\hat{\tau}_{nL}))$, and the feasible convergence rate $\hat{\alpha}_n$, we can generate MCMC draws from the posterior
	
	$$f(\hat{\alpha}_n(\hat{q}_n(\hat{\tau}_{n1}) -\overline{q}),\cdots,\hat{\alpha}_n(\hat{q}_n(\hat{\tau}_{nL}) -\overline{q});\xi)\pi(\overline{q})\phi(\frac{\xi- \hat{\xi}}{\hat{\sigma}})1\{\xi \in \Gamma \}.$$
	
	Then, we can use these MCMC draws to construct point estimator and confidence interval for $q(1)$. The quasi-Bayesian approach requires several tuning parameters, namely $L$, $(k_0,\cdots,k_L)$, and $m$. We discuss the choices of these tuning parameters in Section \ref{sec:tuning}. We also describe the detail of the MCMC procedure in Section \ref{sec:numcomp}. The R code for the quasi-Bayesian inference is available upon request.

	\section{Simulations\label{sec:sim}}
	In this section, we investigate the finite-sample performance of our
	estimation and inference procedures.
	
	\subsection{Data Generating Processes\label{sec:dgp}}
	
	The data generating process (DGP) is based on the model
	\begin{equation}
	Y_{i}=\psi\left(  X_{i}\right)  \mathcal{U}_{i},X_{i}\sim\text{Unif}\left(
	0,6\right)  ,i=1,2,\dots,n, \label{ModelSim}%
	\end{equation}
	where $\psi\left(  X\right)  $ is a function representing the frontier,
	$\mathcal{U}$ is the error term, and Unif$\left(  0,6\right)  $ denotes the
	uniform distribution over $\left[  0,6\right]  $. We consider three different $\psi\left(  X\right)$'s:%
	\[
	\left(  1\right)  \text{ }\psi(X)=X^{0.5},\quad (2)\text{ }\psi
	(X)=X,\quad \text{and} \quad (3)\text{ }\psi(X)=\frac{X^{2}}{6},
	\]
	which are concave, linear, and convex, respectively.
	
	The first two functional forms have been investigated in simulations in previous papers, e.g., \cite{ADT05} and \cite{DFS10}. Convex
	frontiers, like (3), were adopted in
	simulations in \cite{PSW00} and \cite{MY08} among others.
	
	We combine the above three $\psi\left(  X\right)$'s with the following five distributions of $\mathcal{U}$.
	
	\begin{description}
		\item[(1)] $\mathcal{U}\sim\text{Unif$\left(  0,1\right)  $}$ with density
		evenly distributed over the support $\left[  0,1\right]  $.
		
		\item[(2)] $\mathcal{U}\sim\exp(-u),u\sim$exponential$\left(  \frac{1}%
		{3}\right)  $ with density skewed to the left over the support $\left[
		0,1\right]  $.
		
		\item[(3)] $\mathcal{U}\sim\text{Beta}\left(  \frac{1}{2},\frac{3}{2}\right)
		$ with density skewed more to the left compared to the density in (2).
		
		\item[(4)] $\mathcal{U}\sim\text{Beta}\left(  \frac{3}{2},\frac{1}{2}\right)
		$ with density skewed to the right over the support $\left[  0,1\right]  .$
		
		\item[(5)] $\mathcal{U}\sim$ truncated normal with density, $\frac
		{2\phi\left(  \frac{u-1/2}{1/2}\right)  }{\Phi(1)-\Phi(-1)}$ for $u\in
		\lbrack0,1]$ ($\phi$ and $\Phi$ denote the PDF and the CDF of standard normal, respectively), which is concentrated in the middle over the support $\left[
		0,1\right]  .$
	\end{description}
	
	The five distributions above exhibit different types of tail behaviors. Specifically, in our DGPs, $\xi_{0}$ only depends on
	the density of $\mathcal{U}$. Some simple calculation further shows that $\xi_{0}=-\frac{1}{2}$ for the DGPs with distributions (1), (2), and (5), $\ \xi_{0}=$\ $-\frac{2}{5}$ for the DGPs with distribution (3), and $\xi_{0}=-\frac{2}{3}$ for the DGPs with distribution (4). Note that $\xi_{0}=-\frac{1}{2}$, $\xi_{0}>-\frac{1}{2}$, and $\xi_{0}<-\frac{1}{2}$ when the density of $\mathcal{U}$ is bounded and bounded away from zero, decays to zero, and diverges to infinity at the boundary 1, respectively. Overall, we consider DGPs using all the combinations of the functional forms of $\psi$\ and the distributions of $\mathcal{U}$, which results in 
	15 DGPs in total. We denote them as DGP$(i,j)$, where $i=1,2,3$ represents the functional forms of $\psi(\cdot)$ and $j = 1,\cdots,5$ denotes distributions of $\mathcal{U}$.

	Since the data in our empirical application contains
	four outliers, we also add four outliers to each DGP to check the impact of
	outliers on our estimation and inference procedure. Note \cite{ADT05} also introduce outliers in their simulation
	setup. Specifically, the four
	outliers in our DGPs are%
	\begin{align*}
	&  \left(  \text{Unif}\left(  0,1.5\right)  ,\psi\left(  2.25\right)  \right)
	,\text{ }\left(  1.5+\text{Unif}\left(  0,1.5\right)  ,\psi\left(  3.5\right)
	\right)  ,\text{ }\left(  1.5+\text{Unif}\left(  0,1.5\right)  ,\psi\left(
	4\right)  \right)  ,\\
	&  \text{ and }\left(  3+\text{Unif}\left(  0,1.5\right)  ,\psi\left(
	5.25\right)  \right)  .
	\end{align*}
	We report the results for $x=1.5,$ $3.0$, and $4.5.$ Thus, our
	procedure faces 1, 3, and 4 outliers at $1.5,$ $3.0,$ and $4.5,$ respectively.
	
	\subsection{Tuning parameters}
	
	\label{sec:tuning} The tuning parameters used in our procedures are the
	lower quantile index $k_{0}$, the spacing parameter $m$, and the upper quantile index $k_L$. How to choose those tuning parameters optimally is an important
	yet challenging question. Just as argued in \citet[Section 5]{MW16},
	\textquotedblleft\textit{under }$k_{n}\rightarrow\infty$\textit{, the
		determination of }$k_{n}$\textit{ in a given sample size }$n$\textit{ is
		widely recognized as a difficult issue. But the problem is arguably even
		harder under fixed-}$k$\textit{ asymptotics, as there cannot exist a procedure
		based on the largest }$k$\textit{ order statistics that consistently
		determines whether, say, }$k_{1}<k$\textit{ or }$k_{2}<k_{1}$\textit{ is
		appropriate}\textquotedblright. Here, we provide some rules of thumb for
	$k_{0},m,$ and $k_L$ based on the existing literature and some
	unique features of our own procedure. We leave the formal analysis on the
	higher-order impact of the tuning parameters to future research. 
	
	Note that the
	spacing parameter $m$ and the upper quantile index $k_L$ have been well studied
	in \cite{CF11}. We choose $m$ and $k_L$ based on their recommendation. The role of
	$k_{0}$ is to guard against outliers. We detail our rule-of-thumb choice of the tuning parameters below. 
	
	\begin{description}
		\item[(1)] To be robust against outliers, we set $k_{0}$ as
		\[
		k_{0}=\text{Number of spotted outliers }+2.
		\]
		In the simulation, we set $k_{0}=3,5,$ and
		$6$ for $X=1.5,$ $3.0,$ and $4.5,$ respectively.
		
		\item[(2)] As for $m$, \cite{CF11} suggest using $m=1+\frac{sp}{\lceil
			k_{0}\rceil}$, where $sp\in\left[  2,20\right]  $,\footnote{The original
			formula in \cite{CF11} is $m=1+\frac{d+sp}{\lceil k_{0}\rceil}$, where $d$ is
			the dimension of the regressors. In our case, there is no regressor so $d=0$.}
		and they set $sp=5$ for simulations and applications. We follow them and set
		$sp=5$. Then, $m=1+\frac{5}{\lceil k_{0}\rceil},$ which implies
		$k_{1}=k_{0}+5$.
		
		\item[(3)]\cite{CF11} point out that the fixed-k asymptotics has a
		better approximation of the finite sample distribution when $k_{L}$ is within the range $[40,80]$. In addition, denote the effective sample size for each $n$ and $x$ as $n\hat{p}$, where $\hat{p} = \frac{1}{n}\sum_{i=1}^n1\{X_i \leq x\}$. Then, $k_L/(n\hat{p})$ is the quantile index of the $k_L$-th order statistic in the effective sample. As our theories rely on the extreme quantile asymptotics, we require such quantile index to be close to zero. In practice, we require $k_L \leq 0.1 n \hat{p}$. Therefore, our rule-of-thumb choice of $k_L$ is $k_{L}=\min(0.1n\hat{p},40)$. Note as $n \rightarrow \infty$, $\min(0.1n\hat{p},40)$ reduces to $40$, which fits the extreme quantile asymptotics. For robustness check, we also consider $k_{L}' = \min(0.1n\hat{p},35)$ and $k_{L}^{''} = \min(0.1n\hat{p},45)$ for all 15 DGPs. All the simulation results are very close. 
		
		\item[(4)] The more quantiles are used, the more efficient is our quasi-Bayesian estimator. Thus, we use all the integers between $k_1$ and $k_L$ for estimation, i.e., we let 
		$$\left\{  k_{l}\right\}
		_{l=0}^{L} = \left\{  k_{0},k_{1},k_{1}+1,k_{1}+2,...,k_{L}-1,k_{L}%
		\right\}  .$$ 
		Once $k_0$, $m$, and $k_L$ are determined, the whole sequence $\{k_l\}_{l=0}^L$ is determined. 
	\end{description}
	
	
	To show the robustness of our procedure to $k_{0}$ and $sp$ $($or $m),$ we
	experiment $k_{0}=$Number of spotted outliers $+3$ or $k_{0}=$Number of
	spotted outliers $+4$ and $sp=6$ or $7$ for DGP(1,1) and DGP(2,1). These results can be found in Sections \ref{sec:simk0} and \ref{sec:simsp} in the supplement.
	
	\subsection{Detail about the MCMC procedure}
	\label{sec:numcomp} 
	The numerical evaluation of the joint density function
	established in Proposition \ref{prop:PDF} is detailed in Section
	\ref{sec:PDF} in the supplement. The length of burn-in sequence and MCMC sequence should be set
	as large as computationally possible. We use $4,000$ and $10,000$,
	respectively. Second, we need to determine the initial values of the MCMC.
	Given $x$, the initial value $\bar{q}^{\ast}$ is computed by
	\[
	\bar{q}^{\ast}=\underset{q}{\arg\min}\sum_{i=1}^{n}\rho_{\tau}\left(
	Y_{i}-q\right)  1\left\{  X_{i}\leq x\right\}  ,
	\]
	where $\tau=0.99$. The initial value of $\xi_0$ is $\hat{\xi}$ computed by the \textbf{rho\_momt\_pick }function in the R package \textbf{npbr}. We will provide more detail about the estimation of $\xi_0$ in Section \ref{sec:evindex} below. 
	
	
	\subsection{Estimators for comparison}
	
	Based on the characterization in \citet[Table 2]{DLN17}, our paper considers
	point-wise and robust estimation of the production frontier under the
	assumption that the frontier is monotone only. Among all the estimators
	mentioned in \citet[Table 2]{DLN17}, the moment- and Pickands-type estimations
	proposed by \cite{DFS10} and the the probability-weighted moment frontier
	estimation proposed by \cite{DFS12} are in the same category as ours and
	produce not only point estimates but also confidence intervals. Therefore, we
	will compare our method to them. The four methods are labelled as follows:
	
	\begin{description}
		\item[(1)] \textquotedblleft Quasi-Bayesian": our quasi-Bayesian method,
		
		\item[(2)] \textquotedblleft Mom\textquotedblright: the moment frontier estimator,
		
		\item[(3)] \textquotedblleft Momt\_pick\textquotedblright: the Pickands
		frontier estimator,
		
		\item[(4)] \textquotedblleft Pwm\textquotedblright: the probability-weighted
		moment frontier estimator.
	\end{description}
	
	The estimation procedures for ``Mom", ``Momt\_pick", and ``Pwm" are described
	in Section $\ref{sec:compute-estimators}$ in the supplement.
	
	\subsection{Prior of $\bar{q}$}
	\label{sec:simprior}
	For all simulations, we simply set $\pi\left(\cdot\right)  =1,$ which is the uninformative prior for $\bar{q}$. We experiment $\pi\left(\cdot\right)  $ as normal with the mean as the
	initial value $\bar{q}^{\ast}$ and the variance as 1 or 1.5 for DGPs(1,1) and (2,1). The simulation results show our inference method is insensitive to the choice of prior distributions. The detail can be found in Section \ref{sec:simpi} in the supplement.

	\subsection{Estimation of $\xi_0$}
	
	\label{sec:evindex} All four estimation methods above require the estimation
	of the EV index $\xi_0$. For fair comparison, for each replication, we force all
	estimators to share the same EV index estimate $\hat{\xi}$, which is the
	negative reciprocal of the output of the function \textbf{rho\_momt\_pick} in
	\textbf{npbr} with arguments \textbf{method = \textquotedblleft
		Pickands\textquotedblright} and support intervals $\left(  1,3\right)$, $\left(  0.5,2.5\right)  $ and
	$\left(  1.5,3.5\right)  ,$ when the true values of $-1/\xi_0$ are 2 (error distributions (1), (2), (5)), 2.5 (error distribution (3)), and 1.5 (error distribution (4)), respectively. When
	the effective sample size is small, occasionally, the function
	\textbf{rho\_momt\_pick }may return NA value. In this case, we propose to use
	the following equation to compute $\hat{\xi}$:
	\begin{equation}
	\label{eq:xihat}\hat{\xi}=\text{Med}\left[  \left\{  -\frac{1}{r\log\left(
		l\right)  }\log\left[  \frac{\hat{q}_{n}\left(  \tau\left(  \tilde{m}%
		l^{r}\tilde{k}\right)  \right)  -\hat{q}_{n}\left(  \tau\left(  l^{r}\tilde
		{k}\right)  \right)  }{\hat{q}_{n}\left(  \tau\left(  \tilde{m}\tilde
		{k}\right)  \right)  -\hat{q}_{n}\left(  \tau\left(  \tilde{k}\right)
		\right)  }\right]  \right\}  _{\tilde{k}=K-c\hat{p}}^{K+c\hat{p}}\right]  ,
	\end{equation}
	where $\text{Med}\left(  \cdot\right)  $ denotes the median operator,
	$\tau\left(  k\right)  =1-\frac{k}{n\hat{p}}$, $l=2$, $r=2$, $\tilde{m}=1.5$,
	$c=20$ and the tuning parameter $K=n\hat{p}/10$. Once the estimated $\hat{\xi
	}$ is outside $\Gamma$, we directly assume it equals to the closest boundary.
	
	For our quasi-Bayesian method, we use the truncated normal prior
	$\phi(\frac{\xi-\hat{\xi}}{\hat{\sigma}})1\{\xi\in\Gamma\}$, where
	$\hat{\sigma}$ is obtained via bootstrap. Specifically, for the $s$-th
	bootstrap sample, we can generate $\{\zeta_{i}^{(s)}\}_{i=1}^{n}$, which is a
	sequence of i.i.d. standard exponentially distributed random variables. For a
	generic quantile index $\tau$, we can compute
	\[
	\hat{q}_{n}^{(s)}(\tau)=\arg\min_{q}\sum_{i=1}^{n}\zeta_{i}^{(s)}\rho_{\tau
	}(Y_{i}-q)1\{X_{i}\leq x\}.
	\]
	Then, the EV index estimator $\hat{\xi}^{(s)}$ for the $s$-th bootstrap sample
	can be computed similarly using \eqref{eq:xihat} with $\hat{q}_n(\cdot)$ and $K$ replaced by $\hat{q}_n^{(s)}(\cdot)$ and
	$\hat{K}$, respectively, where $\hat{K}$ is the optimal tuning
	parameter associated with the EV index $\hat{\xi}$ obtained by function
	\textbf{rho\_momt\_pick}. For some replication, when \textbf{rho\_momt\_pick}
	returns NA value, we instead set $\hat{K}=n\hat{p}/10$. For
	each replication, we repeat the above procedure for $s=1,\cdots,S$, where $S$
	is a sufficiently large positive integer and obtain $\left\{  \hat{\xi}%
	^{(s)}\right\}  _{s=1}^{S}$. We let
	\[
	\hat{\sigma}=\frac{c_{\sigma}(c_{0.75}-c_{0.25})}{\text{normal inverse}%
		(0.75)-\text{normal inverse}(0.25)},
	\]
	where $c_{\sigma}=1.5+1.5 \cdot 1\{\hat{\xi}>-0.5 \}$, $c_{0.75}$ and $c_{0.25}$ are the $75\%$ and $25\%$ quantiles of
	$\left\{  \hat{\xi}^{(s)}\right\}  _{s=1}^{S}$, and $\text{normal
		inverse}(0.75)$ and $\text{normal inverse}(0.25)$ are the $75\%$ and $25\%$
	quantiles of the standard normal distribution, respectively.
	
	\subsection{Results}
	
	\label{sec:simulresults}We construct 95\% confidence intervals for the four
	estimation methods. We report the results of the coverage probabilities and
	average lengths of the CIs for $\psi(x)$ at $x=1.5,3.0,4.5$. Due to the length limit, we report results for DGPs(1,1), (2,1), and (3,1) in Tables \ref{tabledgp1}--\ref{tabledgp3}. The results for the rest 12 DGPs and various robustness checks are relegated to the supplement. Given the value of $x$, we report the performance when the sample
	size $n=500$, $1,000$, $2,000$ and $4,000$. All simulations are repeated $1,000$ times. 
	

	We can make several observations. First, the quasi-Bayesian method controls
	size well, even when the effective sample size is small. Meanwhile, the CIs for the Pickands, the moment,
	and the probability-weighed frontier methods over- or under-cover quite a bit in the majority of cases. The
	simulation results in Section \ref{sec:addsim2} in the supplement further show that even we
	use the true EV index, the inferences using these methods still have the same issue. This may be due to the fact
	that their tuning parameters selected by \textbf{npbr} are not optimal for
	inference purpose. Second, the average length of the quasi-Bayesian method is in general the shortest among
	all four methods, despite the fact that its coverage rate is also closest to the nominal rate. Third, both the coverage rates and average lengths of
	our method are stable across different values of
	$k_{L}$. In addition, in Sections \ref{sec:simpi}--\ref{sec:simsp} in the supplement, we show the performance of quasi-Bayesian method is insensitive to the choices of $k_{0}$, $sp$ (or equivalently $m$) and the prior $\pi\left(  \bar{q}\right)$. Fourth, the average lengths for the quasi-Bayesian method decrease as the sample size
	increases. This indicates the validity of the fixed-k type asymptotics, which our theory relies on.

	\begin{table}[H]
		\centering{}\centering {\caption{DGP(1,1)}
			\label{tabledgp1} }%
		\begin{adjustbox}{max width=\textwidth}
			\begin{tabular}{r|ccc|ccc}
				\multicolumn{7}{c}{}  \\
				\hline\hline			
				\multicolumn{7}{c}{}  \\
				\multicolumn{7}{c}{Panel A: $x = 1.5$}  \\\hline
				& \multicolumn{3}{c|}{Quasi-Bayesian} & \multicolumn{3}{c}{Pickands}\\\hline
				& $k_{L}'$ & $k_{L}$ & $k_{L}^{''}$ & Mom & Momt-pick & Pwm \\
				\hline
				$n=500$& 0.9660 & 0.9650 & 0.9660 & 1.0000 & 0.9680 & 0.9570 \\
				$n p_0=125$& (0.5328) & (0.5332) & (0.5283) & (0.8451) & (1.9062) & (1.1307) \\
				$n=1000$& 0.9800 & 0.9820 & 0.9820 & 0.9940 & 0.9730 & 0.9880 \\
				$n p_0=250$& (0.3062) & (0.3069) & (0.3068) & (0.4978) & (1.2312) & (0.8488) \\
				$n=2000$& 0.9730 & 0.9700 & 0.9710 & 0.9930 & 0.9800 & 0.9930 \\
				$n p_0=500$& (0.2138) & (0.2117) & (0.2119) & (0.3267) & (0.8659) & (0.5960) \\
				$n=4000$& 0.9650 & 0.9580 & 0.9550 & 0.9880 & 0.9850 & 0.9970 \\
				$n p_0=1000$& (0.1477) & (0.1467) & (0.1458) & (0.2303) & (0.6343) & (0.4158) \\
				\hline			
				\multicolumn{7}{c}{}  \\
				\multicolumn{7}{c}{Panel B: $x = 3.0$}  \\\hline
				& \multicolumn{3}{c|}{Quasi-Bayesian} & \multicolumn{3}{c}{Pickands}\\\hline
				& $k_{L}'$ & $k_{L}$ & $k_{L}^{''}$ & Mom & Momt-pick & Pwm \\
				\hline				
				$n=500$& 0.9390 & 0.9280 & 0.9360 & 1.0000 & 0.9900 & 0.7020 \\
				$n p_0=250$& (0.5485) & (0.5490) & (0.5495) & (0.8204) & (1.9218) & (0.9080) \\
				$n=1000$& 0.9500 & 0.9590 & 0.9450 & 0.9990 & 0.9830 & 0.8150 \\
				$n p_0=500$& (0.3616) & (0.3585) & (0.3559) & (0.5384) & (1.3046) & (0.7082) \\
				$n=2000$& 0.9520 & 0.9530 & 0.9520 & 0.9990 & 0.9880 & 0.8930 \\
				$n p_0=1000$& (0.2461) & (0.2424) & (0.2410) & (0.3669) & (0.9549) & (0.5213) \\
				$n=4000$& 0.9390 & 0.9440 & 0.9380 & 0.9990 & 0.9950 & 0.9390 \\
				$n p_0=2000$& (0.1715) & (0.1686) & (0.1667) & (0.2567) & (0.7050) & (0.3703) \\
				\hline			
				\multicolumn{7}{c}{}  \\
				\multicolumn{7}{c}{Panel C: $x = 4.5$}  \\\hline
				& \multicolumn{3}{c|}{Quasi-Bayesian} & \multicolumn{3}{c}{Pickands}\\\hline
				\hline
				& $k_{L}'$ & $k_{L}$ & $k_{L}^{''}$ & Mom & Momt-pick & Pwm \\
				\hline
				$n=500$& 0.9540 & 0.9550 & 0.9550 & 0.9980 & 0.9850 & 0.8430 \\
				$n p_0=375$& (0.5187) & (0.5167) & (0.5171) & (0.6377) & (1.7471) & (0.4787) \\
				$n=1000$& 0.9320 & 0.9290 & 0.9200 & 0.9880 & 0.9900 & 0.9250 \\
				$n p_0=750$& (0.3730) & (0.3665) & (0.3632) & (0.4262) & (1.2530) & (0.3641) \\
				$n=2000$& 0.9760 & 0.9720 & 0.9720 & 0.9930 & 0.9920 & 0.9400 \\
				$n p_0=1500$& (0.2432) & (0.2396) & (0.2384) & (0.3124) & (0.9603) & (0.2622) \\
				$n=4000$& 0.9830 & 0.9740 & 0.9750 & 0.9960 & 0.9990 & 0.9660 \\
				$n p_0=3000$ & (0.1719) & (0.1680) & (0.1665) & (0.2285) & (0.7277) & (0.1814) \\
				\hline
				\multicolumn{7}{c}{}  \\
				\hline\hline
			\end{tabular}
		\end{adjustbox}
		
		{ Notes: $k_{L}'=\min\{\lceil0.10n \hat{p}\rceil, 35\}, k_{L}=\min\{\lceil0.10n \hat{p}\rceil, 40\}$, and $k_{L}^{''}=\min\{\lceil0.10n \hat{p}\rceil, 45\}$. The coverage rates and average lengths of the CIs
			(in parentheses) are reported.}
	\end{table}
	
	\begin{table}[H]
		\centering{}\centering {\caption{DGP(2,1)}
			\label{tabledgp2} }%
		\begin{adjustbox}{max width=\textwidth}
			\begin{tabular}{r|ccc|ccc}
				\multicolumn{7}{c}{}  \\
				\hline\hline			
				\multicolumn{7}{c}{}  \\
				\multicolumn{7}{c}{Panel A: $x = 1.5$}  \\\hline
				& \multicolumn{3}{c|}{Quasi-Bayesian} & \multicolumn{3}{c}{Pickands}\\\hline
				& $k_{L}'$ & $k_{L}$ & $k_{L}^{''}$ & Mom & Momt-pick & Pwm \\
				\hline
				$n=500$& 0.9700 & 0.9660 & 0.9720 & 0.9980 & 0.9820 & 0.9330 \\
				$n p_0=125$& (0.840) & (0.8180) & (0.830) & (1.4814) & (3.9622) & (1.4388) \\
				$n=1000$& 0.9780 & 0.9800 & 0.9770 & 0.9970 & 0.9940 & 0.9630 \\
				$n p_0=250$& (0.4799) & (0.4805) & (0.4793) & (1.0809) & (3.1801) & (1.1207) \\
				$n=2000$& 0.9690 & 0.9620 & 0.9610 & 0.9760 & 0.9870 & 0.9730 \\
				$n p_0=500$& (0.3457) & (0.3419) & (0.3376) & (0.7511) & (2.3309) & (0.8267) \\
				$n=4000$& 0.9770 & 0.9680 & 0.9700 & 0.9550 & 0.9740 & 0.9790 \\
				$n p_0=1000$& (0.2509) & (0.2481) & (0.2459) & (0.4650) & (1.4488) & (0.5810) \\
				\hline
				\multicolumn{7}{c}{}  \\
				\multicolumn{7}{c}{Panel B: $x = 3.0$}  \\\hline
				& \multicolumn{3}{c|}{Quasi-Bayesian} & \multicolumn{3}{c}{Pickands}\\\hline
				& $k_{L}'$ & $k_{L}$ & $k_{L}^{''}$ & Mom & Momt-pick & Pwm \\
				\hline
				$n=500$& 0.9040 & 0.9030 & 0.9060 & 1.0000 & 0.9960 & 0.3990 \\
				$n p_0=250$& (0.8849) & (0.8871) & (0.8882) & (1.8113) & (4.8315) & (1.1930) \\
				$n=1000$& 0.9080 & 0.9020 & 0.8900 & 0.9990 & 0.9940 & 0.5420 \\
				$n p_0=500$& (0.5986) & (0.5873) & (0.5768) & (1.2152) & (3.4421) & (1.0037) \\
				$n=2000$& 0.9610 & 0.9600 & 0.9640 & 0.9830 & 0.9840 & 0.7430 \\
				$n p_0=1000$& (0.4145) & (0.4075) & (0.4019) & (0.7482) & (2.1716) & (0.7682) \\
				$n=4000$& 0.9410 & 0.9410 & 0.9350 & 0.9750 & 0.9780 & 0.8010 \\
				$n p_0=2000$& (0.2976) & (0.2915) & (0.2883) & (0.4875) & (1.4713) & (0.5490) \\
				\hline
				\multicolumn{7}{c}{}  \\
				\multicolumn{7}{c}{Panel C: $x = 4.5$}  \\\hline
				& \multicolumn{3}{c|}{Quasi-Bayesian} & \multicolumn{3}{c}{Pickands}\\\hline
				\hline
				& $k_{L}'$ & $k_{L}$ & $k_{L}^{''}$ & Mom & Momt-pick & Pwm \\
				\hline
				$n=500$& 0.9550 & 0.9570 & 0.9530 & 0.9980 & 0.9930 & 0.5590 \\
				$n p_0=375$& (0.8537) & (0.8396) & (0.8420) & (1.5060) & (4.9829) & (0.7829) \\
				$n=1000$& 0.9270 & 0.9280 & 0.9250 & 0.9810 & 0.9910 & 0.6990 \\
				$n p_0=750$& (0.6194) & (0.6046) & (0.5975) & (0.9935) & (3.1977) & (0.6487) \\
				$n=2000$& 0.9610 & 0.9640 & 0.9600 & 0.9650 & 0.9760 & 0.8410 \\
				$n p_0=1500$& (0.4328) & (0.4241) & (0.4176) & (0.6307) & (2.0447) & (0.5009) \\
				$n=4000$& 0.9680 & 0.9650 & 0.9640 & 0.9550 & 0.9710 & 0.8640 \\
				$n p_0=3000$ & (0.3019) & (0.2963) & (0.2925) & (0.4454) & (1.4648) & (0.3756) \\
				\hline
				\multicolumn{7}{c}{}  \\
				\hline\hline
			\end{tabular}
		\end{adjustbox}
		
		{ Notes: $k_{L}'=\min\{\lceil0.10n \hat{p}\rceil, 35\}, k_{L}=\min\{\lceil0.10n \hat{p}\rceil, 40\}$, and $k_{L}^{''}=\min\{\lceil0.10n \hat{p}\rceil, 45\}$. The coverage rates and average lengths of the CIs
			(in parentheses) are reported.}
	\end{table}
	
	\begin{table}[H]
		\centering{}\centering {\caption{DGP(3,1)}
			\label{tabledgp3} }%
		\begin{adjustbox}{max width=\textwidth}
			\begin{tabular}{r|ccc|ccc}
				\multicolumn{7}{c}{}  \\
				\hline\hline			
				\multicolumn{7}{c}{}  \\
				\multicolumn{7}{c}{Panel A: $x = 1.5$}  \\\hline
				& \multicolumn{3}{c|}{Quasi-Bayesian} & \multicolumn{3}{c}{Pickands}\\\hline
				& $k_{L}'$ & $k_{L}$ & $k_{L}^{''}$ & Mom & Momt-pick & Pwm \\
				\hline
				$n=500$& 0.9670 & 0.9660 & 0.9640 & 1.0000 & 0.9950 & 0.9400 \\
				$n p_0=125$& (0.9531) & (0.9424) & (0.9367) & (2.0546) & (5.0625) & (1.7818) \\
				$n=1000$& 0.9780 & 0.9810 & 0.9800 & 0.9960 & 0.9930 & 0.9680 \\
				$n p_0=250$& (0.5472) & (0.5478) & (0.5480) & (1.4625) & (3.9807) & (1.3816) \\
				$n=2000$& 0.9770 & 0.9740 & 0.9770 & 0.9940 & 0.9990 & 0.9770 \\
				$n p_0=500$& (0.4063) & (0.3952) & (0.3872) & (1.0924) & (3.0601) & (1.0267) \\
				$n=4000$& 0.9770 & 0.9760 & 0.9780 & 0.9840 & 0.9920 & 0.9890 \\
				$n p_0=1000$& (0.3210) & (0.3137) & (0.3078) & (0.7476) & (2.1403) & (0.7322) \\
				\hline
				\multicolumn{7}{c}{}  \\
				\multicolumn{7}{c}{Panel B: $x = 3.0$}  \\\hline
				& \multicolumn{3}{c|}{Quasi-Bayesian} & \multicolumn{3}{c}{Pickands}\\\hline
				& $k_{L}'$ & $k_{L}$ & $k_{L}^{''}$ & Mom & Momt-pick & Pwm \\
				\hline
				$n=500$& 0.8990 & 0.8970 & 0.8980 & 1.0000 & 0.9990 & 0.3320 \\
				$n p_0=250$& (0.9964) & (1.0069) & (1.0045) & (2.7008) & (6.4429) & (1.4510) \\
				$n=1000$& 0.9110 & 0.9140 & 0.9150 & 1.0000 & 0.9990 & 0.4930 \\
				$n p_0=500$& (0.70) & (0.6728) & (0.6505) & (1.8219) & (4.5428) & (1.2463) \\
				$n=2000$& 0.9560 & 0.9610 & 0.9560 & 0.9990 & 0.9980 & 0.6240 \\
				$n p_0=1000$& (0.5249) & (0.5067) & (0.4960) & (1.2479) & (3.3172) & (0.9707) \\
				$n=4000$& 0.9610 & 0.9670 & 0.9680 & 0.9940 & 0.9970 & 0.7610 \\
				$n p_0=2000$& (0.4035) & (0.3944) & (0.3842) & (0.8147) & (2.2525) & (0.7152) \\
				\hline
				\multicolumn{7}{c}{}  \\
				\multicolumn{7}{c}{Panel C: $x = 4.5$}  \\\hline
				& \multicolumn{3}{c|}{Quasi-Bayesian} & \multicolumn{3}{c}{Pickands}\\\hline
				\hline
				& $k_{L}'$ & $k_{L}$ & $k_{L}^{''}$ & Mom & Momt-pick & Pwm \\
				\hline
				$n=500$& 0.9670 & 0.9610 & 0.9690 & 1.0000 & 0.9990 & 0.4430 \\
				$n p_0=375$& (0.9845) & (0.9624) & (0.9619) & (2.4174) & (7.0347) & (1.0180) \\
				$n=1000$& 0.9590 & 0.9630 & 0.9640 & 1.0000 & 1.0000 & 0.6010 \\
				$n p_0=750$& (0.7604) & (0.7392) & (0.7236) & (1.6713) & (4.9101) & (0.8610) \\
				$n=2000$& 0.9570 & 0.9630 & 0.9570 & 0.9870 & 0.9960 & 0.7430 \\
				$n p_0=1500$& (0.5651) & (0.5496) & (0.5368) & (1.1008) & (3.2691) & (0.6784) \\
				$n=4000$& 0.9840 & 0.9800 & 0.9870 & 0.9780 & 0.9990 & 0.8360 \\
				$n p_0=3000$ & (0.4302) & (0.4164) & (0.4069) & (0.7426) & (2.3010) & (0.5187) \\
				\hline
				\multicolumn{7}{c}{}  \\
				\hline\hline
			\end{tabular}
		\end{adjustbox}
		
		{ Notes: $k_{L}'=\min\{\lceil0.10n \hat{p}\rceil, 35\}, k_{L}=\min\{\lceil0.10n \hat{p}\rceil, 40\}$, and $k_{L}^{''}=\min\{\lceil0.10n \hat{p}\rceil, 45\}$. The coverage rates and average lengths of the CIs
			(in parentheses) are reported.}
	\end{table}



	To sum up, the quasi-Bayesian method works well and is not sensitive to
	reasonable choices of tuning parameters. However, we also want to emphasize
	that these results do not mean our method outperforms the existing methods in
	the literature in all respects. First, the performance of other three existing
	estimators can still be improved. Second, the three existing methods are based
	on the intermediate, rather than extreme, quantile estimations. Therefore,
	they can tolerate more outliers. As put by \cite{DFS10}, \textquotedblleft%
	\emph{ It is difficult to imagine one procedure being preferable in all
		contexts. Hence, a sensible practice is not to restrict the frontier analysis
		to one procedure \ldots.}" We view our quasi-Bayesian method as an alternative
	to the existing inference procedures in the literature. The simulation study
	above shows our method has a better control of size in finite samples with
	small or moderate sample sizes.
	
	\section{An Empirical Application}
	
	\label{sec:app} We apply our inference approach to the frontier analysis of
	French post offices observed in 1994. The same dataset is also studied in
	\cite{DFS10}. In this context, $X$ and $Y$ denote the quantity of labor and
	volume of the delivered mails, respectively. The total number of observations
	is 4,000, which is close to what we consider in our simulations. Table \ref{expleT:summary} contains the summary
	statistics of the data.
	
	\begin{table}[H]
		\caption{Summary Statistics}%
		\label{expleT:summary}
		\begin{centering}
			\begin{tabular}{c|ccccccc}
				\hline
				& MEAN & STD & MIN & LQ & MEDIAN & UQ & MAX\\
				\hline
				$X$ & 1592 & 790 & 177 & 1128 & 1338 & 1730 & 4405\\
				$Y$ & 7.709 & 0.612 & 3.829 & 7.349 & 7.698 & 8.062 & 9.576\\
				\hline
			\end{tabular}
			\par\end{centering}
		\centering{}Notes: STD = standard errors, LQ = 25\% quantile, UQ = 75\%
		quantile.{}\end{table}
	
	There are four data points deviating from the rest of the sample (shown as
	circles in Figures \ref{fig:PF} and \ref{fig:CF}). We view them as outliers.
	We use the same sets of tuning parameters as in the simulations. Specifically, we set
	$k_{0}=$number of spotted outliers $+2,$ $m=1+\frac{5}{\left\lceil
		k_{0}\right\rceil }$ (or, equivalently, $sp=5$), and $k_{L}=\min\left\{
	0.1n\hat{p},40\right\}$. In Figure $\ref{fig:PF}$, we report the point estimators and the
	associated $95\%$ confidence intervals of the production frontier for labor between 800
	and 4400. Note the number of observations with labor less than 800 is 187. For comparison, we also report the point estimates of three existing methods considered in Section \ref{sec:sim}, namely \textquotedblleft
	Mom\textquotedblright, \textquotedblleft Momt\_pick\textquotedblright\ and
	\textquotedblleft Pwm.\textquotedblright\

	\cite{CFS02} consider the estimation and inference for the cost
	function.\footnote{We thank a referee for pointing it out.} Denote the cost
	function as $C\left(  y\right) = \inf\{X: Y \geq y\}$, where $(X,Y)$ follow the joint distribution of input and output. Let $\tilde{X}=M_{1}-Y$,
	$\tilde{Y}=M_{2}-X$, and $\tilde{x} = M_1 - y$, where $M_1$ and $M_2$ are two large positive constants such that $\tilde
	{X}$ and $\tilde{Y}$ are always positive. Then, we have $C\left(  y\right) = M_2-\tilde{\psi}(\tilde{x})$ where $\tilde{\psi}(\tilde{x})\equiv \sup\{\tilde{Y}: \tilde{X} \leq \tilde{x}\}$. This means we can transform the cost function $C(y)$ to a production function $\tilde{\psi}(\tilde{x})$. We first estimate and infer $\tilde{\psi}(\tilde{x})$. The corresponding point estimator and confidence interval are denoted as $\widehat{\tilde{\psi}}_m(\tilde{x})$ and $\widetilde{CI}_{m}$, respectively, where $m \in \{\text{``Quasi-Bayesian", ``Mom", ``Momt\_pick", ``Pwm"}\}$. Then, we can obtain the point estimate and confidence interval for $C(y)$ as $M_2 - \widehat{\tilde{\psi}}_m(\tilde{x})$ and $M_2 - \widetilde{CI}_m$, respectively. We emphasize that our quasi-Bayesian point estimate and confidence interval are numerically invariant to the choices of $M_1$ and $M_2$. We maintain Assumptions in Theorem \ref{thm:simmain} for $(\tilde{X},\tilde{Y})$. In implementation, we set $M_1 = \max_{i = 1,\cdots,n}Y_i$ and $M_2 = \max_{i = 1,\cdots,n}X_i$ and use the same sets of tuning
	parameters as discussed above. In Figure \ref{fig:CF}, we report the results for the cost function when the volume of delivered mails ranges between 0 and 7,500. Note the effective sample for the cost function estimation at output $y$ is all the observations with output $Y \geq y$. Therefore, the effective sample size for estimating the cost of 7,500 mails  is 120. We rescale the input between 8,000 and 14,000 in Figure \ref{fig:CF} to better present our results. We also report the ``FDH" point estimates in both Figures \ref{fig:PF} and \ref{fig:CF}. 
	
	\begin{figure}[H]
		\caption{Estimation and Inference for the Production Function}%
		\label{fig:PF}
		\centering{}\includegraphics[scale=0.6]{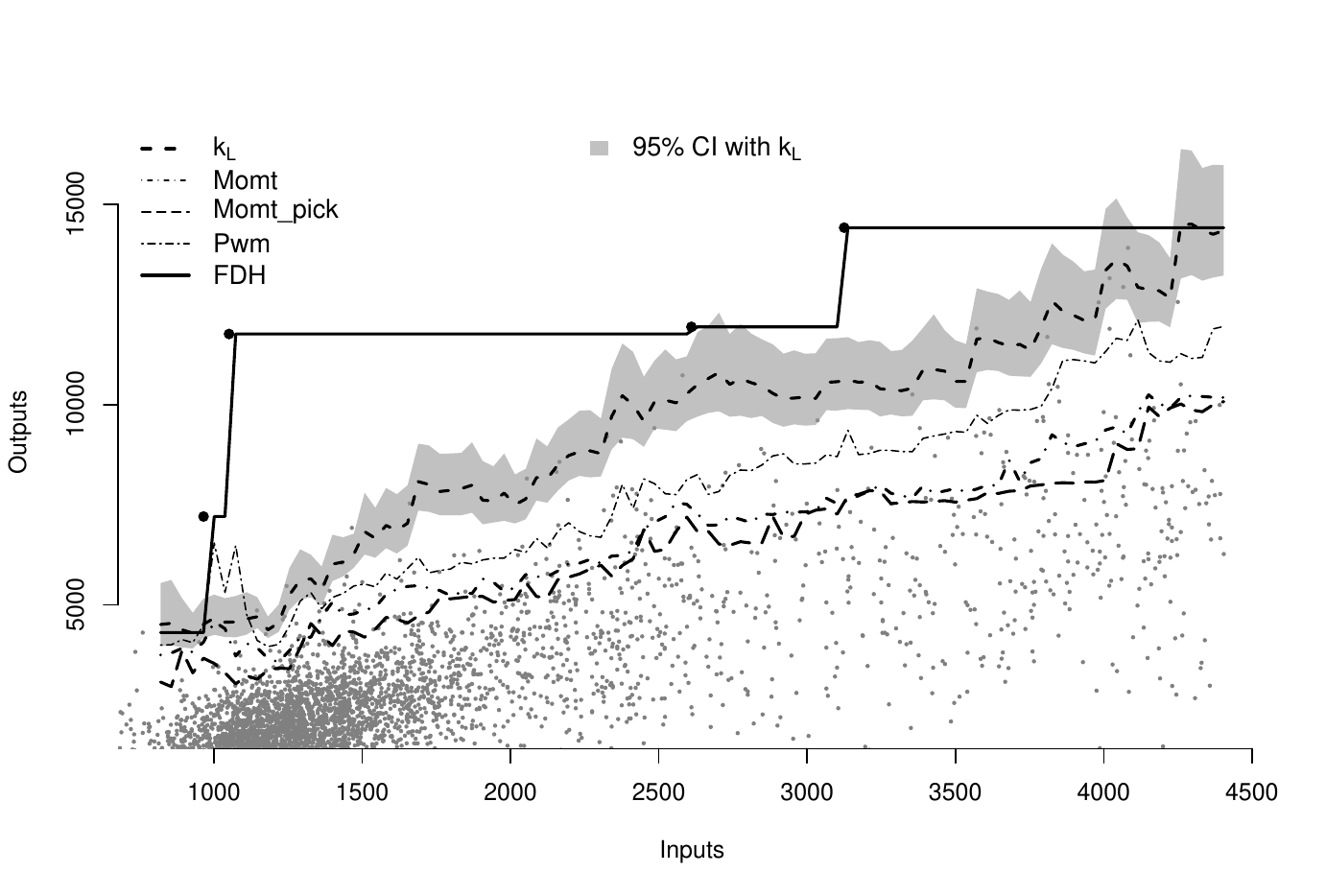}
	\end{figure}
	
	\begin{figure}[H]
		\caption{Estimation and Inference for the Cost Function}%
		\label{fig:CF}
		\centering{}\includegraphics[scale=0.6]{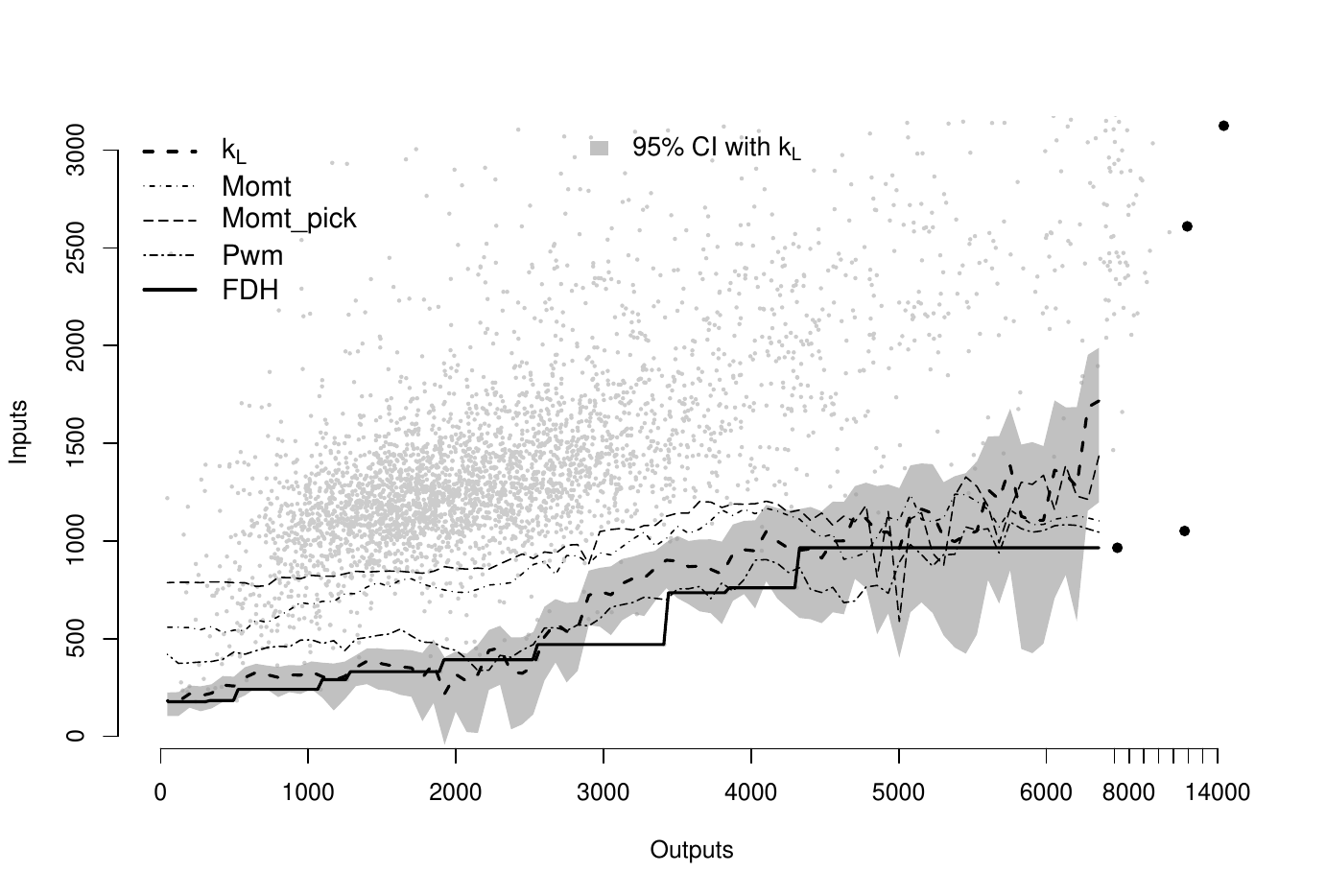}
	\end{figure}

	%
	
	Several comments regarding Figures
	$\ref{fig:PF}$ and $\ref{fig:CF}$ are in order. First, our point estimators
	and confidence intervals are all clearly away from the outliers. This confirms the robustness of our method to outliers. In contrast, the FDH estimator is greatly influenced by the outliers in Figure \ref{fig:PF}. Second, the point estimators of our
	method in general
	envelope the data in Figures \ref{fig:PF} and \ref{fig:CF}. In Figure \ref{fig:PF}, it is above all the other estimators except when the input is around 1,000. In Figure \ref{fig:CF}, it is below all the other estimators when the output ranges from 0 to 3,000. Third, in Figure \ref{fig:CF}, the confidence intervals become wider as the output grows. This is because the effective sample size for the estimation of the cost function shrinks as  output grows. In contrast, in Figure \ref{fig:PF}, as the input grows, more observations are used for estimation and inference of the production frontier, which results in shorter confidence intervals. Fourth, it appears that the majority of French offices deliver mails well below the frontiers (lower than the lower bound of
	the confidence intervals). This indicates that most French offices are not efficient. An
	interesting direction for research is to investigate the relationship between the relative efficiency and
	other demographic variables and identify the key factors that contribute to the inefficiency. Last, the point estimators for the production and cost functions
	are not monotonic in both figures. We note that this is a common feature for 
	point-wise estimation methods. For example, point estimators proposed by \cite{CFS02},
	\cite{ADT05}, \cite{DFS10} and \cite{DFS12} cannot guarantee monotonicity
	either. To ensure monotonicity, one can possibly take our estimates as initial
	estimates and monotonize them following \citet[Section 2.4.2]{DGG14}. The inference procedure,
	however, does need to change accordingly. We believe this direction is rather
	interesting, but is outside the scope of this paper. We leave it for future research.

	\section{Conclusion}
	\label{sec:concl}
	In this article, we propose a quasi-Bayesian method to estimate and infer the production frontier. Our procedure is based on extreme quantile estimators, and thus is robust to a few outliers. The asymptotic validity of our method is theoretically justified. The application to the French post offices dataset shows that our method can be a practical alternative to existing inference methods in the literature.
	

\newpage
\appendix
\begin{center}
	\Large{Supplement to ``Quasi-Bayesian Inference for Production Frontiers"}
\end{center}

\begin{abstract}
	This paper gathers the supplementary material to the original paper. Section \ref{sec:convex} introduces a convexity lemma which will be used later. 	Sections \ref{sec:pf_feasible}, \ref{sec:pf_simmain}, \ref{sec:pf_optimal}, \ref{sec:pf_post}, and \ref{sec:pf_PDF} prove Theorems \ref{thm:feasible}, \ref{thm:simmain}, \ref{thm:optimal}, Corollary \ref{cor:posterior_quantile}, and Proposition \ref{prop:PDF}, respectively. Section  \ref{sec:compute-estimators} describes the computation of three existing methods considered in Section \ref{sec:sim}. Section \ref{sec:PDF} illustrates how to evaluate the density $f\left(u_{1},\dots,u_{L};\xi\right)$ in our MCMC procedure. Section \ref{sec:cal_xi} provides some calculation of $\xi$ for production and cost frontiers. Section \ref{sec:moresimul} contains additional simulation results. 
	
	\medskip
	\textbf{Keywords: Approximate Bayesian Computation, Extreme Value Theory, Fixed-k Asymptotics}

\end{abstract}

\setcounter{page}{1} \renewcommand\thesection{\Alph{section}} %

\renewcommand{\thefootnote}{\arabic{footnote}} \setcounter{footnote}{0}

\setcounter{equation}{0}

\section{The Convexity Lemma due to  \cite{G96} and \cite{K99}}
\label{sec:convex}
We first state the convexity lemma attributed to \cite{G96} and \cite{K99}.
\begin{lem}
	Suppose (i) a sequence of convex lower-semicontinuous functions $Q_n$: $\Re \mapsto \overline{\Re} = \Re \cup \{\pm \infty\}$ marginally converges to $Q_\infty$: $\Re \mapsto \overline{\Re}$ over a dense subset of $\Re$, (ii) $Q_\infty$ is finite over a non-empty open set $\mathcal{Z}_0$, and (iii) $Q_\infty$ is uniquely minimized at a random variable $Z_\infty$, then any minimizer of $Q_n$, denoted $\hat{Z}_n(1)$, converges in distribution of $Z_\infty$.
	\label{lem:convexity}
\end{lem}

\section{Proof of Theorem \ref{thm:feasible}}
\label{sec:pf_feasible}
Let $\alpha_n = 1/(q(1)-q(1-1/(np_0)))$, $\widehat{Z}_n(k) = \alpha_n(\hat{q}_n(\hat{\tau}_{n}) - q(1))$, $\widehat{Z}^c_n(k) = \alpha_n(\hat{q}_n(\hat{\tau}_{n}) - q(\tau_n))$. We divide the proof into two steps. In the first step, we show that 

\begin{equation}
\label{eq:Zhat}
\begin{aligned}
\begin{pmatrix}
\widehat{Z}_n(k_0)  \\
\widehat{Z}_n(mk_0)  \\
\widehat{Z}_n(k_1)  \\
\vdots  \\
\widehat{Z}_n(k_L)  \\
\end{pmatrix} \convD \begin{pmatrix}
Z_\infty(k_0) \\
Z_\infty(mk_0) \\
Z_\infty(k_1) \\
\vdots \\
Z_\infty(k_L)\\
\end{pmatrix}.
\end{aligned}
\end{equation}
In the second step, we derive the desired results in theorem. 

\textbf{Step 1:}\\
Denote $\mathbb{L}(u,v) = (v-u)\mathds{1}\{u < v\}$.
\begin{equation}
\begin{aligned}
\widehat{Z}_n(k) = & \argmin_z \sum_{i=1}^n\frac{1}{\alpha_n}\biggl[\alpha_n(Y_i - q(1)) - z\biggr]\biggl[1-\frac{k}{n\hat{p}} - \mathds{1}\{\alpha_n(Y_i - q(1)) \leq z \} \biggr]\mathds{1}\{X_i\leq x\} \\
= & \argmin_z \sum_{i=1}^n\frac{1}{\alpha_n}\biggl[\alpha_n(Y_i - q(1)) - z\biggr]\biggl[\mathds{1}\{\alpha_n(Y_i - q(1)) > z \}-\frac{k}{n\hat{p}} \biggr]\mathds{1}\{X_i\leq x\} \\
= & \argmin_z kz + \sum_{i=1}^n \mathbb{L}(-\alpha_n(Y_i-q(1)),-z)\mathds{1}\{ X_i \leq x\} \\
= & \argmin_z kz + \int \mathbb{L}(u,-z)d\hat{N}_n \\
= & \argmin_{-z} -kz + \int \mathbb{L}(u,z)d\hat{N}_n \\
= & -\argmin_z -kz + \int \mathbb{L}(u,z)d\hat{N}_n,
\label{eq:obj}
\end{aligned}
\end{equation}
where $\hat{N}_n = \sum_{i=1}^n \mathds{1}\{ -\alpha_n(Y_i-q(1)) \in \cdot,X_i \leq x \}$ is a point process and the second last inequality is due to a change of variables. We denote 
\begin{equation}
\label{eq:obj1}
Q_n(z,k) = -kz + \int \mathbb{L}(u,z)d\hat{N}_n
\end{equation} 
as the sample objective function. Then 
\begin{equation*}
\begin{aligned}
& (-\widehat{Z}_n(k_0),-\widehat{Z}_n(mk_0),-\widehat{Z}_n(k_1),\cdots,-\widehat{Z}_n(k_L)) \\
= & \argmin_{z_0,\tilde{z}_0,z_1,\cdots,z_L} Q_n(z_0,k_0) + Q_n(\tilde{z}_0,mk_0) + \sum_{l=1}^LQ_n(z_l,k_l)
\end{aligned}
\end{equation*}

We first derive the limit of the sample objective function 
\begin{equation}
\label{eq:Qsum}
Q_n(z_0,k_0) + Q_n(\tilde{z}_0,mk_0) + \sum_{l=1}^LQ_n(z_l,k_l)
\end{equation}
point-wise in $(z_0,\tilde{z}_0,z_1,\cdots,z_L)$. Since the check function $\ell_\tau(u)$ and thus the sample objective function are convex, the point-wise convergence in $(z_0,\tilde{z}_0,z_1,\cdots,z_L)$ is sufficient for the uniform convergence in $(z_0,\tilde{z}_0,z_1,\cdots,z_L)$. Given the uniform convergence of the sample objective function, in the second step we show that the limiting objective function has a unique minimizer $$(-Z_\infty(k_0),-Z_\infty(mk_0),-Z_\infty(k_1),\cdots,-Z_\infty(k_L))$$ with probability one. Then, by Lemma \ref{lem:convexity}, we have 
$$(\widehat{Z}_n(k_0),\widehat{Z}_n(mk_0),\widehat{Z}_n(k_1),\cdots,\widehat{Z}_n(k_L)) \convD (Z_\infty(k_0),Z_\infty(mk_0),Z_\infty(k_1),\cdots,Z_\infty(k_L)).$$

We focus on deriving the limit of $Q_n(z,k)$ in \eqref{eq:obj1} with generic $(z,k)$ such that $k$ is not an integer. Then, the limit of \eqref{eq:Qsum} is just the sum of the limits of each term in it. 

For the second term of $Q_n(z,k)$ in \eqref{eq:obj1}, we can show that the point process $\hat{N}_n(\cdot)$ weakly converges to $N(\cdot)$, a Poisson random measure with mean measure $\mu([a,b]) = \eta^{-1}(b) - \eta^{-1}(a)$. In addition, note that both $\hat{N}_n(\cdot)$ and $N(\cdot)$ are random measures on $\Re^+ = [0,\infty)$ because $Y_i \leq q(1)$ for any $i\geq 1$. Then for any fixed $z \geq 0$ and $u \in \Re^+$, $|\mathbb{L}(u,z)|$ is bounded by $z$, vanishes for $u\geq z$, and is continuous in $u$. By the continuous mapping theorem, we have, point-wise in $z$,
$$\int \mathbb{L}(u,z)d\hat{N}_n \convD \int \mathbb{L}(u,z)dN.$$
Now we show
$$\hat{N}_n(\cdot) \convD N(\cdot).$$
Let $T_i = \alpha_n(Y_i - q(1))$. By \citet[][Lemma 9.3 and 9.4]{CH05}, it suffices to show that, for any $0 \leq a < b <\infty$,
$$n\mathbb{P}(-T_i \in [a,b],X_i \leq x) \rightarrow \eta^{-1}(b)-\eta^{-1}(a).$$
Note that $F(y/x) = \mathbb{P}(Y \leq y|X \leq x)$ and
$$\frac{1}{np_0} = 1 - F(q(1) - \frac{1}{\alpha_n}/x).$$
Then,
\begin{equation*}
\begin{aligned}
n\mathbb{P}(-T_i \in [a,b],X_i \leq x) = & np_0 \mathbb{P}(-T_i \in [a,b]|X_i \leq x) \\
= & \mathbb{P}(Y_i \in [q(1) - \frac{b}{\alpha_n},q(1) - \frac{a}{\alpha_n}]|X_i\leq x)/(1-F(q(1) - \frac{1}{\alpha_n}/x)) \\
= & \frac{F(q(1) - \frac{a}{\alpha_n}/x) - F(q(1) - \frac{b}{\alpha_n}/x)}{1-F(q(1) - \frac{1}{\alpha_n}/x)} \\
\rightarrow & \eta^{-1}(b) - \eta^{-1}(a),
\end{aligned}
\end{equation*}
where the last convergence follows Assumption \ref{ass:ev}. By \citet[][Propositions 3.7 and 3.8]{Resnick87}, $N(\cdot)$ can be written as $\sum_{i=1}^\infty \mathds{1}\{\mathcal{J}_i \in \cdot\}$, where $\mathcal{J}_i = (\sum_{j=1}^i \mathcal{E}_i)^{-\xi_0}$ and $\{\mathcal{E}_i\}_{i \geq 1}$ is a sequence of i.i.d. standard exponential random variables. Therefore, the sample objective function will converge to
$$Q_\infty(z,k) = - kz + \int \mathbb{L}(u,z)dN = -kz + \sum_{i=1}^\infty \mathbb{L}(\mathcal{J}_i,z)$$
weakly and uniformly over $z \in \Re^+$.

In addition, from the first-order condition of the limit objective function, we have  
\begin{equation*}
\begin{aligned}
& (-Z_\infty(k_0),-Z_\infty(mk_0),-Z_\infty(k_1),\cdots,-Z_\infty(k_L)) \\
= & \argmin_{z_0,\tilde{z}_0,z_1,\cdots,z_L} Q_\infty(z_0,k_0) + Q_\infty(\tilde{z}_0,mk_0) + \sum_{l=1}^LQ_\infty(z_l,k_l).
\end{aligned}
\end{equation*}
This establishes \eqref{eq:Zhat}. 

\textbf{Step 2:}\\
By \eqref{eq:Zhat}, we have
\begin{equation*}
\begin{aligned}
\begin{pmatrix}
\hat{\alpha}_n(\hat{q}_n(\hat{\tau}_{n1}) - q(1)) \\
\vdots  \\
\hat{\alpha}_n(\hat{q}_n(\hat{\tau}_{nL}) - q(1))\\
\end{pmatrix} \convD \begin{pmatrix}
Z_\infty(k_1)/(Z_\infty(k_0) - Z_\infty(mk_0)) \\
\vdots \\
Z_\infty(k_L)/(Z_\infty(k_0) - Z_\infty(mk_0)) \\
\end{pmatrix} =
\begin{pmatrix}
\tilde{Z}_\infty(k_1) \\
\vdots \\
\tilde{Z}_\infty(k_L)\\
\end{pmatrix}.
\end{aligned}
\end{equation*}
The denominator $Z_\infty(k_0) - Z_\infty(mk_0)$ is nonzero because $Z_\infty(k_0) = -\mathcal{J}_{h_0}$ and $Z_\infty(mk_0) = -\mathcal{J}_{h_0'}$ for $h_0 \in (k_0,k_0+1)$ and $h_0' \in (mk_0,mk_0+1)$, respectively, and by Assumption \ref{ass:m}, $h_0 \neq h_0'$ because $mk_0 > k_0+1$.

\section{Proof of Theorem \ref{thm:simmain}}
\label{sec:pf_simmain}
By Theorem \ref{thm:feasible}, we have 
\begin{align*}
(\tilde{Z}_n(k_1),\cdots,\tilde{Z}_n(k_L)) \convD (\tilde{Z}_\infty(k_1),\cdots,\tilde{Z}_\infty(k_L)) = O_p(1). 
\end{align*}
Therefore, for any $\eps>0$, we can choose a constant $M$ sufficiently large such that 
\begin{align*}
\mathbb{P}\left((\tilde{Z}_n(k_1),\cdots,\tilde{Z}_n(k_L)) \in [-M,M]^L\right) \geq 1-\eps. 
\end{align*}
It suffices to show that 
$$\theta_n^{BE}(z_{1n},\cdots,z_{Ln};\hat{\xi}) \convP \theta_\infty^{BE}(z_1,\cdots,z_L),$$ 
where both $(z_{1n},\cdots,z_{Ln})$ and $(z_{1},\cdots,z_{L})$ are two deterministic sequences that belong to $[-M,M]^L$ and 
\begin{align*}
(z_{1n},\cdots,z_{Ln}) \rightarrow (z_{1},\cdots,z_{L}). 
\end{align*}

Note 
\begin{align*}
& \theta_n^{BE}(z_{1n},\cdots,z_{Ln};\hat{\xi}) \\
= & \argmin_z Q_n(z,z_{1n},\cdots,z_{Ln};\hat{\xi}) \\
= & \argmin_z \frac{1}{\hat{\sigma}}\int \int \ell(z-v)f(z_{1n}-v,\cdots,z_{Ln}-v;\xi)\pi(q(1) + v/\hat{\alpha}_n)\phi( \frac{\xi -\hat{\xi}}{\hat{\sigma}})1\{\xi \in \Gamma\}d\xi dv \\
= & \argmin_z \int \int \ell(z-v)f(z_{1n}-v,\cdots,z_{Ln}-v;\hat{\xi}+u\hat{\sigma})\pi(q(1) + v/\hat{\alpha}_n)\phi(u)1\{u \in \Gamma_n\}du dv,
\end{align*}
where $\Gamma_n = (\Gamma - \hat{\xi})/\hat{\sigma}$. In addition, we have 
\begin{align}
\label{eq:dcp}
& \int \int \ell(z-v)f(z_{1n}-v,\cdots,z_{Ln}-v;\hat{\xi}+u\hat{\sigma})\pi(q(1) + v/\hat{\alpha}_n)\phi(u)1\{u \in \Gamma_n\}du dv \notag \\
= & C_n \int \ell(z-v)f(z_{1n}-v,\cdots,z_{Ln}-v;\hat{\xi})\pi(q(1) + v/\hat{\alpha}_n) dv \notag \\
& + \int \int \ell(z-v)(f(z_{1n}-v,\cdots,z_{Ln}-v;\hat{\xi}+u\hat{\sigma})-f(z_{1n}-v,\cdots,z_{Ln}-v;\hat{\xi})) \notag \\
& \times \pi(q(1) + v/\hat{\alpha}_n) \phi(u)1\{u \in \Gamma_n\}dvdu,
\end{align}
where $C_n  = \int_{\Gamma_n} \phi(u)du \rightarrow 1$ as $n \rightarrow \infty$. By Assumption \ref{ass:rho}.4, $f(z_{1},\cdots,z_{L};\xi)$ is continuous in all its arguments, $\pi(q(1) + v/\hat{\alpha}_n) \rightarrow \pi(q(1))$. Therefore, point-wise in $v$, 
\begin{align*}
C_n \ell(z-v)f(z_{1n}-v,\cdots,z_{Ln}-v;\hat{\xi})\pi(q(1) + v/\hat{\alpha}_n) \convP \ell(z-v)f(z_{1n}-v,\cdots,z_{Ln}-v;\xi_0)\pi(q(1)). 
\end{align*} 
In addition, we have $\mathbb{P}(\hat{\xi} \in \Gamma) \geq 1-\eps$ as $n$ being sufficiently large. Therefore, by Assumption \ref{ass:rho}.4 and with probability greater than $1-\eps$, 
\begin{align*}
\int |\ell(z-v)f(z_{1n}-v,\cdots,z_{Ln}-v;\hat{\xi})\pi(q(1) + v/\hat{\alpha}_n)|dv \lesssim \int |\ell(z-v)|H_{1M}(v)dv < \infty. 
\end{align*}
By the dominated convergence theorem, we have, point-wise in $z$,  
\begin{align*}
& C_n \int \ell(z-v)f(z_{1n}-v,\cdots,z_{Ln}-v;\hat{\xi})\pi(q(1) + v/\hat{\alpha}_n) \{v \in \Omega_n\}dv \\
\convP & \int \ell(z-v)f(z_{1}-v,\cdots,z_{L}-v;\xi_0)\pi(q(1)) dv \equiv Q_\infty(z,z_1,\cdots,z_L). 
\end{align*}
Let $\Gamma_n' = [-\hat{\sigma}^{-1/2},\hat{\sigma}^{-1/2}]$. For the second term on the RHS of \eqref{eq:dcp}, we have 
\begin{align*}
& \biggl \vert \int \int \ell(z-v)(f(z_{1n}-v,\cdots,z_{Ln}-v;\hat{\xi}+u\hat{\sigma})-f(z_{1n}-v,\cdots,z_{Ln}-v;\hat{\xi})) \notag \\
& \times \pi(q(1) + v/\hat{\alpha}_n) \phi(u)1\{u \in \Gamma_n\}1\{v \in \Omega_n\}dvdu \biggr\vert \\
\leq &  \int \int |\ell(z-v)(f(z_{1n}-v,\cdots,z_{Ln}-v;\hat{\xi}+u\hat{\sigma})-f(z_{1n}-v,\cdots,z_{Ln}-v;\hat{\xi}))| \notag \\
& \times \pi(q(1) + v/\hat{\alpha}_n) \phi(u)1\{u \in \Gamma_n'\}dvdu \\
& + \int \int |\ell(z-v)(f(z_{1n}-v,\cdots,z_{Ln}-v;\hat{\xi}+u\hat{\sigma})-f(z_{1n}-v,\cdots,z_{Ln}-v;\hat{\xi}))| \notag \\
& \times \pi(q(1) + v/\hat{\alpha}_n) \phi(u)(1-1\{u \in \Gamma_n'\})dvdu \\ 
\leq & \hat{\sigma}\int \int |\ell(z-v)|H_{2M}(v)|u|\phi(u)1\{u \in \Gamma_n'\}dudv + 2 \int \int |\ell(z-v)| H_1(v) \phi(u)(1-1\{u \in \Gamma_n'\})dudv \\
\convP & 0, 
\end{align*}
where the last inequality is due to Assumption \ref{ass:rho}.4 and the convergence in the last line holds because $\hat{\sigma} \convP 0$ and that 
\begin{align*}
\int \phi(u)(1-1\{u \in \Gamma_n'\})du \convP 0. 
\end{align*}

Therefore, point-wise in $z$, 
\begin{align*}
Q_n(z,z_{1n},\cdots,z_{Ln};\hat{\xi}) \convP Q_\infty(z,z_1,\cdots,z_L). 
\end{align*}

In addition, since $\ell(\cdot)$ is convex in $z$, so be $Q_n(\cdot;\hat{\xi})$ and $Q_\infty(\cdot)$. In view of Lemma \ref{lem:convexity}, we have verified (i) and assumed (ii) and (iii) in Assumption \ref{ass:rho}.7. Therefore, by Lemma \ref{lem:convexity},
\begin{equation*}
\theta_n^{BE}(z_{1n},\cdots,z_{Ln};\hat{\xi}) \convP \theta_\infty^{BE}(z,z_{1},\cdots,z_{L})
\end{equation*}	
where $\theta_n^{BE}(\cdot)$ and $\theta_\infty^{BE}(\cdot)$ are defined in \eqref{eq:gamman} and \eqref{eq:gamma}, respectively. Since the sequence $(z_{1n},\cdots,z_{Ln})$ is arbitrary, we have
$$\theta_n^{BE}(z_{1},\cdots,z_{L};\hat{\xi}) \convP \theta_\infty^{BE}(z_{1},\cdots,z_{L})$$
uniformly over $(z_{1},\cdots,z_{L})$ in any compact subset of the joint support of $(\tilde{Z}_\infty(k_1),\cdots,\tilde{Z}_\infty(k_L))$. In addition, we note that 	
\begin{equation*}
\begin{aligned}
\begin{pmatrix}
\tilde{Z}_n(k_1) \\
\vdots  \\
\tilde{Z}_n(k_L)\\
\end{pmatrix} \convD \begin{pmatrix}
\tilde{Z}_\infty(k_1) \\
\vdots \\
\tilde{Z}_\infty(k_L)\\
\end{pmatrix}.
\end{aligned}
\end{equation*}
Therefore, by the continuous mapping theorem,

\begin{equation*}
\begin{aligned}
\hat{Z}_n^{BE} \equiv \theta_n^{BE}(\tilde{Z}_n(k_1),\cdots,\tilde{Z}_n(k_L);\hat{\xi}) \convD Z_\infty^{BE} \equiv \theta_\infty^{BE}(\tilde{Z}_\infty(k_1),\cdots,\tilde{Z}_\infty(k_L)).
\end{aligned}
\end{equation*}

This concludes the proof.

\section{Proof of Theorem \ref{thm:optimal}}
\label{sec:pf_optimal}
%
%
First, the proof of Theorem \ref{thm:simmain} implies, uniformly over $(z_1,\cdots,z_L) \in [-M,M]^L$, 
$$ \theta_n^{BE}(z_{1},\cdots,z_{L};\hat{\xi}) \convP \theta_\infty^{BE}(z_1,\cdots,z_L),$$
where $\theta_\infty^{BE}(z_1,\cdots,z_L)$ is defined in \eqref{eq:gamma}. In addition, by Assumptions \ref{ass:rho}.3 and \ref{ass:rho}.5, with probability approaching one, 
$$\sup_{v \in K_t}\ell(\theta_n^{BE}(z_1,\cdots,z_L;\hat{\xi})-v)f(z_1-v,\cdots,z_L-v;\xi_0)$$	
is dominated by 
\begin{align*}
C(\sum_{l=1}^L|z_l^{d_2}|+t)^{d_1} H_{3t}(z_1,\cdots,z_L),
\end{align*}
which is an integrable function w.r.t. $(z_1,\cdots,z_L)$ for fixed $t$. Therefore, by the dominated convergence theorem, as $n \rightarrow \infty$
\begin{equation}
\begin{aligned}
& \int_{\Re^L} \ell(\theta_n^{BE}(z_1,\cdots,z_L;\hat{\xi})-v)f(z_1-v,\cdots,z_L-v;\xi_0)dz_1\cdots dz_L \\
\convP &  \int_{\Re^L} \ell(\theta_\infty^{BE}(z_1,\cdots,z_L)-v)f(z_1-v,\cdots,z_L-v;\xi_0)dz_1\cdots dz_L.
\label{eq:R1}
\end{aligned}
\end{equation}
By \eqref{eq:gamma} and a change of variable argument, we have, for any $v$,
$$\theta^{BE}_\infty(z_1,\cdots,z_L)-v = \theta_\infty^{BE}(z_1-v,\cdots,z_L-v).$$
Furthermore, by construction, $f(\cdot;\xi_0)$ is the joint PDF of $(\tilde{Z}_\infty(k_1),\cdots,\tilde{Z}_\infty(k_L))$. Therefore,
\begin{equation*}
\begin{aligned}
\text{the RHS of \eqref{eq:R1}} = & \int_{\Re^L} \ell(\theta_\infty^{BE}(z_1,\cdots,z_L))f(z_1,\cdots,z_L;\xi_0)dz_1\cdots dz_L\\
= & \mathbb{E}\ell(\theta_\infty^{BE}(\tilde{Z}_\infty(k_1),\cdots,\tilde{Z}_\infty(k_L))) = \mathbb{E}\ell(Z_\infty^{BE}),
\end{aligned}
\end{equation*}
where the last equality holds because $Z_\infty^{BE} = \theta_\infty^{BE}(Z_\infty(k_1),\cdots,Z_\infty(k_L))$. Then, we have, for every fixed $t$, 
\begin{align*}
& \int_{K_t} \int_{\Re^L} \ell(\theta_n^{BE}(z_1,\cdots,z_L;\hat{\xi})-v)f(z_1-v,\cdots,z_L-v;\xi_0)dz_1\cdots dz_Ldv/\Lambda(K_t) \\
\convP & \int_{K_t}  \mathbb{E}\ell(Z_\infty^{BE})dv/\Lambda(K_t) = \mathbb{E}\ell(Z_\infty^{BE})
\end{align*}
$$,$$

Taking $\limsup_{t \rightarrow \infty}$ on both sides, we have
$$AAR_{\ell}(\{\theta_n^{BE}\}) = \mathbb{E}\ell(Z_\infty^{BE}).$$

To prove the second result, for each $t \geq 1$, we denote $\tilde{q}_{n,t}^{BE}$ as the quasi-Bayesian estimator with prior $\pi(\overline{q}) = \mathds
{1}\{\hat{\alpha}_n(\overline{q}-q(1)) \in K_t \}$, i.e.,
$$\hat{\alpha}_n (\tilde{q}_{n,t}^{BE} -q(1))= \tilde{\theta}^{BE}_t(\tilde{Z}_n(k_1),\cdots,\tilde{Z}_n(k_1);\hat{\xi}),$$
where $\tilde{\theta}^{BE}_t(z_1,\cdots,z_L)$ is defined in \eqref{eq:thetatilde}. Next, we aim to show
\begin{eqnarray}
\limsup_{t \rightarrow \infty }\limsup_{n \rightarrow \infty} AR_{\ell,K_t}(\tilde{\theta}_t^{BE}) = \mathbb{E}\ell(Z_\infty^{BE}).
\label{eq:RK}
\end{eqnarray}

Note that, 
\begin{equation}
\begin{aligned}
& AR_{\ell,K_t}(\tilde{\theta}_{t}^{BE}(z_1,\cdots,z_L;\hat{\xi})) \\
= &  \int_{-t}^t\int_{\Re^L}\ell(\tilde{\theta}^{BE}_{t}(z_1,\cdots,z_L;\hat{\xi})-v)f(z_1-v,\cdots,z_L-v;\hat{\xi})dz_1 \cdots dz_Ldv/2t \\
= & \int_{-1}^{1}\int_{\Re^L}\ell(\tilde{\theta}^{BE}_{t}(z_1,\cdots,z_L;\hat{\xi})-tu)f(z_1-tu,\cdots,z_L-tu;\hat{\xi})dz_1 \cdots dz_Ldu/2 \\
\convP & \int_{-1}^{1}\int_{\Re^L}\ell(\tilde{\theta}^{BE}_{t}(z_1+tu,\cdots,z_L+tu;\xi_0)-tu)f(z_1,\cdots,z_L;\xi_0)dz_1 \cdots dz_Ldu/2,
\label{eq:RK2}
\end{aligned}
\end{equation}
where the last convergence follows the same argument in \eqref{eq:R1}. By the definition of $\tilde{\theta}_t^{BE}$ in \eqref{eq:thetatilde},
\begin{equation*}
\begin{aligned}
& \tilde{\theta}_t^{BE}(w_1+tu,\cdots,w_L+tu;\xi_0) - tu \\
= & \argmin_\gamma \int_{-t}^t\ell(\gamma + tu - v)f(w_1+tu-v,\cdots,w_L+tu-v;\xi_0)dv \\
= & \argmin_\gamma \int_\Re \mathds{1}\{v \in (t-tu,-t-tu)\}\ell(\gamma - v)f(w_1-v,\cdots,w_L-v;\xi_0)dv.
\end{aligned}
\end{equation*}
Since $ u \in (-1,1)$, as $t \rightarrow \infty$, $\mathds{1}\{v \in (t-tu,-t-tu)\} \uparrow 1$. Therefore, by the monotone convergence theorem, point-wise in $\gamma$, 
\begin{equation*}
\begin{aligned}
& \int_\Re \mathds{1}\{v \in (t-tu,-t-tu)\}\ell(\gamma - v)f(w_1-v,\cdots,w_L-v;\xi_0)dv \\
\rightarrow & \int_\Re \ell(\gamma-v)f(w_1-v,\cdots,w_L-v;\xi_0)dv.
\end{aligned}
\end{equation*}
Then, by Lemma \ref{lem:convexity}, as $t \rightarrow \infty$
\begin{equation}
\begin{aligned}
\tilde{\theta}_t^{BE}(z_1+tu,z_L+tu;\xi_0) - tu  \rightarrow \theta_\infty^{BE}(z_1,\cdots,z_L).
\label{eq:-tu2}
\end{aligned}
\end{equation}
Following \eqref{eq:RK2}, in order to show \eqref{eq:RK}, it suffices to show, as $t\rightarrow \infty$
\begin{equation*}
\begin{aligned}
& \int_{-1}^{1}\int_{\Re^L}\biggl|\ell(\tilde{\theta}^{BE}_{t}(z_1+tu,\cdots,z_L+tu;\xi_0)-tu) - \ell(\theta_\infty^{BE}(z_1,\cdots,z_L))\biggr| \\
& \times f(z_1,\cdots,z_L;\xi_0)dz_1 \cdots dz_Ldu/2 \\
= & \int_{-1}^{1}\int_{\Re^L}\biggl[\ell(\tilde{\theta}^{BE}_{t}(z_1+tu,\cdots,z_L+tu;\xi_0)-tu) - \ell(\theta_\infty^{BE}(z_1,\cdots,z_L))\biggr]^- \\
& \times f(z_1,\cdots,z_L;\xi_0)dz_1 \cdots dz_Ldu/2 \\
& + \int_{-1}^{1}\int_{\Re^L}\biggl[\ell(\tilde{\theta}^{BE}_{t}(z_1+tu,\cdots,z_L+tu;\xi_0)-tu) - \ell(\theta_\infty^{BE}(z_1,\cdots,z_L))\biggr]^+ \\
& \times f(z_1,\cdots,z_L;\xi_0)dz_1 \cdots dz_Ldu/2 \\
= & I_t + II_t \rightarrow 0.
\end{aligned}
\end{equation*}
For $I_t$, we have
$$\biggl[\ell(\tilde{\theta}^{BE}_{t}(z_1+tu,\cdots,z_L+tu;\xi_0)-tu) - \ell(\theta_\infty^{BE}(z_1,\cdots,z_L))\biggr]^- \leq  \ell(\theta_\infty^{BE}(z_1,\cdots,z_L)) $$
which, by Assumption \ref{ass:rho}.4, is integrable w.r.t. $f(z_1,\cdots,z_L;\xi_0)1\{|u|<1\}dz_1\cdots dz_Ldu$. Therefore, by \eqref{eq:-tu2} and the dominated convergence theorem, we have $I_t \rightarrow 0.$

In addition, by \eqref{eq:thetatilde},
\begin{equation*}
\begin{aligned}
& \int_{-1}^1\int_{\Re^L}\ell(\tilde{\theta}^{BE}_{t}(z_1+tu,\cdots,z_L+tu;\xi_0)-tu)f(z_1,\cdots,z_L;\xi_0)dz_1\cdots dz_L du \\
= & \int_{-t}^t\int_{\Re^L}\ell(\tilde{\theta}^{BE}_{t}(z_1,\cdots,z_L;\xi_0)-v)f(z_1-v,\cdots,z_L-v;\xi_0)dz_1 \cdots dz_Ldv/2t \\
\leq & \int_{-t}^t\int_{\Re^L}\ell(\theta^{BE}_{\infty}(z_1,\cdots,z_L)-v)f(z_1-v,\cdots,z_L-v;\xi_0)dz_1 \cdots dz_Ldv/2t \\
= & \int_{-1}^1\int_{\Re^L}\ell(\theta^{BE}_{\infty}(z_1,\cdots,z_L))f(z_1,\cdots,z_L;\xi_0)dz_1\cdots dz_L du/2.
\end{aligned}
\end{equation*}
Therefore,
\begin{align*}
& 2(II_t - I_t) \\
= & \int_{-1}^1\int_{\Re^L}\ell(\theta^{BE}_{\infty}(z_1,\cdots,z_L))f(z_1,\cdots,z_L;\xi_0)dz_1\cdots dz_L du  \\
& - \int_{-1}^1\int_{\Re^L}\ell(\tilde{\theta}^{BE}_{t}(z_1+tu,\cdots,z_L+tu;\xi_0)-tu)f(z_1,\cdots,z_L;\xi_0)dz_1\cdots dz_L du \\
\leq & 0, 
\end{align*}
or equivalently, 
$$0 \leq II_t \leq I_t \rightarrow 0.$$

This concludes \eqref{eq:RK}. If there exists a sequence of estimators, denoted as $\{\breve{\theta}_n\}$, such that $\breve{\theta}_n \in \Theta_n$ and it achieves strictly smaller asymptotic average risk than the quasi-Bayesian estimator $\theta_n^{BE}$, then for infinitely many $t$ and $n$,
$$AR_{\ell,K_t}(\breve{\theta}_n) < AR_{\ell,K_t}(\tilde{\theta}_{t}^{BE}).$$
This is a contradiction because, by construction,
$$\tilde{\theta}_{t}^{BE}(\cdot) \in \argmin_{\theta \in \Theta_n}AR_{\ell,K_t}(\theta).$$
This concludes the proof.

\section{Proof of Corollary \ref{cor:posterior_quantile}}
\label{sec:pf_post}
Denote $\hat{Z}^{BE}_n(\tau') = \hat{\alpha}_n(\hat{q}^{BE}(\tau') - q(1) = \theta_n^{BE}(\tilde{Z}_n(k_1),\cdots,\tilde{Z}_n(k_L);\hat{\xi})$. Then we have
$$\mathbb{P}(\hat{q}^{BE}(\tau') > q(1) ) =  \mathbb{P}(\hat{Z}^{BE}_n(\tau') > 0) \rightarrow \mathbb{P}(Z^{BE}_\infty(\tau') > 0).$$
Next, we show
$$\mathbb{P}( Z^{BE}_\infty(\tau') > 0) = \tau'.$$
Suppose not, then there exists a nonzero constant $c$ such that $\mathbb{P}( Z^{BE}_\infty(\tau') > c) = \tau'$ or equivalently, by the first order condition,
$$\mathbb{E}\tilde{\ell}_{\tau'}( Z^{BE}_\infty(\tau') - c) < \mathbb{E}\tilde{\ell}_{\tau'}( Z^{BE}_\infty(\tau')),$$
where the loss function $\tilde{\ell}_{\tau'}(\cdot)$ is defined in Corollary \ref{cor:posterior_quantile}. Similar to the proof of the first result in Theorem \ref{thm:optimal}, we can show $\mathbb{E}\tilde{\ell}_{\tau'}( Z^{BE}_\infty(\tau') - c)$ is the asymptotic average risk for the estimator $\theta_n^{BE}(\cdot;\hat{\xi})  - c$, i.e.,
$$AAR_{\tilde{\ell}_{\tau'}}(\{\theta_n^{BE}(\cdot;\hat{\xi})  - c\}) = \mathbb{E}\tilde{\ell}_{\tau'}( Z^{BE}_\infty(\tau') - c)  < \mathbb{E}\tilde{\ell}_{\tau'}( Z^{BE}_\infty(\tau')) =  AAR_{\tilde{\ell}_{\tau'}}(\{\theta_n^{BE}(\cdot;\hat{\xi})\}).$$
On the other hand, $\theta_n^{BE}(\cdot;\hat{\xi})  - c \in \Theta_n$. Therefore, we reach a contradiction to the second result in Theorem \ref{thm:optimal}. This implies
$$\mathbb{P}( Z^{BE}_\infty(\tau') > 0) = \tau'.$$
Then, for $\tau' < \tau''$
\begin{align*}
& \mathbb{P}(\hat{q}^{BE}(\tau') \leq q(1) \leq \hat{q}^{BE}(\tau'')) \\
= & 1- \mathbb{P}(q(1) > \hat{q}^{BE}(\tau'')~\text{or}~q(1)<\hat{q}^{BE}(\tau') ) \\
= & 1- \mathbb{P}(q(1) > \hat{q}^{BE}(\tau'')) - \mathbb{P}(q(1)<\hat{q}^{BE}(\tau')) + \mathbb{P}(\hat{q}^{BE}(\tau')>q(1) > \hat{q}^{BE}(\tau'')) \\
= & \mathbb{P}( \hat{Z}^{BE}_n(\tau'') > 0) - \mathbb{P}( \hat{Z}^{BE}_n(\tau') > 0) \\
\rightarrow & \mathbb{P}( Z^{BE}_\infty(\tau'') > 0) - \mathbb{P}( Z^{BE}_\infty(\tau') > 0) = \tau'' - \tau',
\end{align*}
where the third equality holds due to the fact that, by construction, $\hat{q}^{BE}(\tau')$, the $\tau'$-th posterior quantile, is less than or equal to $\hat{q}^{BE}(\tau'')$, the $\tau''$-th posterior quantile, as $\tau''>\tau'$.

\section{Proof of Proposition \ref{prop:PDF}	}
\label{sec:pf_PDF}
We consider the CDF evaluated at $(u_1,\cdots,u_L)$ such that $u_1 < u_2, \cdots, < u_L.$  Note that
$$Z_\infty(k) = - \mathcal{J}_{\lceil k \rceil} = -(\gamma_{1}^{\lceil k \rceil}/p)^{-\xi},$$
where $\gamma_{i}^j = \sum_{l=i}^{j}\mathcal{E}_l$. Therefore,
\begin{equation}
\begin{aligned}
& \mathbb{P}(\tilde{Z}_\infty(k_1 ) \leq u_1,\cdots,\tilde{Z}_\infty( k_L ) \leq u_L) \\
= & \mathbb{E}\mathbb{P}(\tilde{Z}_\infty( k_1 ) \leq u_1,\cdots,\tilde{Z}_\infty( k_L ) \leq u_L|\gamma_1^{\lceil k_0 \rceil},\gamma_{1}^{\lceil mk_0 \rceil}) \\
= & \mathbb{E}\mathbb{P}\biggl(\frac{(\gamma_{1}^{\lceil k_1 \rceil})^{-\xi}}{(\gamma_{1}^{\lceil mk_0 \rceil})^{-\xi} - (\gamma_{1}^{\lceil k_0 \rceil})^{-\xi}} \leq u_1,\cdots,\frac{(\gamma_{1}^{\lceil k_L \rceil})^{-\xi}}{(\gamma_{1}^{\lceil mk_0 \rceil})^{-\xi} - (\gamma_{1}^{\lceil k_0 \rceil})^{-\xi}} \leq u_L\biggl|\gamma_1^{\lceil k_0 \rceil},\gamma_{1}^{\lceil mk_0 \rceil}\biggr) \\
= & \mathbb{E}\mathbb{P}\biggl(\gamma_{\lceil mk_0 \rceil+1}^{\lceil k_1 \rceil} \leq [u_1((\gamma_{1}^{\lceil mk_0 \rceil})^{-\xi} - (\gamma_{1}^{\lceil k_0 \rceil})^{-\xi})]^{-1/\xi} - \gamma_{1}^{\lceil mk_0 \rceil},\cdots, \\
& \gamma_{\lceil mk_0 \rceil+1}^{\lceil k_L \rceil} \leq [u_L((\gamma_{1}^{\lceil mk_0 \rceil})^{-\xi} - (\gamma_{1}^{\lceil k_0 \rceil})^{-\xi})]^{-1/\xi} - \gamma_{1}^{\lceil mk_0 \rceil}|\gamma_1^{\lceil k_0 \rceil},\gamma_{1}^{\lceil mk_0 \rceil}\biggr) \\
\label{eq:PDF1}
\end{aligned}
\end{equation}
Notice that
$$(\gamma_{\lceil mk_0 \rceil+1}^{\lceil k_1 \rceil},\cdots,\gamma_{\lceil mk_0 \rceil+1}^{\lceil k_L \rceil}) \indep (\gamma_1^{\lceil k_0 \rceil},\gamma_{1}^{\lceil mk_0 \rceil}).$$
Let $s= \gamma_{1}^{\lceil k_0 \rceil}$, $t = \gamma_{\lceil k_0 \rceil+1}^{\lceil mk_0 \rceil}$, $\tilde{u} = (t+s)^{-\xi} - s^{-\xi}$, respectively. Then,
\begin{equation*}
\begin{aligned}
& \text{ The RHS of \eqref{eq:PDF1}} \\
= & \int \mathbb{P}\biggl(\gamma_{\lceil mk_0 \rceil+1}^{\lceil k_1 \rceil} \leq (u_1\tilde{u}(t,s))^{-1/\xi} - t,\cdots,\gamma_{\lceil mk_0 \rceil+1}^{\lceil k_L \rceil} \leq (u_L\tilde{u}(t,s))^{-1/\xi} - t\biggr)\\
& \times f_{\lceil k_0 \rceil}(s)f_{\lceil mk_0 \rceil -\lceil k_0 \rceil}(t) dsdt.
\end{aligned}
\end{equation*}
Take derivatives w.r.t. $(u_1,\cdots,u_L)$, we obtain that
\begin{equation*}
\begin{aligned}
& f(u_1,\cdots,u_L;\xi) \\
= & \int (-1/\xi)^L \tilde{u}(t,s)^{-L/\xi} \biggl[\prod_{l=1}^Lu_l^{-1/\xi-1}f_{h_l - h_{l-1}}(v_l - v_{l-1})\biggr]  f_{\lceil k_0 \rceil}(s)f_{\lceil mk_0 \rceil -\lceil k_0 \rceil}(t) dsdt,
\end{aligned}
\end{equation*}
where $h_l = \lceil k_l \rceil$ for $L \geq l \geq 1$, $h_0 = \lceil mk_0 \rceil$, $v_l = (u_l\tilde{u}(t,s))^{-1/\xi}$ for $L \geq l \geq 1$, and $v_0 = t$.

\section{The Computation of the Three Existing Methods}
\label{sec:compute-estimators}
We compute the three estimators in the literature  based on the instructions in \cite{DLN17}. The details are listed below. 
\begin{itemize}
	\item Moment frontier estimator (``Mom'')
	\begin{itemize}
		\item The built-in EV index estimator is computed using the function \textbf{rho\_momt\_pick
		}with argument \textbf{method = ``moment"}.
		\item The tuning parameter $k_{n}$ involved in the estimation of the EV index is computed by the
		function \textbf{kopt\_momt\_pick} with \textbf{method = ``moment''} and estimated EV index.
		\item Based on the above estimation, \textbf{dfs\_momt} is used to compute
		the estimator and the corresponding 95\% confidence interval for the production frontier. 
	\end{itemize}
	\item Pickands frontier estimator (``Momt\_pick")
	\begin{itemize}
		\item The built-in EV index is estimated using function \textbf{rho\_momt\_pick} with argument \textbf{method = ``pickands''}.
		\item The tuning parameter $k_{n}$ involved in the estimation of the EV index is computed by the
		function \textbf{kopt\_momt\_pick} with \textbf{method = ``pickands''} and the estimated EV index.
		\item Based on the above estimation, \textbf{dfs\_pick} is used to compute the estimator and the corresponding 95\% confidence interval for the production frontier.\footnote{Based on \cite{DFS10}, the expressions for the asymptotic variance of the Pickands frontier estimator are different depending on whether the EV index is estimated or not. Since we estimate the EV index, we use the expression of $V_2(\rho_x)$ in \citet[Theorem 2.5]{DFS10}. } 	\end{itemize}
	\item Probability-weighted moment frontier estimator ("Pwm")
	\begin{itemize}
		\item The built-in EV index estimator is computed using \textbf{rho\_pwm} with the default arguments.
		\item The tuning parameter $k_n$ involved in the estimation of the EV index is computed by \textbf{mopt\_pwm} with default arguments.
		\item Based on the above estimation, \textbf{dfs\_pwm } is used to compute the estimator and the corresponding 95\% confidence interval for the production frontier is constructed via bootstrap following the procedure described in \cite{DFS12}.\footnote{Although the R package \textbf{npbr} produces the analytical confidence interval for the probability-weighted estimator as established in \cite{DFS12}, we follow the practice in \cite{DFS12} and conduct bootstrap inference. In our simulation study, we find that the bootstrap inference has better performance in terms of coverage rates. }
	\end{itemize}
\end{itemize}

\section{The Numerical Evaluation of the Density $f\left(u_{1},\dots,u_{L};\xi\right)$}
\label{sec:PDF}
In this section, we introduce the procedure to evaluate the value of $f\left(u_{1},\dots,u_{L};\xi\right)$ established in 
Proposition \ref{prop:PDF}. We use the simple
Trapezoid rule to evaluate the integrals with fine grids. The detailed procedure is as follows. 
\begin{itemize}
	\item Let $p^{right}=0.99999$ and $p^{left}=0.00001$. Obtain the $(p^{left},p^{right})$ quantiles of random variables with densities $f_{\lceil k_0 \rceil}\left(s\right)$ and $f_{\lceil mk_0 \rceil-\lceil k_0 \rceil}\left(t\right)$, and denote them as $(Q_{1}^{left}, Q_{1}^{right})$ and $(Q_{2}^{left}, Q_{2}^{right})$, respectively. 
	\item Construct a $I_{1}\times I_{1}$ grid $G$ for the rectangle area $[Q_{1}^{left},Q_{1}^{right}] \times [Q_{2}^{left},Q_{2}^{right}]$. Further denote $g_{i}^{1}=Q_{1}^{left}+i\times\frac{Q_{1}^{right}-Q_{1}^{left}}{I_{1}-1}$,
	$g_{j}^{2}=Q_{2}^{left}+j\times\frac{Q_{2}^{right}-Q_{2}^{left}}{I_{1}-1}$ for $i,j = 0,\cdots,I_1-1$, and 
	\begin{align*}
	\tilde{f}_{i,j}= & \left(-\frac{1}{\xi}\right)^{L}\tilde{u}\left(g_{i}^{1},g_{j}^{2}\right)^{-\frac{L}{\xi}}\left[\prod_{l=1}^{L}u_{l}^{-\frac{1}{\xi}-1}f_{h_{l}-h_{l-1}}\left(v_{l}-v_{l-1}\right)\right]\\
	& \times f_{\lceil k_0 \rceil}\left(g_{i}^{1}\right)f_{\lceil mk_0 \rceil-\lceil k_0 \rceil}\left(g_{j}^{2}\right).
	\end{align*}
	\item Evaluate the density $f\left(u_{1},\dots,u_{L};\xi\right)$ numerically, i.e.,
	\begin{align*}
	\hat{f}\left(u_{1},\dots,u_{L};\xi\right)= & \frac{\left(Q_{1}^{right}-Q_{1}^{left}\right)\left(Q_{2}^{right}-Q_{2}^{left}\right)}{\left(I_{1}-1\right)\left(I_{1}-1\right)}\\
	& \times\sum_{i=0}^{I_{1}-1}\sum_{j=0}^{I_{1}-1}\frac{1}{4}\left[\tilde{f}_{i,j}+\tilde{f}_{i,j+1}+\tilde{f}_{i+1,j}+\tilde{f}_{i+1,j+1}\right].
	\end{align*}
\end{itemize}
For implementation, we let $I_{1}=100$. Based on our simulation experience, such numerical integration is much faster and more accurate
than the usual Monte Carlo method with $100,000$ random draws. 

\section{Some calculation of $\xi$ for production and cost frontiers}\label{sec:cal_xi}
In this section, we show the calculation of $\xi$ for DGP 1 using the first
$f$ and the first $\mathcal{U}$. We also obtain $\xi$ for DGP 1 when we try to
estimate the cost function, based on what we propose in the application.
Throughout this section, we assume all functions are smooth enough, and limits
exist so that we could apply L'Hospital's rule. The calculations here can be easily extended to other DGPs. It turns our $\xi$ are the same for production frontier and cost frontier for all 15 DGPs. We omit the details, due to similarity.

Before the calculations, we show a general result.

Suppose $Y=f\left(  X\right)  \mathcal{U}.$ The density of $\mathcal{U}$ is
$g\left(  u\right)  .$ Let $f^{-1}\left(  y\right)  $ denote the inverse of
$f\left(  x\right)  .$ Furthure $X\sim$Unif$\left(  0,1\right)  .$ Then by the
definition,%
\begin{align*}
F\left(  y|x\right)   &  =P\left(  \left.  f\left(  X\right)  \mathcal{U\leq
}y\right\vert X\leq x\right)  \\
&  =P\left(  \left.  X\leq f^{-1}\left(  \frac{y}{\mathcal{U}}\right)
\right\vert X\leq x\right)  \\
&  =\int_{0}^{1}\int_{0}^{f^{-1}\left(  \frac{y}{u}\right)  \wedge x}\frac
{1}{x}dtg\left(  u\right)  du=\frac{1}{x}\int_{0}^{1}\left(  f^{-1}\left(
\frac{y}{u}\right)  \wedge x\right)  g\left(  u\right)  du\\
&  =\frac{1}{x}\int_{0}^{\frac{y}{f\left(  x\right)  }}xg\left(  u\right)
du+\frac{1}{x}\int_{\frac{y}{f\left(  x\right)  }}^{1}f^{-1}\left(  \frac
{y}{u}\right)  g\left(  u\right)  du\\
&  =\int_{0}^{\frac{y}{f\left(  x\right)  }}g\left(  u\right)  du+\frac{1}%
{x}\int_{\frac{y}{f\left(  x\right)  }}^{1}f^{-1}\left(  \frac{y}{u}\right)
g\left(  u\right)  du
\end{align*}

Further,%
\begin{align*}
&  \lim_{z\rightarrow0}\frac{1-F\left(  f\left(  x\right)  -vz\right)
}{1-F\left(  f\left(  x\right)  -z\right)  }\\
&  =\lim_{z\rightarrow0}\frac{1-\int_{0}^{1-\frac{vz}{f\left(  x\right)  }%
	}g\left(  u\right)  du-\frac{1}{x}\int_{1-\frac{vz}{f\left(  x\right)  }}%
	^{1}f^{-1}\left(  \frac{f\left(  x\right)  -vz}{u}\right)  g\left(  u\right)
	du}{1-\int_{0}^{1-\frac{z}{f\left(  x\right)  }}g\left(  u\right)  du-\frac
	{1}{x}\int_{1-\frac{z}{f\left(  x\right)  }}^{1}f^{-1}\left(  \frac{f\left(
		x\right)  -z}{u}\right)  g\left(  u\right)  du}\\
&  \overset{\text{L'Hospital's rule}}{=}\lim_{z\rightarrow0}\frac{\frac
	{v}{f\left(  x\right)  }g\left(  1-\frac{vz}{f\left(  x\right)  }\right)
	-\frac{1}{x}\frac{v}{f\left(  x\right)  }f^{-1}\left(  f\left(  x\right)
	\right)  g\left(  1-\frac{vz}{f\left(  x\right)  }\right)  -\frac{1}{x}%
	\int_{1-\frac{vz}{f\left(  x\right)  }}^{1}\frac{\partial}{\partial z}%
	f^{-1}\left(  \frac{f\left(  x\right)  -vz}{u}\right)  g\left(  u\right)
	du}{\frac{1}{f\left(  x\right)  }g\left(  1-\frac{z}{f\left(  x\right)
	}\right)  -\frac{1}{x}\frac{1}{f\left(  x\right)  }f^{-1}\left(  f\left(
	x\right)  \right)  g\left(  1-\frac{z}{f\left(  x\right)  }\right)  -\frac
	{1}{x}\int_{1-\frac{z}{f\left(  x\right)  }}^{1}\frac{\partial}{\partial
		z}f^{-1}\left(  \frac{f\left(  x\right)  -z}{u}\right)  g\left(  u\right)
	du}\\
&  =\lim_{z\rightarrow0}\frac{\int_{1-\frac{vz}{f\left(  x\right)  }}^{1}%
	\frac{\partial}{\partial z}f^{-1}\left(  \frac{f\left(  x\right)  -vz}%
	{u}\right)  g\left(  u\right)  du}{\int_{1-\frac{z}{f\left(  x\right)  }}%
	^{1}\frac{\partial}{\partial z}f^{-1}\left(  \frac{f\left(  x\right)  -z}%
	{u}\right)  g\left(  u\right)  du}%
\end{align*}

Therefore,%
\[
\lim_{z\rightarrow0}\frac{1-F\left(  f\left(  x\right)  -vz\right)
}{1-F\left(  f\left(  x\right)  -z\right)  }=\lim_{z\rightarrow0}\frac
{\int_{1-\frac{vz}{f\left(  x\right)  }}^{1}\frac{\partial}{\partial z}%
	f^{-1}\left(  \frac{f\left(  x\right)  -vz}{u}\right)  g\left(  u\right)
	du}{\int_{1-\frac{z}{f\left(  x\right)  }}^{1}\frac{\partial}{\partial
		z}f^{-1}\left(  \frac{f\left(  x\right)  -z}{u}\right)  g\left(  u\right)  du}%
\]
For DGP 1, $Y=X\mathcal{U},$ $\mathcal{U}\sim$Unif$\left(  0,1\right)  ,$
$g\left(  u\right)  =1,$ $f\left(  x\right)  =x$ and $f^{-1}\left(  y\right)
=y.$ Using the above result,%

\[
\lim_{z\rightarrow0}\frac{\int_{1-\frac{vz}{x}}^{1}\frac{\partial}{\partial
		z}\left(  x-vz\right)  u^{-1}du}{\int_{1-\frac{z}{x}}^{1}\frac{\partial
	}{\partial z}\left(  x-z\right)  u^{-1}du}=\lim_{z\rightarrow0}\frac
{-v\int_{1-\frac{vz}{x}}^{1}u^{-1}du}{-\int_{1-\frac{z}{x}}^{1}u^{-1}du}%
=\lim_{z\rightarrow0}\frac{-\frac{v^{2}}{x}\left(  1-\frac{vz}{x}\right)
	^{-1}}{-\frac{1}{x}\left(  1-\frac{z}{x}\right)  ^{-1}}=v^{2}%
\]
Thus, by the definition of $\xi,$%
\[
\xi=-\frac{1}{2}.
\]

Now, we turn to the cost function. Again, we first present a general result,
then we apply it to DGP 1.

Let $\breve{Y}=-X,$ $\breve{X}=-Y$. Let $C\left(  \breve{X}\right)
=-f^{-1}\left(  -\breve{X}\right)  $. Then, by $Y=f\left(  X\right)  U,$%
\[
\breve{Y}=C\left(  \breve{X}U^{-1}\right)  ,\text{ }U\in\left[  0,1\right]  .
\]
We can alternatively let $\tilde{Y}=-X+M,$ $\tilde{X}=-Y+M$, where $M$ is a
large positive constant. Note this transformation gives the same value of
$\xi$ as that from $\breve{X}$ and $\breve{Y}.$ We obtain $\xi$ based on
$\breve{X}$ and $\breve{Y}.$

By definition,%

\begin{align*}
F\left(  \breve{y}|\breve{x}\right)   &  =P\left(  \left.  \breve
{Y}\mathcal{\leq}\breve{y}\right\vert \breve{X}\leq\breve{x}\right)  =P\left(
\left.  -X\mathcal{\leq}\breve{y}\right\vert -Y\leq\breve{x}\right)  \\
&  =P\left(  \left.  X\mathcal{\geq-}\breve{y}\right\vert Y\mathcal{\geq
	-}\breve{x}\right)  \\
&  =P\left(  \left.  X\mathcal{\geq-}\breve{y}\right\vert f\left(  X\right)
U\mathcal{\geq-}\breve{x}\right)  \\
&  =\frac{P\left(  X\mathcal{\geq-}\breve{y},f\left(  X\right)  U\mathcal{\geq
		-}\breve{x}\right)  }{P\left(  f\left(  X\right)  U\mathcal{\geq-}\breve
	{x}\right)  }\\
&  =\frac{1}{P\left(  f\left(  X\right)  U\mathcal{\geq-}\breve{x}\right)
}P\left(  X\mathcal{\geq-}\breve{y},X\mathcal{\geq}f^{-1}\left(  \frac
{-\breve{x}}{\mathcal{U}}\right)  \right)  \\
&  =\frac{1}{P\left(  f\left(  X\right)  U\mathcal{\geq-}\breve{x}\right)
}\int_{0}^{1}\int_{\mathcal{-}\breve{y}\vee f^{-1}\left(  \frac{-\breve{x}%
	}{\mathcal{U}}\right)  }^{1}dtg\left(  u\right)  du\\
&  =\frac{1}{P\left(  f\left(  X\right)  U\mathcal{\geq-}\breve{x}\right)
}\int_{0}^{1}\left(  1-\left(  \mathcal{-}\breve{y}\vee f^{-1}\left(
\frac{-\breve{x}}{u}\right)  \right)  \right)  g\left(  u\right)  du\\
&  =\frac{1}{P\left(  f\left(  X\right)  U\mathcal{\geq-}\breve{x}\right)
}\left[  \int_{0}^{-\frac{\breve{x}}{f\left(  -\breve{y}\right)  }}\left(
1-f^{-1}\left(  \frac{-\breve{x}}{u}\right)  \right)  g\left(  u\right)
du+\left(  1+\breve{y}\right)  \int_{-\frac{\breve{x}}{f\left(  -\breve
		{y}\right)  }}^{1}g\left(  u\right)  du\right]
\end{align*}
\ By the definition of $\xi,$%
\[
\frac{1-F\left(  -f^{-1}\left(  -\breve{x}\right)  -vz\right)  }{1-F\left(
	-f^{-1}\left(  -\breve{x}\right)  -z\right)  }\rightarrow v^{-1/\xi},
\]
and (note $P\left(  f\left(  X\right)  U\mathcal{\geq-}\breve{x}\right)  $ is
cancelled out)%
\begin{align*}
&  \lim_{z\rightarrow0}\frac{1-F\left(  -f^{-1}\left(  -\breve{x}\right)
	-vz\right)  }{1-F\left(  -f^{-1}\left(  -\breve{x}\right)  -z\right)  }\\
&  =\lim_{z\rightarrow0}\frac{1-\int_{0}^{-\frac{\breve{x}}{f\left(
			f^{-1}\left(  -\breve{x}\right)  +vz\right)  }}\left(  1-f^{-1}\left(
	\frac{-\breve{x}}{u}\right)  \right)  g\left(  u\right)  du-\left(
	1-f^{-1}\left(  -\breve{x}\right)  -vz\right)  \int_{-\frac{\breve{x}%
		}{f\left(  f^{-1}\left(  -\breve{x}\right)  +vz\right)  }}^{1}g\left(
	u\right)  du}{1-\int_{0}^{-\frac{\breve{x}}{f\left(  f^{-1}\left(  -\breve
			{x}\right)  +z\right)  }}\left(  1-f^{-1}\left(  \frac{-\breve{x}}{u}\right)
	\right)  g\left(  u\right)  du-\left(  1-f^{-1}\left(  -\breve{x}\right)
	-z\right)  \int_{-\frac{\breve{x}}{f\left(  f^{-1}\left(  -\breve{x}\right)
			+z\right)  }}^{1}g\left(  u\right)  du}\\
&  \overset{\text{L'Hospital's rule}}{=}\lim_{z\rightarrow0}\frac
{v\int_{-\frac{\breve{x}}{f\left(  f^{-1}\left(  -\breve{x}\right)
			+vz\right)  }}^{1}g\left(  u\right)  du}{\int_{-\frac{\breve{x}}{f\left(
			f^{-1}\left(  -\breve{x}\right)  +z\right)  }}^{1}g\left(  u\right)  du}\text{
	(first two derivatives cancelled out, similar to before).}%
\end{align*}

Therefore,%
\[
\lim_{z\rightarrow0}\frac{1-F\left(  -f^{-1}\left(  -\breve{x}\right)
	-vz\right)  }{1-F\left(  -f^{-1}\left(  -\breve{x}\right)  -z\right)  }%
=\lim_{z\rightarrow0}\frac{v\int_{-\frac{\breve{x}}{f\left(  f^{-1}\left(
			-\breve{x}\right)  +vz\right)  }}^{1}g\left(  u\right)  du}{\int%
	_{-\frac{\breve{x}}{f\left(  f^{-1}\left(  -\breve{x}\right)  +z\right)  }%
	}^{1}g\left(  u\right)  du}.
\]
We apply the above result for DGP 1 where $Y=X\mathcal{U},$ $\mathcal{U}\sim
$Unif$\left(  0,1\right)  ,$ $g\left(  u\right)  =1,$ $f\left(  x\right)  =x$
and $f^{-1}\left(  y\right)  =y.$ Then%

\[
\lim_{z\rightarrow0}\frac{v\int_{-\frac{\breve{x}}{f\left(  f^{-1}\left(
			-\breve{x}\right)  +vz\right)  }}^{1}g\left(  u\right)  du}{\int%
	_{-\frac{\breve{x}}{f\left(  f^{-1}\left(  -\breve{x}\right)  +z\right)  }%
	}^{1}g\left(  u\right)  du}=\lim_{z\rightarrow0}\frac{v\int_{\frac{-\breve{x}%
		}{-\breve{x}+vz}}^{1}du}{\int_{\frac{-\breve{x}}{-\breve{x}+z}}^{1}du}%
=\lim_{z\rightarrow0}\frac{v^{2}\frac{-\breve{x}}{\left(  -\breve
		{x}+vz\right)  ^{2}}}{\frac{-\breve{x}}{\left(  -\breve{x}+z\right)  ^{2}}%
}=v^{2}%
\]
So%
\[
\xi=-\frac{1}{2}.
\]

\section{Additional Simulation Results}\label{sec:moresimul}

\subsection{Addition Simulation Results}
\label{sec:addsim}
\setcounter{table}{4}
In this section, we report simulation results for DGPs$(1,2),(2,2),(3,2),(1,3),\cdots,(3,5)$ in Tables \ref{tabledgp4}--\ref{tabledgp15}. 

\begin{table}[H]
	\centering{}\centering {\caption{DGP(1,2)}
		\label{tabledgp4} }%
	\begin{adjustbox}{max width=\textwidth}
\begin{tabular}{r|ccc|ccc}
			\hline\hline			
			\multicolumn{7}{c}{Panel A: $x = 1.5$}  \\\hline
			& \multicolumn{3}{c|}{Quasi-Bayesian} & \multicolumn{3}{c}{Pickands}\\\hline
			& $k_L'$ & $k_L$ & $k_L^{''}$ & Mom & Momt-pick & Pwm \\
			\hline
			$n=500$& 0.9680 & 0.9760 & 0.9710 & 0.8670 & 0.8880 & 0.9740 \\
			$n p_0=125$& (0.2223) & (0.2222) & (0.2236) & (0.2835) & (0.5236) & (0.7557) \\
			$n=1000$& 0.9500 & 0.9590 & 0.9540 & 0.9070 & 0.8920 & 1.0000 \\
			$n p_0=250$& (0.1236) & (0.1239) & (0.1230) & (0.1406) & (0.2827) & (0.5458) \\
			$n=2000$& 0.9350 & 0.9430 & 0.9340 & 0.9110 & 0.8990 & 1.0000 \\
\hline\hline
		\end{tabular}
	\end{adjustbox}
	
	{ Notes: $k_L'=\min\{\lceil0.10n \hat{p}\rceil, 35\}, k_L=\min\{\lceil0.10n \hat{p}\rceil, 40\}$, and $k_L^{''}=\min\{\lceil0.10n \hat{p}\rceil, 45\}$. The coverage rates and average lengths of the CIs
		(in parentheses) are reported.}
\end{table}

\begin{table}[H]
	\centering{}\centering {\caption{DGP(2,2)}
		\label{tabledgp5} }%
	\begin{adjustbox}{max width=\textwidth}
		\begin{tabular}{r|ccc|ccc}
			\hline\hline			
			\multicolumn{7}{c}{Panel A: $x = 1.5$}  \\\hline
			& \multicolumn{3}{c|}{Quasi-Bayesian} & \multicolumn{3}{c}{Pickands}\\\hline
			& $k_L'$ & $k_L$ & $k_L^{''}$ & Mom & Momt-pick & Pwm \\
			\hline
			$n=500$& 0.9740 & 0.9720 & 0.9710 & 1.0000 & 0.9810 & 0.9380 \\
			$n p_0=125$& (0.6393) & (0.6427) & (0.6415) & (1.1016) & (2.6322) & (1.2612) \\
			$n=1000$& 0.9820 & 0.9810 & 0.9790 & 0.9870 & 0.9560 & 0.9860 \\
			$n p_0=250$& (0.3708) & (0.3692) & (0.3704) & (0.5646) & (1.4879) & (0.9243) \\
			$n=2000$& 0.9740 & 0.9720 & 0.9770 & 0.9560 & 0.9510 & 0.9940 \\
\hline\hline
		\end{tabular}
	\end{adjustbox}
	
	{ Notes: $k_L'=\min\{\lceil0.10n \hat{p}\rceil, 35\}, k_L=\min\{\lceil0.10n \hat{p}\rceil, 40\}$, and $k_L^{''}=\min\{\lceil0.10n \hat{p}\rceil, 45\}$. The coverage rates and average lengths of the CIs
		(in parentheses) are reported.}
\end{table}

\begin{table}[H]
	\centering{}\centering {\caption{DGP(3,2)}
		\label{tabledgp6} }%
	\begin{adjustbox}{max width=\textwidth}
		\begin{tabular}{r|ccc|ccc}
			\hline\hline			
			\multicolumn{7}{c}{Panel A: $x = 1.5$}  \\\hline
			& \multicolumn{3}{c|}{Quasi-Bayesian} & \multicolumn{3}{c}{Pickands}\\\hline
			& $k_L'$ & $k_L$ & $k_L^{''}$ & Mom & Momt-pick & Pwm \\
			\hline
			$n=500$& 0.9740 & 0.9730 & 0.9730 & 1.0000 & 0.9830 & 0.9530 \\
			$n p_0=125$& (0.8416) & (0.8272) & (0.8289) & (1.8886) & (3.8828) & (3.0487) \\
			$n=1000$& 0.9800 & 0.9820 & 0.9780 & 0.9990 & 0.9820 & 0.9960 \\
			$n p_0=250$& (0.4854) & (0.4864) & (0.4847) & (1.1938) & (2.8319) & (2.3919) \\
			$n=2000$& 0.9680 & 0.9710 & 0.9630 & 0.9950 & 0.9800 & 1.0000 \\
\hline\hline
		\end{tabular}
	\end{adjustbox}
	
	{ Notes: $k_L'=\min\{\lceil0.10n \hat{p}\rceil, 35\}, k_L=\min\{\lceil0.10n \hat{p}\rceil, 40\}$, and $k_L^{''}=\min\{\lceil0.10n \hat{p}\rceil, 45\}$. The coverage rates and average lengths of the CIs
		(in parentheses) are reported.}
\end{table}

\begin{table}[H]
	\centering{}\centering {\caption{DGP(1,3)}
		\label{tabledgp7} }%
	\begin{adjustbox}{max width=\textwidth}
		\begin{tabular}{r|ccc|ccc}
			\hline\hline			
			\multicolumn{7}{c}{Panel A: $x = 1.5$}  \\\hline
			& \multicolumn{3}{c|}{Quasi-Bayesian} & \multicolumn{3}{c}{Pickands}\\\hline
			& $k_L'$ & $k_L$ & $k_L^{''}$ & Mom & Momt-pick & Pwm \\
			\hline
			$n=500$& 0.9510 & 0.9470 & 0.9490 & 0.9950 & 0.9610 & 0.9590 \\
			$n p_0=125$& (1.1584) & (1.1606) & (1.1519) & (2.2188) & (5.4216) & (3.4788) \\
			$n=1000$& 0.9580 & 0.9520 & 0.9590 & 0.9780 & 0.9690 & 0.9850 \\
			$n p_0=250$& (0.6645) & (0.6664) & (0.6673) & (1.5651) & (4.2438) & (2.7218) \\
			$n=2000$& 0.9710 & 0.9670 & 0.9620 & 0.9880 & 0.9910 & 1.0000 \\
\hline\hline
		\end{tabular}
	\end{adjustbox}
	
	{ Notes: $k_L'=\min\{\lceil0.10n \hat{p}\rceil, 35\}, k_L=\min\{\lceil0.10n \hat{p}\rceil, 40\}$, and $k_L^{''}=\min\{\lceil0.10n \hat{p}\rceil, 45\}$. The coverage rates and average lengths of the CIs
		(in parentheses) are reported.}
\end{table}

\begin{table}[H]
	\centering{}\centering {\caption{DGP(2,3)}
		\label{tabledgp8} }%
	\begin{adjustbox}{max width=\textwidth}
		\begin{tabular}{r|ccc|ccc}
			\hline\hline			
			\multicolumn{7}{c}{Panel A: $x = 1.5$}  \\\hline
			& \multicolumn{3}{c|}{Quasi-Bayesian} & \multicolumn{3}{c}{Pickands}\\\hline
			& $k_L'$ & $k_L$ & $k_L^{''}$ & Mom & Momt-pick & Pwm \\
			\hline
			$n=500$& 0.9590 & 0.9580 & 0.9600 & 1.0000 & 0.9850 & 0.9670 \\
			$n p_0=125$& (1.2570) & (1.2608) & (1.2656) & (3.0654) & (6.9954) & (4.1390) \\
			$n=1000$& 0.9680 & 0.9720 & 0.9680 & 0.9990 & 0.9910 & 0.9960 \\
			$n p_0=250$& (0.7282) & (0.7302) & (0.7276) & (2.0424) & (5.1007) & (3.2040) \\
			$n=2000$& 0.9690 & 0.9710 & 0.9600 & 0.9880 & 0.9950 & 1.0000 \\
\hline\hline
		\end{tabular}
	\end{adjustbox}
	
	{ Notes: $k_L'=\min\{\lceil0.10n \hat{p}\rceil, 35\}, k_L=\min\{\lceil0.10n \hat{p}\rceil, 40\}$, and $k_L^{''}=\min\{\lceil0.10n \hat{p}\rceil, 45\}$. The coverage rates and average lengths of the CIs
		(in parentheses) are reported.}
\end{table}

\begin{table}[H]
	\centering{}\centering {\caption{DGP(3,3)}
		\label{tabledgp9} }%
	\begin{adjustbox}{max width=\textwidth}
		\begin{tabular}{r|ccc|ccc}
			\hline\hline			
			\multicolumn{7}{c}{Panel A: $x = 1.5$}  \\\hline
			& \multicolumn{3}{c|}{Quasi-Bayesian} & \multicolumn{3}{c}{Pickands}\\\hline
			& $k_L'$ & $k_L$ & $k_L^{''}$ & Mom & Momt-pick & Pwm \\
			\hline
			$n=500$& 0.9630 & 0.9690 & 0.9660 & 0.8640 & 0.9380 & 0.9560 \\
			$n p_0=125$& (0.4289) & (0.4309) & (0.4328) & (0.8002) & (1.3381) & (2.2473) \\
			$n=1000$& 0.9680 & 0.9720 & 0.9640 & 0.9200 & 0.9560 & 0.9970 \\
			$n p_0=250$& (0.2483) & (0.2485) & (0.2496) & (0.4577) & (0.9028) & (1.7155) \\
			$n=2000$& 0.9530 & 0.9350 & 0.9410 & 0.8880 & 0.9700 & 0.9980 \\
\hline\hline
		\end{tabular}
	\end{adjustbox}
	
	{ Notes: $k_L'=\min\{\lceil0.10n \hat{p}\rceil, 35\}, k_L=\min\{\lceil0.10n \hat{p}\rceil, 40\}$, and $k_L^{''}=\min\{\lceil0.10n \hat{p}\rceil, 45\}$. The coverage rates and average lengths of the CIs
		(in parentheses) are reported.}
\end{table}

\begin{table}[H]
	\centering{}\centering {\caption{DGP(1,4)}
		\label{tabledgp10} }%
	\begin{adjustbox}{max width=\textwidth}
		\begin{tabular}{r|ccc|ccc}
			\hline\hline			
			\multicolumn{7}{c}{Panel A: $x = 1.5$}  \\\hline
			& \multicolumn{3}{c|}{Quasi-Bayesian} & \multicolumn{3}{c}{Pickands}\\\hline
			& $k_L'$ & $k_L$ & $k_L^{''}$ & Mom & Momt-pick & Pwm \\
			\hline
			$n=500$& 0.9740 & 0.9740 & 0.9690 & 0.9990 & 0.9790 & 0.9560 \\
			$n p_0=125$& (0.9623) & (0.9610) & (0.9614) & (2.1086) & (4.6819) & (3.2110) \\
			$n=1000$& 0.9710 & 0.9700 & 0.9710 & 0.9960 & 0.9690 & 0.9980 \\
			$n p_0=250$& (0.5612) & (0.5618) & (0.5620) & (1.3366) & (3.3795) & (2.5349) \\
			$n=2000$& 0.9720 & 0.9760 & 0.9720 & 0.9850 & 0.9780 & 0.9990 \\
\hline\hline
		\end{tabular}
	\end{adjustbox}
	
	{ Notes: $k_L'=\min\{\lceil0.10n \hat{p}\rceil, 35\}, k_L=\min\{\lceil0.10n \hat{p}\rceil, 40\}$, and $k_L^{''}=\min\{\lceil0.10n \hat{p}\rceil, 45\}$. The coverage rates and average lengths of the CIs
		(in parentheses) are reported.}
\end{table}

\begin{table}[H]
	\centering{}\centering {\caption{DGP(2,4)}
		\label{tabledgp11} }%
	\begin{adjustbox}{max width=\textwidth}
		\begin{tabular}{r|ccc|ccc}
			\hline\hline			
			\multicolumn{7}{c}{Panel A: $x = 1.5$}  \\\hline
			& \multicolumn{3}{c|}{Quasi-Bayesian} & \multicolumn{3}{c}{Pickands}\\\hline
			& $k_L'$ & $k_L$ & $k_L^{''}$ & Mom & Momt-pick & Pwm \\
			\hline
			$n=500$& 0.9660 & 0.9640 & 0.9680 & 0.9990 & 0.9810 & 0.9620 \\
			$n p_0=125$& (0.2670) & (0.2674) & (0.2670) & (0.6926) & (1.2594) & (2.0065) \\
			$n=1000$& 0.9710 & 0.9770 & 0.9700 & 1.0000 & 0.9830 & 1.0000 \\
			$n p_0=250$& (0.1553) & (0.1556) & (0.1555) & (0.4422) & (0.9558) & (1.6196) \\
			$n=2000$& 0.9670 & 0.9640 & 0.9660 & 1.0000 & 0.9880 & 1.0000 \\
\hline\hline
		\end{tabular}
	\end{adjustbox}
	
	{ Notes: $k_L'=\min\{\lceil0.10n \hat{p}\rceil, 35\}, k_L=\min\{\lceil0.10n \hat{p}\rceil, 40\}$, and $k_L^{''}=\min\{\lceil0.10n \hat{p}\rceil, 45\}$. The coverage rates and average lengths of the CIs
		(in parentheses) are reported.}
\end{table}

\begin{table}[H]
	\centering{}\centering {\caption{DGP(3,4)}
		\label{tabledgp12} }%
	\begin{adjustbox}{max width=\textwidth}
		\begin{tabular}{r|ccc|ccc}
			\hline\hline			
			\multicolumn{7}{c}{Panel A: $x = 1.5$}  \\\hline
			& \multicolumn{3}{c|}{Quasi-Bayesian} & \multicolumn{3}{c}{Pickands}\\\hline
			& $k_L'$ & $k_L$ & $k_L^{''}$ & Mom & Momt-pick & Pwm \\
			\hline
			$n=500$& 0.9360 & 0.9380 & 0.9330 & 0.9900 & 0.9250 & 0.9890 \\
			$n p_0=125$& (0.3203) & (0.3228) & (0.3186) & (0.6860) & (1.5014) & (2.1898) \\
			$n=1000$& 0.9220 & 0.9250 & 0.9260 & 0.9640 & 0.9490 & 1.0000 \\
			$n p_0=250$& (0.1851) & (0.1848) & (0.1848) & (0.4224) & (1.0371) & (1.6840) \\
			$n=2000$& 0.9470 & 0.9450 & 0.9430 & 0.9570 & 0.9690 & 1.0000 \\
\hline\hline
		\end{tabular}
	\end{adjustbox}
	
	{ Notes: $k_L'=\min\{\lceil0.10n \hat{p}\rceil, 35\}, k_L=\min\{\lceil0.10n \hat{p}\rceil, 40\}$, and $k_L^{''}=\min\{\lceil0.10n \hat{p}\rceil, 45\}$. The coverage rates and average lengths of the CIs
		(in parentheses) are reported.}
\end{table}

\begin{table}[H]
	\centering{}\centering {\caption{DGP(1,5)}
		\label{tabledgp13} }%
	\begin{adjustbox}{max width=\textwidth}
		\begin{tabular}{r|ccc|ccc}
			\hline\hline			
			\multicolumn{7}{c}{Panel A: $x = 1.5$}  \\\hline
			& \multicolumn{3}{c|}{Quasi-Bayesian} & \multicolumn{3}{c}{Pickands}\\\hline
			& $k_L'$ & $k_L$ & $k_L^{''}$ & Mom & Momt-pick & Pwm \\
			\hline
			$n=500$& 0.9290 & 0.9250 & 0.9280 & 1.0000 & 0.9640 & 0.9970 \\
			$n p_0=125$& (0.3370) & (0.3358) & (0.3368) & (0.9546) & (1.8920) & (2.5371) \\
			$n=1000$& 0.9040 & 0.8980 & 0.9000 & 0.9950 & 0.9630 & 1.0000 \\
			$n p_0=250$& (0.1934) & (0.1932) & (0.1939) & (0.5695) & (1.3096) & (1.9455) \\
			$n=2000$& 0.9420 & 0.9390 & 0.9370 & 0.9780 & 0.9830 & 1.0000 \\
\hline\hline
		\end{tabular}
	\end{adjustbox}
	
	{ Notes: $k_L'=\min\{\lceil0.10n \hat{p}\rceil, 35\}, k_L=\min\{\lceil0.10n \hat{p}\rceil, 40\}$, and $k_L^{''}=\min\{\lceil0.10n \hat{p}\rceil, 45\}$. The coverage rates and average lengths of the CIs
		(in parentheses) are reported.}
\end{table}

\begin{table}[H]
	\centering{}\centering {\caption{DGP(2,5)}
		\label{tabledgp14} }%
	\begin{adjustbox}{max width=\textwidth}
		\begin{tabular}{r|ccc|ccc}
			\hline\hline			
			\multicolumn{7}{c}{Panel A: $x = 1.5$}  \\\hline
			& \multicolumn{3}{c|}{Quasi-Bayesian} & \multicolumn{3}{c}{Pickands}\\\hline
			& $k_L'$ & $k_L$ & $k_L^{''}$ & Mom & Momt-pick & Pwm \\
			\hline
			$n=500$& 0.9750 & 0.9680 & 0.9680 & 0.8950 & 0.9590 & 0.9620 \\
			$n p_0=125$& (0.1659) & (0.1654) & (0.1658) & (0.4583) & (0.7104) & (1.6608) \\
			$n=1000$& 0.9460 & 0.9560 & 0.9450 & 0.9000 & 0.9480 & 1.0000 \\
			$n p_0=250$& (0.0937) & (0.0937) & (0.0935) & (0.2429) & (0.4537) & (1.3028) \\
			$n=2000$& 0.9490 & 0.9520 & 0.9430 & 0.8540 & 0.9470 & 1.0000 \\
\hline\hline
		\end{tabular}
	\end{adjustbox}
	
	{ Notes: $k_L'=\min\{\lceil0.10n \hat{p}\rceil, 35\}, k_L=\min\{\lceil0.10n \hat{p}\rceil, 40\}$, and $k_L^{''}=\min\{\lceil0.10n \hat{p}\rceil, 45\}$. The coverage rates and average lengths of the CIs
		(in parentheses) are reported.}
\end{table}

\begin{table}[H]
	\centering{}\centering {\caption{DGP(3,5)}
		\label{tabledgp15} }%
	\begin{adjustbox}{max width=\textwidth}
		\begin{tabular}{r|ccc|ccc}
			\hline\hline			
			\multicolumn{7}{c}{Panel A: $x = 1.5$}  \\\hline
			& \multicolumn{3}{c|}{Quasi-Bayesian} & \multicolumn{3}{c}{Pickands}\\\hline
			& $k_L'$ & $k_L$ & $k_L^{''}$ & Mom & Momt-pick & Pwm \\
			\hline
			$n=500$& 0.9660 & 0.9680 & 0.9700 & 0.9990 & 0.9610 & 0.9740 \\
			$n p_0=125$& (0.2897) & (0.2908) & (0.2917) & (0.6908) & (1.3902) & (2.0587) \\
			$n=1000$& 0.9650 & 0.9610 & 0.9640 & 0.9980 & 0.9810 & 0.9990 \\
			$n p_0=250$& (0.1729) & (0.1728) & (0.1729) & (0.450) & (1.0407) & (1.6376) \\
			$n=2000$& 0.9680 & 0.9650 & 0.9540 & 0.9960 & 0.9870 & 1.0000 \\
\hline\hline
		\end{tabular}
	\end{adjustbox}
	
	{ Notes: $k_L'=\min\{\lceil0.10n \hat{p}\rceil, 35\}, k_L=\min\{\lceil0.10n \hat{p}\rceil, 40\}$, and $k_L^{''}=\min\{\lceil0.10n \hat{p}\rceil, 45\}$. The coverage rates and average lengths of the CIs
		(in parentheses) are reported.}
\end{table}

\subsection{Sensitivity of $\pi(\overline{q})$}
\label{sec:simpi}
In this section, we examine the sensitive of the choice of  $\pi(\overline{q})$ in DGPs(1,1) and (2,1). The first and second main columns in Tables \ref{q1priordgp1} and \ref{q1priordgp2} report the results for the quasi-Bayesian method when we use $\pi_1(\overline{q})$ and $\pi_2(\overline{q})$ as the prior respectively, 
where $\pi_1(\overline{q})$ and $\pi_2(\overline{q})$ are normal with mean $\overline{q}^{*}$ and variance 1 and 1.5, respectively.
Note we set $k_0=$number of outliers $+2$, $sp=5$, and  $k_L'=\min\{\lceil0.10n \hat{p}\rceil, 35\}, k_L=\min\{\lceil0.10n \hat{p}\rceil, 40\}$, and $k_L^{''}=\min\{\lceil0.10n \hat{p}\rceil, 45\}$.

\begin{table}[H]
	\caption{Robustness check of $\pi(\overline{q})$ for DGP(1,1)}
	\label{q1priordgp1}\centering{}\centering { } \begin{adjustbox}{max
			width=\textwidth} %
		\begin{tabular}{r|ccc|ccc}
			\hline \hline 
			\multicolumn{7}{c}{Panel A: $x=1.5$}\\
			\hline
			& \multicolumn{3}{c|}{ $\pi_1(\overline{q})$} & \multicolumn{3}{c}{ $\pi_2(\overline{q})$}\\
			\hline
			& $k_L'$  & $k_L$  & $k_L^{''}$  & $k_L'$  & $k_L$  & $k_L^{''}$ \\
			\hline
			$n=500$ & 0.9650 & 0.9670 & 0.9650 & 0.9650 & 0.9670 & 0.9650 \\ 
			$n p_0 = 125$ & (0.5322) & (0.5302) & (0.5325) & (0.5322) & (0.5302) & (0.5325) \\ 
			$n=1000$ & 0.9800 & 0.9810 & 0.9800 & 0.9800 & 0.9810 & 0.9800 \\ 
			$n p_0 = 250$ & (0.3073) & (0.3072) & (0.3061) & (0.3073) & (0.3072) & (0.3061) \\ 
			$n=2000$ & 0.9730 & 0.9710 & 0.9710 & 0.9730 & 0.9710 & 0.9710 \\ 
			$n p_0 = 500$ & (0.2138) & (0.2117) & (0.2120) & (0.2138) & (0.2117) & (0.2120) \\ 
			$n=4000$ & 0.9650 & 0.9580 & 0.9550 & 0.9650 & 0.9580 & 0.9550 \\ 
			$n p_0 = 1000$ & (0.1477) & (0.1467) & (0.1458) & (0.1477) & (0.1467) & (0.1458) \\ 
			\hline \hline 
			\multicolumn{7}{c}{Panel B: $x=3.0$}\\
			\hline
			& \multicolumn{3}{c|}{$\pi_1(\overline{q})$} & \multicolumn{3}{c}{ $\pi_2(\overline{q})$}\\
			\hline
			& $k_L'$  & $k_L$  & $k_L^{''}$  & $k_L'$  & $k_L$  & $k_L^{''}$ \\
			\hline
			$n=500$ & 0.9540 & 0.9550 & 0.9540 & 0.9540 & 0.9550 & 0.9540 \\ 
			$n p_0 = 250$ & (0.5348) & (0.5343) & (0.5360) & (0.5348) & (0.5343) & (0.5360) \\ 
			$n=1000$ & 0.9800 & 0.9750 & 0.9720 & 0.9800 & 0.9750 & 0.9720 \\ 
			$n p_0 = 500$ & (0.3479) & (0.3457) & (0.3423) & (0.3479) & (0.3457) & (0.3423) \\ 
			$n=2000$ & 0.9520 & 0.9530 & 0.9510 & 0.9520 & 0.9530 & 0.9510 \\ 
			$n p_0 = 1000$ & (0.2461) & (0.2425) & (0.2409) & (0.2461) & (0.2425) & (0.2409) \\ 
			$n=4000$ & 0.9390 & 0.9440 & 0.9380 & 0.9390 & 0.9440 & 0.9380 \\ 
			$n p_0 = 2000$ & (0.1715) & (0.1686) & (0.1667) & (0.1715) & (0.1686) & (0.1667) \\ 
			\hline \hline 
			\multicolumn{7}{c}{Panel C: $x=4.5$}\\
			\hline
			& \multicolumn{3}{c|}{ $\pi_1(\overline{q})$} & \multicolumn{3}{c}{ $\pi_2(\overline{q})$}\\
			\hline
			& $k_L'$  & $k_L$  & $k_L^{''}$  & $k_L'$  & $k_L$  & $k_L^{''}$ \\
			\hline
			$n=500$ & 0.9680 & 0.9660 & 0.9630 & 0.9680 & 0.9660 & 0.9630 \\ 
			$n p_0 = 375$ & (0.5136) & (0.5121) & (0.5095) & (0.5136) & (0.5121) & (0.5095) \\ 
			$n=1000$ & 0.9750 & 0.9720 & 0.9730 & 0.9750 & 0.9720 & 0.9730 \\ 
			$n p_0 = 1000$ & (0.3528) & (0.3486) & (0.3465) & (0.3528) & (0.3486) & (0.3465) \\ 
			$n=2000$ & 0.9760 & 0.9720 & 0.9720 & 0.9760 & 0.9720 & 0.9720 \\ 
			$n p_0 = 1500$ & (0.2431) & (0.2395) & (0.2384) & (0.2431) & (0.2395) & (0.2384) \\ 
			$n=4000$ & 0.9830 & 0.9740 & 0.9750 & 0.9830 & 0.9740 & 0.9750 \\ 
			$n p_0 = 3000$ & (0.1719) & (0.1680) & (0.1665) & (0.1719) & (0.1680) & (0.1665) \\ 
			\hline
	\end{tabular}\end{adjustbox}
	
	{ Notes: The coverage rates and average lengths of the CIs (in parentheses)
		are reported.}
\end{table}

\begin{table}[H]
	\caption{Robustness check of $\pi(\overline{q})$ for DGP(2,1)}
	\label{q1priordgp2}\centering{}\centering { } \begin{adjustbox}{max
			width=\textwidth} %
		\begin{tabular}{r|ccc|ccc}
			\hline \hline 
			\multicolumn{7}{c}{Panel A: $x=1.5$}\\
			\hline
			& \multicolumn{3}{c|}{ $\pi_1(\overline{q})$} & \multicolumn{3}{c}{ $\pi_2(\overline{q})$}\\
			\hline
			& $k_L'$  & $k_L$  & $k_L^{''}$  & $k_L'$  & $k_L$  & $k_L^{''}$ \\
			\hline
			$n=500$ & 0.9690 & 0.9710 & 0.9720 & 0.9720 & 0.9730 & 0.9770 \\ 
			$n p_0 = 125$ & (0.8276) & (0.8277) & (0.8376) & (0.8293) & (0.8359) & (0.8269) \\ 
			$n=1000$ & 0.9810 & 0.9820 & 0.9770 & 0.9810 & 0.9810 & 0.9810 \\ 
			$n p_0 = 250$ & (0.4848) & (0.4826) & (0.4803) & (0.4836) & (0.4868) & (0.4861) \\ 
			$n=2000$ & 0.9620 & 0.9610 & 0.9590 & 0.9680 & 0.9710 & 0.9620 \\ 
			$n p_0 = 500$ & (0.3443) & (0.3405) & (0.3371) & (0.3450) & (0.3419) & (0.3384) \\ 
			$n=4000$ & 0.9500 & 0.9550 & 0.9500 & 0.9550 & 0.9480 & 0.9510 \\ 
			$n p_0 = 1000$ & (0.2466) & (0.2447) & (0.2415) & (0.2482) & (0.2459) & (0.2449) \\ 
			\hline \hline 
			\multicolumn{7}{c}{Panel B: $x=3.0$}\\
			\hline
			& \multicolumn{3}{c|}{ $\pi_1(\overline{q})$} & \multicolumn{3}{c}{ $\pi_2(\overline{q})$}\\
			\hline
			& $k_L'$  & $k_L$  & $k_L^{''}$  & $k_L'$  & $k_L$  & $k_L^{''}$ \\
			\hline
			$n=500$ & 0.9790 & 0.9790 & 0.9770 & 0.9840 & 0.9820 & 0.9830 \\ 
			$n p_0 = 250$ & (1.1597) & (1.1670) & (1.1541) & (1.1641) & (1.1572) & (1.1646) \\ 
			$n=1000$ & 0.9720 & 0.9690 & 0.9650 & 0.9700 & 0.9730 & 0.9680 \\ 
			$n p_0 = 500$ & (0.7865) & (0.7687) & (0.7593) & (0.7842) & (0.7705) & (0.7622) \\ 
			$n=2000$ & 0.9470 & 0.9450 & 0.9470 & 0.9510 & 0.9560 & 0.9530 \\ 
			$n p_0 = 1000$ & (0.5727) & (0.5644) & (0.5564) & (0.5749) & (0.5640) & (0.5573) \\ 
			$n=4000$ & 0.9360 & 0.9320 & 0.9360 & 0.9430 & 0.9400 & 0.9380 \\ 
			$n p_0 = 2000$ & (0.4238) & (0.4151) & (0.4128) & (0.4249) & (0.4173) & (0.4122) \\ 
			\hline \hline 
			\multicolumn{7}{c}{Panel C: $x=4.5$}\\
			\hline
			& \multicolumn{3}{c|}{ $\pi_1(\overline{q})$} & \multicolumn{3}{c}{ $\pi_2(\overline{q})$}\\
			\hline
			& $k_L'$  & $k_L$  & $k_L^{''}$  & $k_L'$  & $k_L$  & $k_L^{''}$ \\
			\hline
			$n=500$ & 0.9810 & 0.9820 & 0.9770 & 0.9800 & 0.9760 & 0.9850 \\ 
			$n p_0 = 375$ & (1.3457) & (1.3233) & (1.3391) & (1.3657) & (1.3399) & (1.3423) \\ 
			$n=1000$ & 0.9790 & 0.9760 & 0.9760 & 0.9830 & 0.9790 & 0.9800 \\ 
			$n p_0 = 1000$ & (0.9763) & (0.9581) & (0.9473) & (0.9830) & (0.9680) & (0.9507) \\ 
			$n=2000$ & 0.9680 & 0.9710 & 0.9640 & 0.9690 & 0.9720 & 0.9690 \\ 
			$n p_0 = 1500$ & (0.7086) & (0.6982) & (0.6920) & (0.7097) & (0.6981) & (0.6915) \\ 
			$n=4000$ & 0.9450 & 0.9440 & 0.9380 & 0.9450 & 0.9490 & 0.9470 \\ 
			$n p_0 = 3000$ & (0.5142) & (0.5032) & (0.4980) & (0.5166) & (0.5060) & (0.5003) \\ 
			\hline
	\end{tabular}\end{adjustbox}
	
	{ Notes: The coverage rates and average lengths of the CIs (in parentheses)
		are reported. }
\end{table}

\subsection{Sensitivity to $k_0$}
\label{sec:simk0}
In this section, we examine the sensitivity of the choice of $k_0$ in DGPs(1,1) and (2,1).  The first and second main columns in Tables \ref{k0dgp1} and \ref{k0dgp2} report the results for the quasi-Bayesian method with $k_{0}=$Number of outliers $+3$ and $k_{0}=$ Number of 
outliers $+4$, respectively. Note we set $sp=5$ and  $k_L'=\min\{\lceil0.10n \hat{p}\rceil, 35\}, k_L=\min\{\lceil0.10n \hat{p}\rceil, 40\}$, and $k_L^{''}=\min\{\lceil0.10n \hat{p}\rceil, 45\}$.

\begin{table}[H]
	\caption{Robustness check of $k_{0}$ for DGP(1,1)}
	\label{k0dgp1}%
	\centering{}\centering { } \begin{adjustbox}{max width=\textwidth}
		\begin{tabular}{r|ccc|ccc}
			\hline \hline 
			\multicolumn{7}{c}{Panel A: $x=1.5$}\\
			\hline
			& \multicolumn{3}{c|}{ $k_{0}=3+\sharp(\text{outliers})$} & \multicolumn{3}{c}{ $k_{0}=4+\sharp(\text{outliers})$}\\
			\hline
			& $k_L'$  & $k_L$  & $k_L^{''}$  & $k_L'$  & $k_L$  & $k_L^{''}$ \\
			\hline
			$n=500$ & 0.9870 & 0.9850 & 0.9870 & 0.9830 & 0.9860 & 0.9830 \\ 
			$n p_0 = 125$ & (0.4944) & (0.4972) & (0.4943) & (0.5786) & (0.580) & (0.5802) \\ 
			$n=1000$ & 0.9880 & 0.9880 & 0.9900 & 0.9910 & 0.9870 & 0.9900 \\ 
			$n p_0 = 250$ & (0.3174) & (0.3168) & (0.3173) & (0.3551) & (0.3526) & (0.3524) \\ 
			$n=2000$ & 0.9830 & 0.9850 & 0.9820 & 0.9910 & 0.9890 & 0.9890 \\ 
			$n p_0 = 500$ & (0.2273) & (0.2267) & (0.2263) & (0.2537) & (0.2504) & (0.2487) \\ 
			$n=4000$ & 0.9730 & 0.9710 & 0.9700 & 0.9850 & 0.9820 & 0.9840 \\ 
			$n p_0 = 1000$ & (0.1578) & (0.1563) & (0.1550) & (0.1756) & (0.1733) & (0.1717) \\ 
			\hline \hline 
			\multicolumn{7}{c}{Panel B: $x=3.0$}\\
			\hline
			& \multicolumn{3}{c|}{ $k_{0}=3+\sharp(\text{outliers})$} & \multicolumn{3}{c}{ $k_{0}=4+\sharp(\text{outliers})$}\\
			\hline
			& $k_L'$  & $k_L$  & $k_L^{''}$  & $k_L'$  & $k_L$  & $k_L^{''}$ \\
			\hline
			$n=500$ & 0.9560 & 0.9550 & 0.9540 & 0.9910 & 0.9930 & 0.9940 \\ 
			$n p_0 = 250$ & (0.5405) & (0.5339) & (0.5354) & (0.5722) & (0.5683) & (0.5675) \\ 
			$n=1000$ & 0.9630 & 0.9670 & 0.9620 & 0.9800 & 0.9810 & 0.9820 \\ 
			$n p_0 = 500$ & (0.3776) & (0.3730) & (0.3697) & (0.4062) & (0.4003) & (0.3950) \\ 
			$n=2000$ & 0.9620 & 0.9700 & 0.9650 & 0.9640 & 0.9710 & 0.9690 \\ 
			$n p_0 = 1000$ & (0.2671) & (0.2613) & (0.2582) & (0.2884) & (0.2817) & (0.2769) \\ 
			$n=4000$ & 0.9590 & 0.9610 & 0.9620 & 0.9670 & 0.9660 & 0.9700 \\ 
			$n p_0 = 2000$ & (0.1802) & (0.1779) & (0.1755) & (0.1994) & (0.1941) & (0.1916) \\ 
			\hline \hline 
			\multicolumn{7}{c}{Panel C: $x=4.5$}\\
			\hline
			& \multicolumn{3}{c|}{ $k_{0}=3+\sharp(\text{outliers})$} & \multicolumn{3}{c}{ $k_{0}=4+\sharp(\text{outliers})$}\\
			\hline
			& $k_L'$  & $k_L$  & $k_L^{''}$  & $k_L'$  & $k_L$  & $k_L^{''}$ \\
			\hline
			$n=500$ & 0.9720 & 0.9710 & 0.9690 & 0.9860 & 0.9870 & 0.9860 \\ 
			$n p_0 = 375$ & (0.5363) & (0.5282) & (0.5271) & (0.5721) & (0.5599) & (0.5577) \\ 
			$n=1000$ & 0.9610 & 0.9600 & 0.9600 & 0.9750 & 0.9780 & 0.9790 \\ 
			$n p_0 = 1000$ & (0.3875) & (0.3797) & (0.3756) & (0.4119) & (0.4011) & (0.3965) \\ 
			$n=2000$ & 0.9500 & 0.9470 & 0.9480 & 0.9650 & 0.9680 & 0.9640 \\ 
			$n p_0 = 1500$ & (0.2744) & (0.2687) & (0.2655) & (0.2910) & (0.2844) & (0.2816) \\ 
			$n=4000$ & 0.9420 & 0.9380 & 0.9410 & 0.9440 & 0.9460 & 0.9440 \\ 
			$n p_0 = 3000$ & (0.1897) & (0.1858) & (0.1830) & (0.2045) & (0.1986) & (0.1947) \\
			\hline
		\end{tabular}
	\end{adjustbox}
	
	{ Notes: The coverage rates and average lengths of the CIs (in parentheses)
		are reported. }
\end{table}

\begin{table}[H]
	\caption{Robustness check of $k_{0}$ for DGP(2,1)}
	\label{k0dgp2}%
	\centering{}\centering { } \begin{adjustbox}{max width=\textwidth}
		\begin{tabular}{r|ccc|ccc}
			\hline \hline 
			\multicolumn{7}{c}{Panel A: $x=1.5$}\\
			\hline
			& \multicolumn{3}{c|}{ $k_{0}=3+\sharp(\text{outliers})$} & \multicolumn{3}{c}{ $k_{0}=4+\sharp(\text{outliers})$}\\
			\hline
			& $k_L'$  & $k_L$  & $k_L^{''}$  & $k_L'$  & $k_L$  & $k_L^{''}$ \\
			\hline
			$n=500$ & 0.9770 & 0.9750 & 0.9730 & 0.9670 & 0.9740 & 0.9770 \\ 
			$n p_0 = 125$ & (0.7416) & (0.7415) & (0.7450) & (0.8529) & (0.8477) & (0.8501) \\ 
			$n=1000$ & 0.9840 & 0.9810 & 0.9820 & 0.9880 & 0.9890 & 0.9890 \\ 
			$n p_0 = 250$ & (0.4906) & (0.4899) & (0.490) & (0.5414) & (0.5373) & (0.5355) \\ 
			$n=2000$ & 0.9830 & 0.9840 & 0.9850 & 0.9940 & 0.9940 & 0.9910 \\ 
			$n p_0 = 500$ & (0.3668) & (0.3615) & (0.3594) & (0.4041) & (0.3977) & (0.3932) \\ 
			$n=4000$ & 0.9710 & 0.9690 & 0.9690 & 0.9830 & 0.9860 & 0.9870 \\ 
			$n p_0 = 1000$ & (0.2644) & (0.2615) & (0.2586) & (0.2916) & (0.2873) & (0.2843) \\ 
			\hline \hline 
			\multicolumn{7}{c}{Panel B: $x=3.0$}\\
			\hline
			& \multicolumn{3}{c|}{ $k_{0}=3+\sharp(\text{outliers})$} & \multicolumn{3}{c}{ $k_{0}=4+\sharp(\text{outliers})$}\\
			\hline
			& $k_L'$  & $k_L$  & $k_L^{''}$  & $k_L'$  & $k_L$  & $k_L^{''}$ \\
			\hline
			$n=500$ & 0.9610 & 0.9710 & 0.9690 & 0.9920 & 0.9870 & 0.9870 \\ 
			$n p_0 = 250$ & (1.1746) & (1.1580) & (1.1509) & (1.2520) & (1.2390) & (1.2286) \\ 
			$n=1000$ & 0.9910 & 0.9870 & 0.9860 & 0.9850 & 0.9910 & 0.9880 \\ 
			$n p_0 = 500$ & (0.8485) & (0.8354) & (0.8230) & (0.920) & (0.8957) & (0.8831) \\ 
			$n=2000$ & 0.9760 & 0.9740 & 0.9730 & 0.9860 & 0.9850 & 0.9880 \\ 
			$n p_0 = 1000$ & (0.6072) & (0.5980) & (0.5888) & (0.6650) & (0.6517) & (0.6399) \\ 
			$n=4000$ & 0.9360 & 0.9380 & 0.9370 & 0.9660 & 0.9630 & 0.9700 \\ 
			$n p_0 = 2000$ & (0.4363) & (0.4293) & (0.4249) & (0.4692) & (0.4594) & (0.4521) \\ 
			\hline \hline 
			\multicolumn{7}{c}{Panel C: $x=4.5$}\\
			\hline
			& \multicolumn{3}{c|}{ $k_{0}=3+\sharp(\text{outliers})$} & \multicolumn{3}{c}{ $k_{0}=4+\sharp(\text{outliers})$}\\
			\hline
			& $k_L'$  & $k_L$  & $k_L^{''}$  & $k_L'$  & $k_L$  & $k_L^{''}$ \\
			\hline
			$n=500$ & 0.9690 & 0.9710 & 0.9630 & 0.9840 & 0.9830 & 0.9800 \\ 
			$n p_0 = 375$ & (1.5038) & (1.4763) & (1.4495) & (1.6091) & (1.5643) & (1.530) \\ 
			$n=1000$ & 0.9580 & 0.9590 & 0.9590 & 0.9700 & 0.9740 & 0.9750 \\ 
			$n p_0 = 1000$ & (1.1084) & (1.0815) & (1.0649) & (1.1840) & (1.1445) & (1.1299) \\ 
			$n=2000$ & 0.9540 & 0.9460 & 0.9500 & 0.9600 & 0.9550 & 0.9570 \\ 
			$n p_0 = 1500$ & (0.7966) & (0.7761) & (0.7646) & (0.8544) & (0.8266) & (0.8133) \\ 
			$n=4000$ & 0.9150 & 0.9150 & 0.9150 & 0.9270 & 0.9220 & 0.9330 \\ 
			$n p_0 = 3000$ & (0.5627) & (0.5470) & (0.5381) & (0.6103) & (0.5880) & (0.5784) \\ 
			\hline
		\end{tabular}
	\end{adjustbox}
	
	{ Notes: The coverage rates and average lengths of the CIs (in parentheses)
		are reported. }
\end{table}

\subsection{Sensitivity to $sp$}
\label{sec:simsp}
In this section, we examine the sensitivity of the choice of $sp$ in  DGPs(1,1) and (2,1). The first and second main columns in  Tables
\ref{spdgp1} and \ref{spdgp2} report the results for $sp= 6$ and $sp=7$, respectively. 
Note we set $k_0=$number of outliers $+2$ and  $k_L'=\min\{\lceil0.10n \hat{p}\rceil, 35\}, k_L=\min\{\lceil0.10n \hat{p}\rceil, 40\}$, and $k_L^{''}=\min\{\lceil0.10n \hat{p}\rceil, 45\}$.
\begin{table}[H]
	\caption{Robustness check of $sp$ for DGP(1,1)}
	\label{spdgp1}
	\centering{}\centering { } \begin{adjustbox}{max
			width=\textwidth} %
		\begin{tabular}{r|ccc|ccc}
			\hline \hline 
			\multicolumn{7}{c}{Panel A: $x=1.5$}\\
			\hline
			& \multicolumn{3}{c|}{Quasi-Bayesian with $sp=6$} & \multicolumn{3}{c}{Quasi-Bayesian with $sp=7$}\\
			\hline
			& $k_L'$  & $k_L$  & $k_L^{''}$  & $k_L'$  & $k_L$  & $k_L^{''}$ \\
			\hline
			$n=500$ & 0.9760 & 0.9800 & 0.9780 & 0.9760 & 0.9800 & 0.9780 \\ 
			$n p_0 = 125$ & (0.4537) & (0.4552) & (0.4559) & (0.4537) & (0.4552) & (0.4559) \\ 
			$n=1000$ & 0.9810 & 0.9750 & 0.9780 & 0.9810 & 0.9750 & 0.9780 \\ 
			$n p_0 = 250$ & (0.2964) & (0.2968) & (0.2964) & (0.2964) & (0.2968) & (0.2964) \\ 
			$n=2000$ & 0.9650 & 0.9650 & 0.9620 & 0.9650 & 0.9650 & 0.9620 \\ 
			$n p_0 = 500$ & (0.2124) & (0.2115) & (0.2106) & (0.2124) & (0.2115) & (0.2106) \\ 
			$n=4000$ & 0.9590 & 0.9620 & 0.9630 & 0.9590 & 0.9620 & 0.9630 \\ 
			$n p_0 = 1000$ & (0.1475) & (0.1453) & (0.1452) & (0.1475) & (0.1453) & (0.1452) \\
			\hline \hline 
			\multicolumn{7}{c}{Panel B: $x=3.0$}\\
			\hline
			& \multicolumn{3}{c|}{Quasi-Bayesian with $sp=6$} & \multicolumn{3}{c}{Quasi-Bayesian with $sp=7$}\\
			\hline
			& $k_L'$  & $k_L$  & $k_L^{''}$  & $k_L'$  & $k_L$  & $k_L^{''}$ \\
			\hline
			$n=500$ & 0.9140 & 0.9140 & 0.9160 & 0.9140 & 0.9140 & 0.9160 \\ 
			$n p_0 = 250$ & (0.5182) & (0.5130) & (0.5138) & (0.5182) & (0.5130) & (0.5138) \\ 
			$n=1000$ & 0.9600 & 0.9520 & 0.9500 & 0.9600 & 0.9520 & 0.9500 \\ 
			$n p_0 = 500$ & (0.3608) & (0.3581) & (0.3551) & (0.3608) & (0.3581) & (0.3551) \\ 
			$n=2000$ & 0.9550 & 0.9650 & 0.9610 & 0.9550 & 0.9650 & 0.9610 \\ 
			$n p_0 = 1000$ & (0.2522) & (0.2479) & (0.2458) & (0.2522) & (0.2479) & (0.2458) \\ 
			$n=4000$ & 0.9480 & 0.9480 & 0.9450 & 0.9480 & 0.9480 & 0.9450 \\ 
			$n p_0 = 2000$ & (0.1708) & (0.1684) & (0.1669) & (0.1708) & (0.1684) & (0.1669) \\ 
			\hline \hline 
			\multicolumn{7}{c}{Panel C: $x=4.5$}\\
			\hline
			& \multicolumn{3}{c|}{Quasi-Bayesian with $sp=6$} & \multicolumn{3}{c}{Quasi-Bayesian with $sp=7$}\\
			\hline
			& $k_L'$  & $k_L$  & $k_L^{''}$  & $k_L'$  & $k_L$  & $k_L^{''}$ \\
			\hline
			$n=500$ & 0.9550 & 0.9500 & 0.9440 & 0.9550 & 0.9500 & 0.9440 \\ 
			$n p_0 = 375$ & (0.5176) & (0.5125) & (0.5114) & (0.5176) & (0.5125) & (0.5114) \\ 
			$n=1000$ & 0.9400 & 0.9430 & 0.9330 & 0.9400 & 0.9430 & 0.9330 \\ 
			$n p_0 = 1000$ & (0.3735) & (0.3656) & (0.3618) & (0.3735) & (0.3656) & (0.3618) \\ 
			$n=2000$ & 0.9370 & 0.9310 & 0.9270 & 0.9370 & 0.9310 & 0.9270 \\ 
			$n p_0 = 1500$ & (0.2625) & (0.2575) & (0.2551) & (0.2625) & (0.2575) & (0.2551) \\ 
			$n=4000$ & 0.9440 & 0.9450 & 0.9380 & 0.9440 & 0.9450 & 0.9380 \\ 
			$n p_0 = 3000$ & (0.1808) & (0.1769) & (0.1747) & (0.1808) & (0.1769) & (0.1747) \\ 
			\hline
	\end{tabular}\end{adjustbox}
	
	{ Notes: The coverage rates and average lengths of the CIs (in parentheses)
		are reported. }
\end{table}

\begin{table}[H]
	\caption{Robustness check of $sp$ for DGP(2,1)}
	\label{spdgp2}\centering{}\centering { } \begin{adjustbox}{max
			width=\textwidth} %
		\begin{tabular}{r|ccc|ccc}
			\hline
			\hline
			\multicolumn{7}{c}{Panel A: $x=1.5$}\\
			\hline
			& \multicolumn{3}{c|}{Quasi-Bayesian with $sp=6$} & \multicolumn{3}{c}{Quasi-Bayesian with $sp=7$}\\
			\hline
			& $k_L'$  & $k_L$  & $k_L^{''}$  & $k_L'$  & $k_L$  & $k_L^{''}$ \\
			\hline
			$n=500$ & 0.9750 & 0.9720 & 0.9660 & 0.9750 & 0.9720 & 0.9660 \\ 
			$n p_0 = 125$ & (0.6902) & (0.6913) & (0.6889) & (0.6902) & (0.6913) & (0.6889) \\ 
			$n=1000$ & 0.9770 & 0.9790 & 0.9740 & 0.9770 & 0.9790 & 0.9740 \\ 
			$n p_0 = 250$ & (0.460) & (0.4586) & (0.4563) & (0.460) & (0.4586) & (0.4563) \\ 
			$n=2000$ & 0.9670 & 0.9710 & 0.9630 & 0.9670 & 0.9710 & 0.9630 \\ 
			$n p_0 = 500$ & (0.3442) & (0.3383) & (0.3368) & (0.3442) & (0.3383) & (0.3368) \\ 
			$n=4000$ & 0.9590 & 0.9610 & 0.9580 & 0.9590 & 0.9610 & 0.9580 \\ 
			$n p_0 = 1000$ & (0.2476) & (0.2448) & (0.2425) & (0.2476) & (0.2448) & (0.2425) \\ 
			\hline \hline 
			\multicolumn{7}{c}{Panel B: $x=3.0$}\\
			\hline
			& \multicolumn{3}{c|}{Quasi-Bayesian with $sp=6$} & \multicolumn{3}{c}{Quasi-Bayesian with $sp=7$}\\
			\hline
			& $k_L'$  & $k_L$  & $k_L^{''}$  & $k_L'$  & $k_L$  & $k_L^{''}$ \\
			\hline
			$n=500$ & 0.9680 & 0.9650 & 0.9640 & 0.9680 & 0.9650 & 0.9640 \\ 
			$n p_0 = 250$ & (1.1332) & (1.1175) & (1.1199) & (1.1332) & (1.1175) & (1.1199) \\ 
			$n=1000$ & 0.9770 & 0.9700 & 0.9740 & 0.9770 & 0.9700 & 0.9740 \\ 
			$n p_0 = 500$ & (0.8101) & (0.7948) & (0.7862) & (0.8101) & (0.7948) & (0.7862) \\ 
			$n=2000$ & 0.9610 & 0.9530 & 0.9570 & 0.9610 & 0.9530 & 0.9570 \\ 
			$n p_0 = 1000$ & (0.5774) & (0.5678) & (0.5616) & (0.5774) & (0.5678) & (0.5616) \\ 
			$n=4000$ & 0.9000 & 0.9030 & 0.9060 & 0.9000 & 0.9030 & 0.9060 \\ 
			$n p_0 = 2000$ & (0.4197) & (0.4134) & (0.4072) & (0.4197) & (0.4134) & (0.4072) \\ 
			\hline \hline 
			\multicolumn{7}{c}{Panel C: $x=4.5$}\\
			\hline
			& \multicolumn{3}{c|}{Quasi-Bayesian with $sp=6$} & \multicolumn{3}{c}{Quasi-Bayesian with $sp=7$}\\
			\hline
			& $k_L'$  & $k_L$  & $k_L^{''}$  & $k_L'$  & $k_L$  & $k_L^{''}$ \\
			\hline
			$n=500$ & 0.9470 & 0.9480 & 0.9400 & 0.9470 & 0.9480 & 0.9400 \\ 
			$n p_0 = 375$ & (1.4608) & (1.4285) & (1.3940) & (1.4608) & (1.4285) & (1.3940) \\ 
			$n=1000$ & 0.9550 & 0.9530 & 0.9520 & 0.9550 & 0.9530 & 0.9520 \\ 
			$n p_0 = 1000$ & (1.0613) & (1.0457) & (1.0291) & (1.0613) & (1.0457) & (1.0291) \\ 
			$n=2000$ & 0.9470 & 0.9510 & 0.9570 & 0.9470 & 0.9510 & 0.9570 \\ 
			$n p_0 = 1500$ & (0.7644) & (0.7449) & (0.7316) & (0.7644) & (0.7449) & (0.7316) \\ 
			$n=4000$ & 0.9470 & 0.9550 & 0.9610 & 0.9470 & 0.9550 & 0.9610 \\ 
			$n p_0 = 3000$ & (0.5162) & (0.5041) & (0.5006) & (0.5162) & (0.5041) & (0.5006) \\ 
			\hline
	\end{tabular}\end{adjustbox}
	
	{ Notes: The coverage rates and average lengths of the CIs (in parentheses)
		are reported.}
\end{table}

\subsection{Additional Simulation Results for Methods 
	\textquotedblleft Mom", \textquotedblleft Momt\_pick", and \textquotedblleft
	Pwm"}
\label{sec:addsim2}

In this section, we report additional results for
\textquotedblleft Mom", \textquotedblleft Momt\_pick", and \textquotedblleft
Pwm" with different EV index estimators in Tables \ref{tablenpbr1}--\ref{tablenpbr15}. Specifically, in those
tables, \textquotedblleft Pickands" and \textquotedblleft Built-in" mean the
EV index is computed by \textbf{rho\_momt\_pick} with the argument
\textbf{method = \textquotedblleft pickands\textquotedblright} and each
estimation method's built-in function as illustrated in Section
\ref{sec:compute-estimators}, respectively. The \textquotedblleft True Value"
method indicates that we simply use the infeasible true value of the EV index
for inference. The results in Tables \ref{tablenpbr1}--\ref{tablenpbr15} show that the three existing estimators mostly either under-cover or over-cover quite a bit even when
the true index is used.


\begin{table}[H]
	\centering{}\centering {\caption{DGP(1,1)}
		\label{tablenpbr1} }\begin{adjustbox}{max width=\textwidth} %
\begin{tabular}{r|ccc|ccc|ccc}
			\multicolumn{10}{c}{}\\
			\hline
			\hline
			\multicolumn{10}{c}{}\\
			\multicolumn{10}{c}{Panel A: $x=1.5$}\\
			\hline
			& \multicolumn{3}{c|}{Pickands} & \multicolumn{3}{c|}{Built-in} & \multicolumn{3}{c}{True}\\
			\hline
			& Mom  & Momt-pick  & Pwm  & Mom  & Momt-pick  & Pwm  & Mom  & Momt-pick  & Pwm \\
			\hline
			$n=500$  & 1.0000  & 0.9680  & 0.9570  & 1.0000  & 0.9680  & 1.0000  & 1.0000  & 1.0000  & 1.0000 \\
			$np_{0}=125$ & (0.8451)  & (1.9062)  & (1.1307)  & (1.0030)  & (1.9062)  & (0.8803)  & (1.0228)  & (2.3633)  & (1.1420) \\
			$n=1000$  & 0.9940  & 0.9730  & 0.9880  & 0.9980  & 0.9730  & 1.0000  & 1.0000  & 0.9990  & 1.0000 \\
\hline
	\end{tabular}\end{adjustbox}
	
	{ Notes: The coverage rates and average lengths of the CIs (in parentheses)
		are reported.}
\end{table}

\begin{table}[H]
	\centering{}\centering {\caption{DGP(2,1)}
		\label{tablenpbr2} }\begin{adjustbox}{max width=\textwidth} %
		\begin{tabular}{r|ccc|ccc|ccc}
			\multicolumn{10}{c}{}\\
			\hline
			\hline
			\multicolumn{10}{c}{}\\
			\multicolumn{10}{c}{Panel A: $x=1.5$}\\
			\hline
			& \multicolumn{3}{c|}{Pickands} & \multicolumn{3}{c|}{Built-in} & \multicolumn{3}{c}{True}\\
			\hline
			& Mom  & Momt-pick  & Pwm  & Mom  & Momt-pick  & Pwm  & Mom  & Momt-pick  & Pwm \\
			\hline
			$n=500$  & 1.0000  & 0.9830  & 0.9530  & 1.0000  & 0.9830  & 1.0000  & 1.0000  & 0.9990  & 1.0000 \\
			$np_{0}=125$ & (1.8886)  & (3.8828)  & (3.0487)  & (1.7846)  & (3.8828)  & (2.2567)  & (1.7452)  & (3.5711)  & (2.9555) \\
			$n=1000$  & 0.9990  & 0.9820  & 0.9960  & 0.9990  & 0.9820  & 1.0000  & 1.0000  & 0.9990  & 1.0000 \\
\hline
	\end{tabular}\end{adjustbox}
	
	{ Notes: The coverage rates and average lengths of the CIs (in parentheses)
		are reported.}
\end{table}

\begin{table}[H]
	\centering{}\centering {\caption{DGP(3,1)}
		\label{tablenpbr3} }\begin{adjustbox}{max width=\textwidth} %
		\begin{tabular}{r|ccc|ccc|ccc}
			\multicolumn{10}{c}{}\\
			\hline
			\hline
			\multicolumn{10}{c}{}\\
			\multicolumn{10}{c}{Panel A: $x=1.5$}\\
			\hline
			& \multicolumn{3}{c|}{Pickands} & \multicolumn{3}{c|}{Built-in} & \multicolumn{3}{c}{True}\\
			\hline
			& Mom  & Momt-pick  & Pwm  & Mom  & Momt-pick  & Pwm  & Mom  & Momt-pick  & Pwm \\
			\hline
			$n=500$  & 0.9990  & 0.9810  & 0.9620  & 1.0000  & 0.9810  & 1.0000  & 1.0000  & 0.9980  & 1.0000 \\
			$np_{0}=125$ & (0.6926)  & (1.2594)  & (2.0065)  & (1.0749) & (1.2594)  & (1.4580)  & (0.6077)  & (1.1008)  & (1.9234) \\
			$n=1000$  & 1.0000  & 0.9830  & 1.0000  & 1.0000  & 0.9830  & 1.0000  & 1.0000  & 0.9950  & 1.0000 \\
\hline
	\end{tabular}\end{adjustbox}
	
	{ Notes: The coverage rates and average lengths of the CIs (in parentheses)
		are reported.}
\end{table}

\begin{table}[H]
	\centering{}\centering {\caption{DGP(1,2)}
		\label{tablenpbr4} }\begin{adjustbox}{max width=\textwidth} %
		\begin{tabular}{r|ccc|ccc|ccc}
			\multicolumn{10}{c}{}\\
			\hline
			\hline
			\multicolumn{10}{c}{}\\
			\multicolumn{10}{c}{Panel A: $x=1.5$}\\
			\hline
			& \multicolumn{3}{c|}{Pickands} & \multicolumn{3}{c|}{Built-in} & \multicolumn{3}{c}{True}\\
			\hline
			& Mom  & Momt-pick  & Pwm  & Mom  & Momt-pick  & Pwm  & Mom  & Momt-pick  & Pwm \\
			\hline
			$n=500$  & 0.9980  & 0.9820  & 0.9330  & 1.0000  & 0.9820  & 1.0000  & 1.0000  & 0.9980  & 0.9940 \\
			$np_{0}=125$ & (1.4814)  & (3.9622)  & (1.4388)  & (1.9523)  & (3.9622)  & (1.0401)  & (1.3366)  & (3.6022)  & (1.3374) \\
			$n=1000$  & 0.9970  & 0.9940  & 0.9630  & 0.9950  & 0.9940  & 1.0000  & 1.0000  & 1.0000  & 1.0000 \\
\hline
	\end{tabular}\end{adjustbox}
	
	{ Notes: The coverage rates and average lengths of the CIs (in parentheses)
		are reported.}
\end{table}

\begin{table}[H]
	\centering{}\centering {\caption{DGP(2,2)}
		\label{tablenpbr5} }\begin{adjustbox}{max width=\textwidth} %
		\begin{tabular}{r|ccc|ccc|ccc}
			\multicolumn{10}{c}{}\\
			\hline
			\hline
			\multicolumn{10}{c}{}\\
			\multicolumn{10}{c}{Panel A: $x=1.5$}\\
			\hline
			& \multicolumn{3}{c|}{Pickands} & \multicolumn{3}{c|}{Built-in} & \multicolumn{3}{c}{True}\\
			\hline
			& Mom  & Momt-pick  & Pwm  & Mom  & Momt-pick  & Pwm  & Mom  & Momt-pick  & Pwm \\
			\hline
			$n=500$  & 0.9950  & 0.9610  & 0.9590  & 1.0000  & 0.9610  & 1.0000  & 1.0000  & 0.9940  & 1.0000 \\
			$np_{0}=125$ & (2.2188)  & (5.4216)  & (3.4788)  & (3.4650)  & (5.4216)  & (2.4951)  & (1.9792)  & (4.8058)  & (3.2518) \\
			$n=1000$  & 0.9780  & 0.9690  & 0.9850  & 0.9930  & 0.9690  & 1.0000  & 1.0000  & 0.9940  & 1.0000 \\
\hline
	\end{tabular}\end{adjustbox}
	
	{ Notes: The coverage rates and average lengths of the CIs (in parentheses)
		are reported.}
\end{table}

\begin{table}[H]
	\centering{}\centering {\caption{DGP(3,2)}
		\label{tablenpbr6} }\begin{adjustbox}{max width=\textwidth} %
		\begin{tabular}{r|ccc|ccc|ccc}
			\multicolumn{10}{c}{}\\
			\hline
			\hline
			\multicolumn{10}{c}{}\\
			\multicolumn{10}{c}{Panel A: $x=1.5$}\\
			\hline
			& \multicolumn{3}{c|}{Pickands} & \multicolumn{3}{c|}{Built-in} & \multicolumn{3}{c}{True}\\
			\hline
			& Mom  & Momt-pick  & Pwm  & Mom  & Momt-pick  & Pwm  & Mom  & Momt-pick  & Pwm \\
			\hline
			$n=500$  & 0.9900  & 0.9250  & 0.9890  & 1.0000  & 0.9250  & 1.0000  & 1.0000  & 0.9760  & 1.0000 \\
			$np_{0}=125$ & (0.6860)  & (1.5014)  & (2.1898)  & (1.2009)  & (1.5014)  & (1.5300)  & (0.5864)  & (1.2579)  & (2.0148) \\
			$n=1000$  & 0.9640  & 0.9490  & 1.0000  & 0.9980  & 0.9490  & 1.0000  & 1.0000  & 0.9870  & 1.0000 \\
\hline
	\end{tabular}\end{adjustbox}
	
	{ Notes: The coverage rates and average lengths of the CIs (in parentheses)
		are reported.}
\end{table}

\begin{table}[H]
	\centering{}\centering {\caption{DGP(1,3)}
		\label{tablenpbr7} }\begin{adjustbox}{max width=\textwidth} %
		\begin{tabular}{r|ccc|ccc|ccc}
			\multicolumn{10}{c}{}\\
			\hline
			\hline
			\multicolumn{10}{c}{}\\
			\multicolumn{10}{c}{Panel A: $x=1.5$}\\
			\hline
			& \multicolumn{3}{c|}{Pickands} & \multicolumn{3}{c|}{Built-in} & \multicolumn{3}{c}{True}\\
			\hline
			& Mom  & Momt-pick  & Pwm  & Mom  & Momt-pick  & Pwm  & Mom  & Momt-pick  & Pwm \\
			\hline
			$n=500$  & 1.0000  & 0.9950  & 0.9400  & 1.0000  & 0.9950  & 1.0000  & 1.0000  & 1.0000  & 0.9940 \\
			$np_{0}=125$ & (2.0546)  & (5.0625)  & (1.7819)  & (2.4149)  & (5.0625)  & (1.3876)  & (1.9314)  & (4.8767)  & (1.6855) \\
			$n=1000$  & 0.9960  & 0.9930  & 0.9680  & 0.9960  & 0.9930  & 1.0000  & 1.0000  & 0.9980  & 0.9980 \\
\hline
	\end{tabular}\end{adjustbox}
	
	{ Notes: The coverage rates and average lengths of the CIs (in parentheses)
		are reported.}
\end{table}

\begin{table}[H]
	\centering{}\centering {\caption{DGP(2,3)}
		\label{tablenpbr8} }\begin{adjustbox}{max width=\textwidth} %
		\begin{tabular}{r|ccc|ccc|ccc}
			\multicolumn{10}{c}{}\\
			\hline
			\hline
			\multicolumn{10}{c}{}\\
			\multicolumn{10}{c}{Panel A: $x=1.5$}\\
			\hline
			& \multicolumn{3}{c|}{Pickands} & \multicolumn{3}{c|}{Built-in} & \multicolumn{3}{c}{True}\\
			\hline
			& Mom  & Momt-pick  & Pwm  & Mom  & Momt-pick  & Pwm  & Mom  & Momt-pick  & Pwm \\
			\hline
			$n=500$  & 1.0000  & 0.9850  & 0.9670  & 1.0000  & 0.9850  & 1.0000  & 1.0000  & 0.9980  & 1.0000 \\
			$np_{0}=125$ & (3.0654)  & (6.9954)  & (4.1390)  & (4.4747)  & (6.9954)  & (3.2049)  & (2.7782)  & (6.1927)  & (3.8966) \\
			$n=1000$  & 0.9990  & 0.9910  & 0.9960  & 1.0000  & 0.9910  & 1.0000  & 1.0000  & 0.9980  & 1.0000 \\
\hline
	\end{tabular}\end{adjustbox}
	
	{ Notes: The coverage rates and average lengths of the CIs (in parentheses)
		are reported.}
\end{table}

\begin{table}[H]
	\centering{}\centering {\caption{DGP(3,3)}
		\label{tablenpbr9} }\begin{adjustbox}{max width=\textwidth} %
		\begin{tabular}{r|ccc|ccc|ccc}
			\multicolumn{10}{c}{}\\
			\hline
			\hline
			\multicolumn{10}{c}{}\\
			\multicolumn{10}{c}{Panel A: $x=1.5$}\\
			\hline
			& \multicolumn{3}{c|}{Pickands} & \multicolumn{3}{c|}{Built-in} & \multicolumn{3}{c}{True}\\
			\hline
			& Mom  & Momt-pick  & Pwm  & Mom  & Momt-pick  & Pwm  & Mom  & Momt-pick  & Pwm \\
			\hline
			$n=500$  & 1.0000  & 0.9640  & 0.9970  & 1.0000  & 0.9640  & 1.0000  & 1.0000  & 0.9960  & 1.0000 \\
			$np_{0}=125$ & (0.9546)  & (1.8920)  & (2.5371)  & (1.4896)  & (1.8920)  & (1.9158)  & (0.8056)  & (1.5757)  & (2.3371) \\
			$n=1000$  & 0.9950  & 0.9630  & 1.0000  & 1.0000  & 0.9630  & 1.0000  & 1.0000  & 0.9960  & 1.0000 \\
\hline
	\end{tabular}\end{adjustbox}
	
	{ Notes: The coverage rates and average lengths of the CIs (in parentheses)
		are reported.}
\end{table}

\begin{table}[H]
	\centering{}\centering {\caption{DGP(1,4)}
		\label{tablenpbr10} }\begin{adjustbox}{max width=\textwidth} %
		\begin{tabular}{r|ccc|ccc|ccc}
			\multicolumn{10}{c}{}\\
			\hline
			\hline
			\multicolumn{10}{c}{}\\
			\multicolumn{10}{c}{Panel A: $x=1.5$}\\
			\hline
			& \multicolumn{3}{c|}{Pickands} & \multicolumn{3}{c|}{Built-in} & \multicolumn{3}{c}{True}\\
			\hline
			& Mom  & Momt-pick  & Pwm  & Mom  & Momt-pick  & Pwm  & Mom  & Momt-pick  & Pwm \\
			\hline
			$n=500$  & 0.8670  & 0.8880  & 0.9740  & 0.8720  & 0.8880  & 1.0000  & 0.7260  & 0.7120  & 1.0000 \\
			$np_{0}=125$ & (0.2835)  & (0.5236)  & (0.7557)  & (0.2892)  & (0.5236)  & (0.5891)  & (0.1402)  & (0.2131)  & (0.6310) \\
			$n=1000$  & 0.9070  & 0.8920  & 1.0000  & 0.9140  & 0.8920  & 1.0000  & 0.8730  & 0.6000  & 1.0000 \\
\hline
	\end{tabular}\end{adjustbox}
	
	{ Notes: The coverage rates and average lengths of the CIs (in parentheses)
		are reported.}
\end{table}

\begin{table}[H]
	\centering{}\centering {\caption{DGP(2,4)}
		\label{tablenpbr11} }\begin{adjustbox}{max width=\textwidth} %
		\begin{tabular}{r|ccc|ccc|ccc}
			\multicolumn{10}{c}{}\\
			\hline
			\hline
			\multicolumn{10}{c}{}\\
			\multicolumn{10}{c}{Panel A: $x=1.5$}\\
			\hline
			& \multicolumn{3}{c|}{Pickands} & \multicolumn{3}{c|}{Built-in} & \multicolumn{3}{c}{True}\\
			\hline
			& Mom  & Momt-pick  & Pwm  & Mom  & Momt-pick  & Pwm  & Mom  & Momt-pick  & Pwm \\
			\hline
			$n=500$  & 0.8640  & 0.9380  & 0.9560  & 0.9200  & 0.9380  & 1.0000  & 0.5350  & 0.6920  & 1.0000 \\
			$np_{0}=125$ & (0.8002)  & (1.3381)  & (2.2473)  & (1.0034)  & (1.3381)  & (1.6313)  & (0.2876)  & (0.3861)  & (1.7405) \\
			$n=1000$  & 0.9200  & 0.9560  & 0.9970  & 0.9480  & 0.9560  & 1.0000  & 0.7850  & 0.5760  & 1.0000 \\
\hline
	\end{tabular}\end{adjustbox}
	
	{ Notes: The coverage rates and average lengths of the CIs (in parentheses)
		are reported.}
\end{table}

\begin{table}[H]
	\centering{}\centering {\caption{DGP(3,4)}
		\label{tablenpbr12} }\begin{adjustbox}{max width=\textwidth} %
		\begin{tabular}{r|ccc|ccc|ccc}
			\multicolumn{10}{c}{}\\
			\hline
			\hline
			\multicolumn{10}{c}{}\\
			\multicolumn{10}{c}{Panel A: $x=1.5$}\\
			\hline
			& \multicolumn{3}{c|}{Pickands} & \multicolumn{3}{c|}{Built-in} & \multicolumn{3}{c}{True}\\
			\hline
			& Mom  & Momt-pick  & Pwm  & Mom  & Momt-pick  & Pwm  & Mom  & Momt-pick  & Pwm \\
			\hline
			$n=500$  & 0.8950  & 0.9590  & 0.9620  & 0.8310  & 0.9590  & 1.0000  & 0.3980  & 0.6220  & 1.0000 \\
			$np_{0}=125$ & (0.4583)  & (0.7104)  & (1.6608)  & (0.5132)  & (0.7104)  & (1.1148)  & (0.1174)  & (0.1361)  & (1.1847) \\
			$n=1000$  & 0.9000  & 0.9480  & 1.0000  & 0.8590  & 0.9480  & 1.0000  & 0.6430  & 0.5330  & 1.0000 \\
\hline
	\end{tabular}\end{adjustbox}
	
	{ Notes: The coverage rates and average lengths of the CIs (in parentheses)
		are reported.}
\end{table}

\begin{table}[H]
	\centering{}\centering {\caption{DGP(1,5)}
		\label{tablenpbr13} }\begin{adjustbox}{max width=\textwidth} %
		\begin{tabular}{r|ccc|ccc|ccc}
			\multicolumn{10}{c}{}\\
			\hline
			\hline
			\multicolumn{10}{c}{}\\
			\multicolumn{10}{c}{Panel A: $x=1.5$}\\
			\hline
			& \multicolumn{3}{c|}{Pickands} & \multicolumn{3}{c|}{Built-in} & \multicolumn{3}{c}{True}\\
			\hline
			& Mom  & Momt-pick  & Pwm  & Mom  & Momt-pick  & Pwm  & Mom  & Momt-pick  & Pwm \\
			\hline
			$n=500$  & 1.0000  & 0.9810  & 0.9380  & 1.0000  & 0.9810  & 1.0000  & 1.0000  & 0.9980  & 0.9990 \\
			$np_{0}=125$ & (1.1016)  & (2.6322)  & (1.2612)  & (1.0221)  & (2.6322)  & (0.9387)  & (1.1406)  & (2.7074)  & (1.2140) \\
			$n=1000$  & 0.9870  & 0.9560  & 0.9860  & 0.9930  & 0.9560  & 1.0000  & 1.0000  & 0.9990  & 1.0000 \\
\hline
	\end{tabular}\end{adjustbox}
	
	{ Notes: The coverage rates and average lengths of the CIs (in parentheses)
		are reported.}
\end{table}

\begin{table}[H]
	\centering{}\centering {\caption{DGP(2,5)}
		\label{tablenpbr14} }\begin{adjustbox}{max width=\textwidth} %
		\begin{tabular}{r|ccc|ccc|ccc}
			\multicolumn{10}{c}{}\\
			\hline
			\hline
			\multicolumn{10}{c}{}\\
			\multicolumn{10}{c}{Panel A: $x=1.5$}\\
			\hline
			& \multicolumn{3}{c|}{Pickands} & \multicolumn{3}{c|}{Built-in} & \multicolumn{3}{c}{True}\\
			\hline
			& Mom  & Momt-pick  & Pwm  & Mom  & Momt-pick  & Pwm  & Mom  & Momt-pick  & Pwm \\
			\hline
			$n=500$  & 0.9990  & 0.9790  & 0.9560  & 1.0000  & 0.9790  & 1.0000  & 1.0000  & 0.9970  & 1.0000 \\
			$np_{0}=125$ & (2.1086)  & (4.6819)  & (3.2110)  & (1.9884)  & (4.6819)  & (2.3459)  & (1.8575)  & (4.0159)  & (3.0659) \\
			$n=1000$  & 0.9960  & 0.9690  & 0.9980  & 0.9940  & 0.9690  & 1.0000  & 1.0000  & 0.9980  & 1.0000 \\
\hline
	\end{tabular}\end{adjustbox}
	
	{ Notes: The coverage rates and average lengths of the CIs (in parentheses)
		are reported.}
\end{table}

\begin{table}[H]
	\centering{}\centering {\caption{DGP(3,5)}
		\label{tablenpbr15} }\begin{adjustbox}{max width=\textwidth} %
		\begin{tabular}{r|ccc|ccc|ccc}
			\multicolumn{10}{c}{}\\
			\hline
			\hline
			\multicolumn{10}{c}{}\\
			\multicolumn{10}{c}{Panel A: $x=1.5$}\\
			\hline
			& \multicolumn{3}{c|}{Pickands} & \multicolumn{3}{c|}{Built-in} & \multicolumn{3}{c}{True}\\
			\hline
			& Mom  & Momt-pick  & Pwm  & Mom  & Momt-pick  & Pwm  & Mom  & Momt-pick  & Pwm \\
			\hline
			$n=500$  & 0.9990  & 0.9610  & 0.9740  & 0.9990  & 0.9610  & 1.0000  & 1.0000  & 0.9960  & 1.0000 \\
			$np_{0}=125$ & (0.6908)  & (1.3902)  & (2.0587)  & (1.0690)  & (1.3902)  & (1.4859)  & (0.6110)  & (1.1944)  & (1.9589) \\
			$n=1000$  & 0.9980  & 0.9810  & 0.9990  & 0.9960  & 0.9810  & 1.0000  & 1.0000  & 0.9990  & 1.0000 \\
\hline
	\end{tabular}\end{adjustbox}
	
	{ Notes: The coverage rates and average lengths of the CIs (in parentheses)
		are reported.}
\end{table}

\bibliography{PF}

\end{document}